%% file: Main_Document.tex
\documentclass[twoside,BCOR=12mm]{sty/idcthesis}
\IDCThesisOptions{language=en,onehalfspacing=true,linkcolor=black!50!blue,fixfloatplacement=true}

\usepackage[latin1]{inputenc}
\usepackage{microtype}
\usepackage[titletoc]{appendix}
\usepackage{algorithm,algorithmic,cases,latexsym}
\usepackage{epsfig,graphicx,psfrag}
\usepackage{amsmath}
\usepackage{color}
\usepackage{url}
\usepackage{scrtime}
\graphicspath{{graphics/}}

\input{Macros}
\numberwithin{Def}{chapter}
\numberwithin{equation}{section}
\numberwithin{Prob}{chapter}
\numberwithin{proposition}{chapter}

\setupthesis{Master Thesis}
  {Multi-Objective Power Allocation for Energy Efficient Wireless Information and Power Transfer Systems }
  {Shiyang Leng}
  {LaTeX-Vorlage, Hinweise zu LaTeX, Nomenklatur}
  {\today}

\begin{document}
\begin{spacing}{1}
    \KOMAoptions{cleardoublepage=empty}
    \pagestyle{empty} \pagenumbering{roman} \setcounter{page}{1}
    \input{0_Title}
    \input{Themenstellung}

    \input{decl_en}

    \tableofcontents
\end{spacing}
\addchap{\abstractname}
\input{0_Abstract}
\begin{spacing}{1}
    \addchap{\glossaryname}
    \input{0_Glossary}
\end{spacing}
\chapter{Introduction}
\label{chap:1_Introduction}
\pagenumbering{arabic} \setcounter{page}{1}
\input{1_Introduction}

\chapter{Multi-Objective SWIPT with Separated Receivers}
\label{chap:2_Sepuser}
\input{2_WIPT_sepuser}

\chapter{Power-Efficient SWIPT in Secure Communication Systems}
\label{chap:4_secureWIPT}
\input{4_WIPT_secure}

\chapter{Conclusion}
\label{chap:5_conclusion}
\input{5_conclusion}

\bibliographystyle{IEEEtran}
\bibliography{vicky_leng}

\begin{appendices}
\chapter{Mathematical Preliminaries}
\input{Appendix_A}

\chapter{Calculations}

\input{Appendix_B}
\end{appendices}

%

\end{document}

%% file: Macros.tex
\newcommand{\gmat}[1]{\ensuremath{\boldsymbol{#1}}}       
\newcommand{\est}[1]{{E}\!\left\{#1\right\}}
\newcommand{\abs}[1]{\lvert#1\rvert}
\newcommand{\norm}[1]{\lVert#1\rVert}

\newtheorem{Thm}{Theorem}
\newtheorem{Lem}{Lemma}

\newtheorem{Def}{Definition}

\newtheorem{Prob}{Problem}
\newtheorem{T-Prob}{Transformed Problem}
\newtheorem{proposition}{Proposition}
\newtheorem{Remark}{Remark}

\DeclareMathOperator{\Tr}{Tr}
\DeclareMathOperator{\Rank}{Rank}
\DeclareMathOperator{\maxo}{maximize}
\DeclareMathOperator{\mino}{minimize}
\DeclareMathOperator{\nullspace}{\mathrm{Null}}

%% file: 0_Title.tex
\newlength{\logoheight}\setlength{\logoheight}{20mm}
\newlength{\logomargin}\setlength{\logomargin}{35mm}

\addcontentsline{toc}{chapter}{\titlename}
\begin{titlepage}
\begin{tikzpicture}[remember picture, overlay]
  \begin{scope}[every node/.style={text badly centered,text width=1.1\textwidth}]
    \path (current page.north)
      +(0mm,-40mm)  node[font=\Large] {\Thesis}
      +(0mm,-60mm)  node[font=\huge\bfseries,minimum height=5cm] {\Title}
      +(0mm,-95mm)  node[font=\large] {\Author}
      +(0mm,-135mm) node[font=\large] {
        \textbf{Lehrstuhl f\"{u}r Digitale \"{U}bertragung}\\
        Prof. Dr.-Ing. Robert Schober\\
        Universit\"{a}t Erlangen-N\"{u}rnberg\\
      }
      +(0mm,-175mm) node {
        \begin{tabular}{ll}
          Supervisor: & Dr. Derrick Wing Kwan Ng\\
                    &  Prof. Dr.-Ing. Robert Schober\\
        \end{tabular}
      }
      +(0mm,-225mm) node {\Date}
    ; 
  \end{scope}
  \node[shift={(+\logomargin,1.5\logoheight)},anchor=west] at (current page.south west){%
    \includegraphics[height=\logoheight]{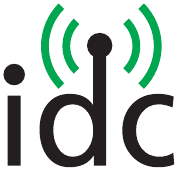}%
  };
  \node[shift={(-\logomargin,1.5\logoheight)},anchor=east] at (current page.south east){%
    \includegraphics[height=\logoheight]{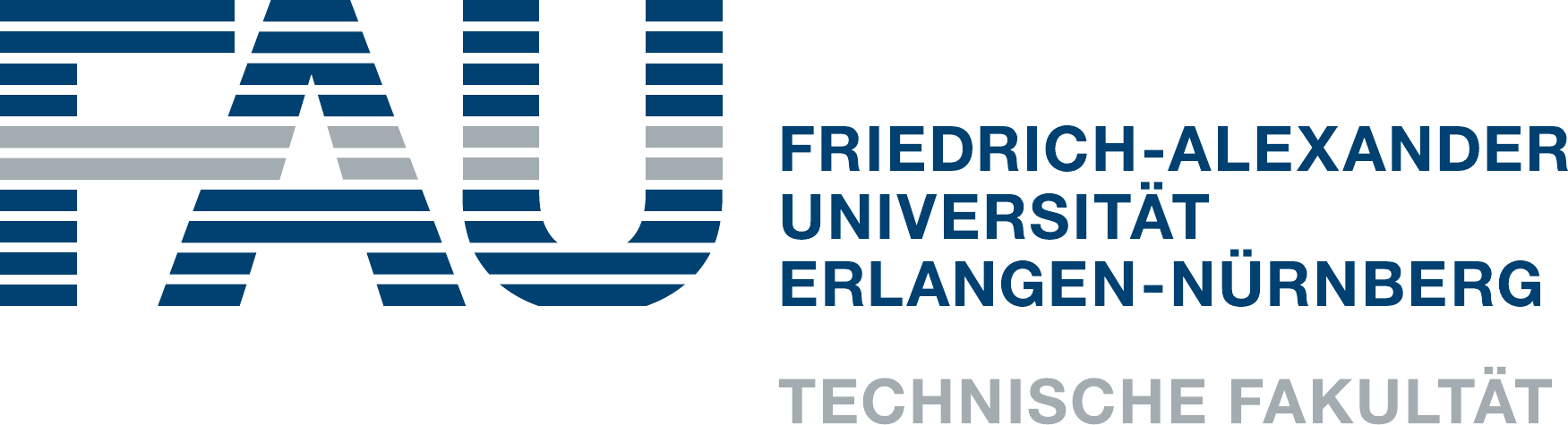}%
  };
\end{tikzpicture}
\end{titlepage}

%% file: Themenstellung.tex
\vspace*{1cm}
\addchap*{Themenstellung}

%% file: decl_en.tex

\vspace*{1cm}

\addchap*{Declaration}
To the best of my knowledge and belief this work was prepared without aid from any other sources except where indicated. Any reference to material previously published by any other person has been duly acknowledged. This work contains no material which has been submitted or accepted for the award of any other degree in any institution.\\[30mm]
Erlangen, \today\\[20mm]
\line(1,0){8}\\[2mm]
\hspace*{25mm}
\parbox[h]{80mm}{%
  \small
  Shiyang Leng\\
  Erlangen}

%% file: 0_Abstract.tex
Simultaneous wireless information and power transfer (SWIPT) provides a promising solution for enabling perpetual wireless networks. As energy efficiency (EE) is an important evaluation of system performance, this thesis studies energy-efficient resource allocation algorithm designs in SWIPT systems. We first investigate the trade-off between the EE for information transmission, the EE for power transfer, and the total transmit power in a basic SWIPT system with separated receivers. A multi-objective optimization problem is formulated under the constraint of maximum transmit power. We propose an algorithm which achieves flexible resource allocation for energy efficiencies maximization and transmit power minimization. The trade-off region of the system design objectives is shown in simulation results. Further, we consider secure communication in a SWIPT system with power splitting receivers. Artificial noise is injected to the communication channel to combat the eavesdropping capability of potential eavesdroppers. A power-efficient resource allocation algorithm is developed when multiple legitimate information receivers and multi-antenna potential eavesdroppers co-exist in the system. Simulation results demonstrate a significant performance gain by the proposed optimal algorithm compared to suboptimal baseline schemes.

%% file: 0_Glossary.tex
\newcommand{\glossaryfirstcolumnlength}{\hspace{6.1em}}
\addsec{\abbreviationsname}
\begin{acronym}[\glossaryfirstcolumnlength]
\acro{AWGN}{Additive White Gaussian Noise}
\acro{CSIT}{Channel State Information at the Transmitter}
\acro{EE}{Energy Efficiency}
\acro{EH}{Energy Harvesting}
\acro{EH-EE}{Energy Harvesting Efficiency}
\acro{IR-EE}{Information Receiving Energy Efficiency}
\acro{MISO}{Multiple-Input Single-Output}
\acro{MIMO}{Multiple-Input Multiple-Output}
\acro{MOO}{Multi-objective Optimization}
\acro{MOOP}{Multi-objective Optimization Problem}
\acro{QoS}{Quality of Service}
\acro{RF}{Radio Frequency}
\acro{SDP}{Semi-definite Programming}
\acro{SIC}{Successive Interference Cancellation}
\acro{SINR}{Signal-to-Interference-plus-Noise Ratio}
\acro{SWIPT}{Simultaneous Wireless Information and Power Transfer}
\acro{TDD}{Time Division Duplexing}
\acro{WPT}{Wireless Power Transfer}

\end{acronym}
\addsec{\operatorsname}
\begin{symbollist}{\glossaryfirstcolumnlength}
  \sym{$(\cdot)^H$}{Hermitian transpose}
  \sym{$(\cdot)^*$}{Complex conjugate}
  \sym{$\mathbf{0}$}{All-zero matrix}
  \sym{$E\{\cdot\}$}{Statistical expectation}
  \sym{$\Tr(\cdot)$}{Trace of a square matrix}
  \sym{$\abs{\cdot}$}{Absolute value}
  \sym{$\norm{\cdot}$}{Euclidean norm}
  \sym{$\det(\cdot)$}{Matrix determinant}
  \sym{Null$(\cdot)$}{Orthonormal null space of a matrix}
  \sym{$[x]^+$}{max\{0,x\}}
  \sym{$\mathbb{C}^{N\times M}$}{The space of all $N\times M$ matrices with complex entries}
  \sym{${\cal CN}(\mathbf{m,\gmat{\Sigma}})$}{A complex Gaussian random  variable vector with mean vector $\mathbf{m}$ and covariance matrix $\gmat{\Sigma}$}

\end{symbollist}
\addsec{\symbolsname}
\begin{symbollist}{\glossaryfirstcolumnlength}
  \sym{$\cal{B}$}{Bandwidth [\si{\Hz}]}
  \sym{$N_{\mathrm{T}}$}{Number of antenna equipped at the transmitter}
  \sym{$N_{\mathrm{R}}$}{Number of antenna equipped at the receiver}
  \sym{$P_{\mathrm{ant}}$}{Antenna power consumption [\si{\joule/\second}]}
  \sym{$P_{\mathrm{c}}$}{Static circuit power consumption [\si{\joule/\second}]}
  \sym{$P_{\mathrm{tot}}$}{Total power consumption [\si{\joule/\second}]}
  \sym{$\xi$}{Power amplifier efficiency}
  \sym{$\eta$}{Energy conversion efficiency}
\end{symbollist}
\clearpage

%% file: 1_Introduction.tex
Wireless communication networks has been rapidly developing over decades. High speed and ubiquitous service is consistently expected in the evolution. This leads to tremendous energy demand for supporting the system operation. However, mobile devices with limited battery supply creates the bottleneck in providing continuous communication services. In particular, the slow improvement of battery capacity has hindered the fulfillment of high quality of service (QoS) requirements. Consequently, energy harvesting (EH) provides a new paradigm that enables self-sustainability for energy constrained wireless devices. Among different EH technologies, a promising one is wireless power transfer \footnote{In this thesis, normalized unit energy is considered, i.e., Joule-per-second, which means the terms ``energy" and ``power" are interchangeable here.} (WPT) where communication terminals harvest energy from radio frequency (RF) signals. Recently, simultaneous wireless information and power transfer (SWIPT) has drawn much attention as an interesting and challenging scenario, since RF signal is a carrier of information and power concurrently.

On the other hand, communication security is a critical issue as wireless communications has become an indispensable media by which people may exchange secret information. As an alternative to traditional cryptographic techniques, physical layer security derives perfectly secure communication by exploiting the physical properties of wireless channel.

In this chapter, we first give a brief overview of SWIPT and physical layer security. Then, we state the motivation of this thesis.

\section{Simultaneous Wireless Information and Power Transfer}
The integration of EH capability into wireless devices is a promising solution for prolonging the lifetime of communication networking. In practice, natural energy sources such as solar, wind, and geothermal are exploited. However, the challenge is that these renewable sources are usually weather and location dependent which may not be suitable for portable mobile terminals. Fortunately, WPT is a promising solution since it allows energy harvester scavenge energy from relative stable and controllable electromagnetic waves in both indoor and outdoor environment. In particular, RF signal as a carrier of energy is an abundant source for WPT. Nowadays, EH circuits are able to harvest microwatt to milliwatt of power over the range of several meters for a transmit power of $1$ Watt and a carrier frequency less than $1$ GHz  \cite{Powercast}. Thus, RF energy can be a viable energy source for devices with low-power consumption, e.g. wireless sensors \cite{Krikidis2014,Ding2014}. Moreover, RF EH provides the possibility for simultaneous wireless information and power transfer (SWIPT) \cite{CN:WIPT_fundamental,CN:Shannon_meets_tesla}. As a carrier of both information and energy, RF signal has a remarkable superiority that unifies information transmission and power transfer, and embrace energy saving by enabling tremendous energy consumed by wireless signals recyclable. As a result, SWIPT has gained recent attention in academic research.

In a SWIPT system, when information signals conveyed from the transmitter to the intended receivers via wireless channels, energy harvesters (the same or other receivers) can harvest energy from the information signals due to the broadcasting nature of wireless channels. Different from conventional wireless communication systems where data rate is the most fundamental system performance metric, the amount of harvested energy has become an equally important QoS requirement in SWIPT systems. Thus, new resource allocation algorithms are needed to fulfill the emerging need \cite{CN:WIPT_fundamental}--\nocite{CN:Shannon_meets_tesla,JR:WIP_receiver,CN:MIMO_WIPT,JR:multiuser_MISO,JR:beamforming_MISO,
Xu2013,CN:multiuser_OFDM_WIPT,CN:Eurosip_SWIPT,JR:WIPT_fullpaper_OFDMA,CN:strategies_twouserMIMO,Morsi2014,CN:Kwan_globecom2014,
CN:Maryna_2015,CN:tao_2015,JR:WIPT_relaying_timeswitching}\cite{JR:WIPT_CR}.
In \cite{CN:WIPT_fundamental}--\nocite{CN:Shannon_meets_tesla,JR:WIP_receiver}\cite{CN:MIMO_WIPT}, the fundamental trade-off between channel capacity and harvested energy was studied for frequency flat fading channel and frequency selective fading channel. Specifically, an ideal receiver managing synchronous information decoding and energy harvesting from the same received signal was assumed in \cite{CN:WIPT_fundamental} and \cite{CN:Shannon_meets_tesla}. However, the signal used for information decoding, cannot be reused for EH due to the limitation of current practical circuits. Subsequently, the author in \cite{JR:WIP_receiver} and \cite{CN:MIMO_WIPT} proposed three different types of receivers, namely, power splitting, separated, and time switching receivers. In particular, a power splitting receiver, cf. Figure \ref{fig:power_splitter}, splits the received signal into two power streams with a power splitting ratio $\rho$, i.e., $\rho$ portion of the received signal remains for information decoding, and the other $1-\rho$ portion is reserved for EH. The power splitting scheme also generalizes the separated receivers scheme, i.e., independent information decoding is performed when $\rho=1$ and independent EH is done when $\rho=0$. A time switching receiver means switching time slots for information decoding and EH successively.
\begin{figure}
\centering
\includegraphics[scale=0.9]{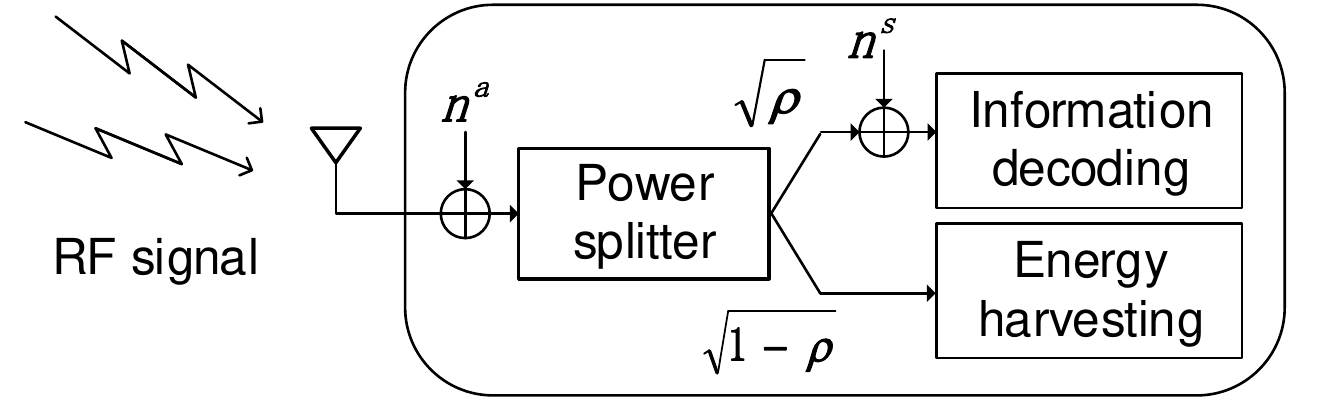}
\caption{Power splitting receiver.}
\label{fig:power_splitter}
\end{figure}

Furthermore, \cite{JR:multiuser_MISO} and \cite{JR:beamforming_MISO} focused on transmit beamforming design in multiple-input single-output (MISO) SWIPT systems for separated and power splitting receivers, respectively. In \cite{Xu2013}, beamformers were optimized for the maximization of sum harvested energy under the minimum required signal-to-interference-plus-noise ratio (SINR) constraints for multiple information receivers. In \cite{CN:multiuser_OFDM_WIPT}, \cite{CN:Eurosip_SWIPT}, and \cite{JR:WIPT_fullpaper_OFDMA}, energy-efficient SWIPT, was studied in multi-carrier systems, where power allocation, user and subcarrier scheduling were considered. It was shown that EE could be improved by implementing SWIPT.  In \cite{CN:strategies_twouserMIMO}, a transmission strategy was proposed for multiuser multiple-input multiple-output (MIMO) SWIPT systems with interference channel. Recently, multiuser scheduling, which exploits multiuser diversity for improving the system performance of multiuser systems was considered for SWIPT in \cite{Morsi2014} and \cite{CN:Maryna_2015}. On the other hand, SWIPT also envisions new opportunities for cooperative communications. In \cite{JR:WIPT_relaying_timeswitching}, the performance of SWIPT systems was analyzed for different relaying protocols. For cognitive radio networks, \cite{JR:WIPT_CR} focused on the cooperation between primary and secondary systems at both the information and energy levels. In these literatures, SWIPT demonstrates significant gains in many aspects, for instance, energy consumption, spectral efficiency, and time delay.

\section{Physical Layer Security}
Nowadays, wireless communication security has become an extremely important issue. Traditionally, communication security relies on cryptographic technologies that are applied in the application layer of wireless networks. These algorithms usually require reliable key distribution and may encounter high computational complexity. As an alternative, physical layer security is a viable solution for ensuring communication security. It guarantees secure communication by exploiting the physical properties of wireless channels. In \cite{Report:Wire_tap}, the basic idea of physical layer security was first proposed, cf. Figure \ref{fig:PHY}. It is verified that confidential messages can be reliably exchanged between a transmitter (Alice) and a receiver (Bob) if the receiver enjoys a better channel than an eavesdropper (Eve). In particular, by exploiting the extra degrees of freedom offered by multiple transmitting antennas, the channel between the transmitter and the eavesdropper can be weakened by injecting a properly designed artificial noise. Thus, perfect communication security is achievable.
\begin{figure}
\centering
\includegraphics[scale=0.85]{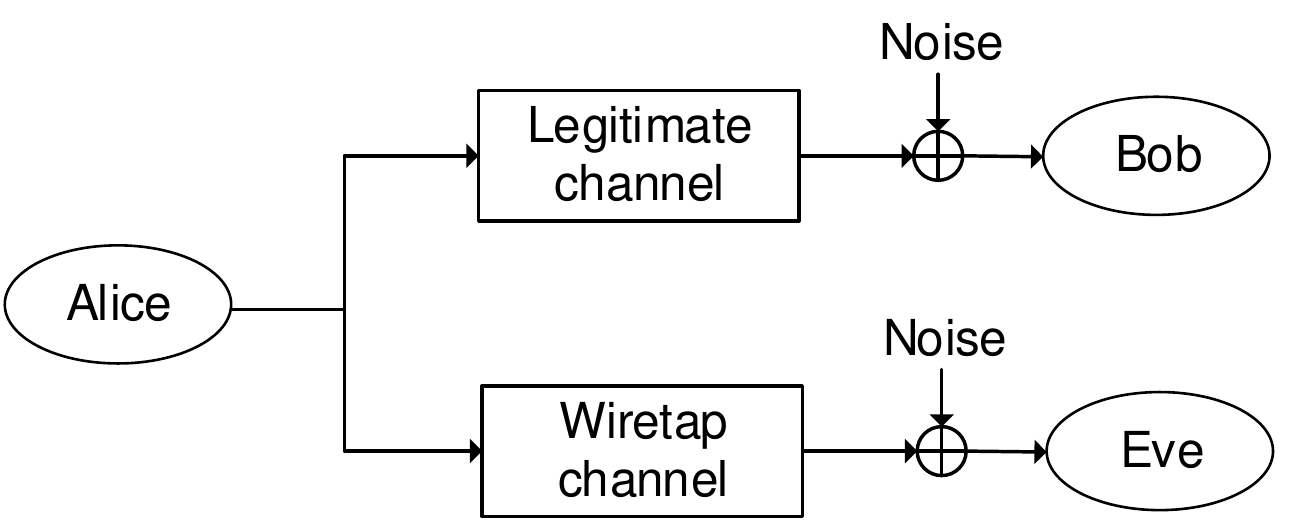}
\caption{A general wiretap model.}
\label{fig:PHY}
\end{figure}

Under this context, secure communication guaranteed by artificial noise generation led to a series of discussion. In \cite{JR:Artifical_Noise1}, a power allocation algorithm was proposed for the maximization of ergodic secrecy capacity. In \cite{JR:AN_MISO_secrecy}, the author maximized secrecy capacity in a system with multiple single-antenna eavesdroppers. \cite{JR:eff_secure_ofdma_kwan} studied secure OFDMA systems by energy-efficient resource allocation algorithm design. A trade-off between energy efficiency and secure communication was revealed. Further, \cite{JR:secure_ofdma_DFrelay_kwan} focused on secure resource allocation and scheduling in OFDMA relay networks, where a passive multi-antenna eavesdropper was considered. In \cite{JR:robust_secure_CR_kwan}, the author studied secure communication in cognitive radio networks. A robust resource allocation algorithm was proposed to ensure video communication secrecy in the secondary system.

In fact, a large amount of power is allocated to artificial noise for providing secure communication. This implies that artificial noise can act as a potential energy source for WPT while ensuring secure communication at the same time. Resource allocation algorithm design for secure communication with SWIPT was studied in \cite{JR:secure_WIPT_MISO_ruizh}--\nocite{JR:Kwan_secure_imperfect,CN:mulobj_secure_WIPT,CN:Kwan_globecom2013,CN:kwan_vicky,
CN:Multicast_SWIPT,JR:Kwan_SEC_DAS,CN:PHY_SEC_max_min}\cite{JR:kwan_MOOP_SWIPT}. In \cite{JR:secure_WIPT_MISO_ruizh}, an additive energy signal was adopted to facilitate power transfer and to ensure secure communication for separated information and energy receivers. \cite{JR:Kwan_secure_imperfect}, \cite{CN:Kwan_globecom2013}, and \cite{CN:Multicast_SWIPT} studied power-efficient resource allocation algorithm for secure SWIPT systems by minimizing the total transmit power. In \cite{JR:Kwan_secure_imperfect}, robust beamforming was studied in a system with imperfect CSIT, where both artificial noise and energy beams are used in order to provide secure communication and to improve WPT. Besides, in \cite{CN:mulobj_secure_WIPT} multi-objective optimization (MOO) approach was applied to jointly optimize multiple system design objectives in a secure communication system with SWIPT. Especially, EE of energy harvesting was maximized. In \cite{JR:kwan_MOOP_SWIPT}, SWIPT in secure communication system was combined with cognitive radio, where multiple objectives, including total transmit power minimization and EH efficiency maximization were considered in a multi-objective optimization problem (MOOP).

\section{Motivation}
As aforementioned, in SWIPT systems, the system performance is evaluated on both information delivery and power transfer. A trade-off naturally arises when considering resource allocation on both aspects. The trade-off between data rate and the amount of harvested energy has been investigated in recent literatures. However, the trade-off between EE of information delivery and EH has not been considered so far. On the other hand, the conflicting system design objectives, i.e., information receiving EE maximization, energy harvesting efficiency maximization, and transmit power minimization, leads to multiple resource allocation algorithm designs. To provide a resource allocation algorithm that can flexibly achieve multiple objectives, multi-objective system design for SWIPT requires to be studied.

Furthermore, communication security is a serious issue in SWIPT systems. In particular, when the transmit power of the information signal is increased to facilitate SWIPT, the signal becomes more vulnerable to eavesdropping due to a higher potential for information leakage. Thus, communication security arises as a new QoS concern in SWIPT systems. On the other hand, artificial noise can serve as a energy source for WPT and be harvested to extend the lifetime of power-constrained devices. In terms of power-efficient resource allocation in secure SWIPT systems, recent literatures limit the system configuration with a single information receiver and multiple single-antenna eavesdroppers. However, optimal resource allocation for secure communication in SWIPT systems with multiple desired power splitting receivers and multi-antenna potential eavesdroppers remains an unsolved problem.

The rest of the thesis is organized as follows. In Chapter \ref{chap:2_Sepuser}, we study multi-objective SWIPT design for separated receivers. The energy efficiencies for information transmission and EH are investigated by taking into account the maximum transmit power constraint. A resource allocation algorithm considering the trade-off between multiple objectives is proposed. In Chapter \ref{chap:4_secureWIPT}, we focus on SWIPT in secure communication systems with multiple power splitting receivers and multi-antenna eavesdroppers. A power-efficient resource allocation algorithm is proposed. Finally, we summarize the contributions of this thesis in Chapter \ref{chap:5_conclusion}. Appendix A contains basic theories of optimization problem. Appendix B includes the proofs of the theorems and propositions in the thesis.

%% file: 2_WIPT_sepuser.tex
In this chapter, we study multi-objective power allocation for energy-efficient SWIPT with separated receivers. In a MISO system, information signals and energy beams are transmitted simultaneously to jointly support information delivery to a information receiver and energy supply to a energy harvester. Under a maximum transmit power constraint, we focus on three desired system design objectives, namely, information receiving efficiency (IR-EE) maximization, energy harvesting efficiency (EH-EE) maximization, and total transmit power minimization. In particular, we jointly optimize the information beamforming vector and covariance matrix of the energy signal to achieve the considered system objective. The problem is formulated as a non-convex MOOP. To deal with the fractional objective functions, Charnes-Cooper transformation method is adopted. Subsequently, the transformed problem is solved by semi-definite program (SDP) relaxation approach. We prove that the SDP relaxation is tight. In particular, a tractable structure of the optimal solution is verified. Simulation results shows the trade-off between IR-EE, EH-EE, and the total transmit power.
\section{System Model}
\begin{figure}
\centering
\includegraphics[scale=0.8]{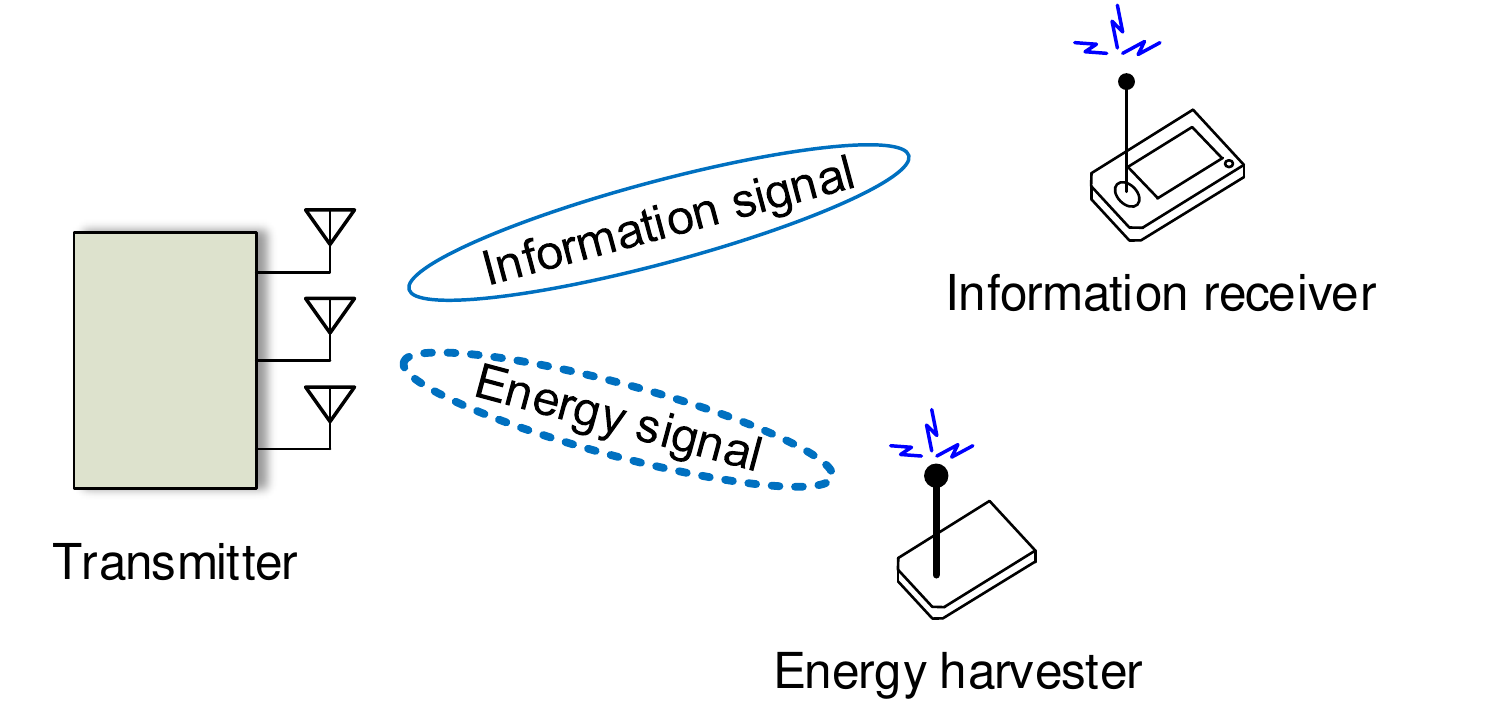}
\caption{A SWIPT system with separated receivers.}
\label{fig:system_model2}
\end{figure}

We focus on a downlink MISO system with SWIPT. The system consists of one multi-antenna transmitter, one single-antenna information receiver, and one single-antenna energy harvester. The transmitter is equipped with $N_\mathrm{T}$ antennas. It sends precoded information signal and energy beams simultaneously to facilitate information transmission and power transfer, cf. Figure \ref{fig:system_model2}. The transmission is divided into time slots. The transmitted signal in each time slot is given by
\begin{eqnarray}
\mathbf{x}=\mathbf{w}_\mathrm{I}s+\mathbf{w}_\mathrm{E},
\end{eqnarray}
where $s\in \mathbb{C}$ is the information-bearing symbol with $\est{{\abs{s}}^2}=1$. $\mathbf{w}_\mathrm{I}\in \mathbb{C}^{{N_\mathrm{T}}\times 1}$ is the corresponding precoded beamforming vector for the information receiver. $\mathbf{w}_\mathrm{E}\in \mathbb{C}^{{N_\mathrm{T}}\times 1}$ is the energy signal beamforming vector facilitating energy transfer to the energy harvester. The energy beamforming vector $\mathbf{w}_\mathrm{E}$ is modeled as a complex Gaussian pseudo-random sequence as $\mathbf{w}_\mathrm{E}\sim{\cal CN}(0,\mathbf{W}_\mathrm{E})$, where $\mathbf{W}_\mathrm{E}=\est{\mathbf{w}_\mathrm{E}\mathbf{w}_{\mathrm{E}}^H}$ is the covariance matrix of the energy signal. Assume $\mathbf{w}_\mathrm{E}$ is generated at the transmitter by a pseudo-random sequence generator with a predefined seed. The seed can be delivered to the information receiver before effective information transmission. Thus, the interference of energy signal can be totally cancelled at the information receiver.

We assume a narrow-band slow fading channel between the transmitter and receivers. The channel is assumed to be perfectly known at the transmitter. Then, the received signals at the information receiver and energy harvester are expressed as
\begin{eqnarray}
y_\mathrm{IR}=\mathbf{h}^H(\mathbf{w}_\mathrm{I}s+\mathbf{w}_\mathrm{E})+n_\mathrm{I},\\
y_\mathrm{EH}=\mathbf{g}^H(\mathbf{w}_\mathrm{I}s+\mathbf{w}_\mathrm{E})+n_\mathrm{E},
\end{eqnarray}
where $\mathbf{h}\in \mathbb{C}^{N_{\mathrm{T}}\times 1}$ is the channel vector between the transmitter and the information receiver, and $\mathbf{g}\in \mathbb{C}^{N_{\mathrm{T}}\times 1}$ is the channel vector between the transmitter and the energy harvester. They capture the joint effect of multipath fading and path loss. $n_\mathrm{I}\in \mathbb{C}$ and $n_\mathrm{E}\in \mathbb{C}$ are additive white Gaussian noise (AWGN) at the information receiver and energy harvester, respectively, which are distributed as ${\cal CN}(0,\sigma_\mathrm{I}^2)$ and ${\cal CN}(0,\sigma_\mathrm{E}^2)$.

Information receiver focuses on decoding the information signal. The achievable rate (bit/s/Hz) at the information receiver can be described as
\begin{eqnarray}
R=\log_2\Big(1+\frac{1}{\sigma_\mathrm{I}^2}\abs{\mathbf{h}^H\mathbf{w}_\mathrm{I}}^2\Big)=\log_2\Big(1+\frac{1}{\sigma_\mathrm{I}^2}\mathbf{w}_{\mathrm{I}}^H\mathbf{H}\mathbf{w}_\mathrm{I}\Big),
\end{eqnarray}
where $\mathbf{H}=\mathbf{h}\mathbf{h}^H$. At the same time, both the information signal and the energy signal can act as RF energy source for the energy harvester due to the broadcast nature of wireless channels. According to the law of energy conservation, the harvested energy is proportional to the received signal power. The total harvested energy at the energy harvester is given by
\begin{eqnarray}
P_\mathrm{harv}(\mathbf{w}_\mathrm{I},\mathbf{W}_\mathrm{E})=\eta\Big(
\abs{\mathbf{g}^H\mathbf{w}_\mathrm{I}}^2+\abs{\mathbf{g}^H\mathbf{w}_\mathrm{E}}^2\Big)
=\eta\Big(\mathbf{w}^H_\mathrm{I}\mathbf{G}\mathbf{w}_\mathrm{I}+\Tr(\mathbf{G}\mathbf{W}_\mathrm{E})\Big),
\end{eqnarray}
with $\mathbf{G}=\mathbf{g}\mathbf{g}^H$. $\eta$ is the energy conversion efficiency, which is a constant with $0\leq\eta\leq1$. It implies a energy loss in the process of converting the received RF energy to electrical energy for storage. We ignore the thermal noise at the receiving antenna as it is relative small compared to the received signal power.

Apart from system throughput, EE is also a fundamental system performance metric in modern communication networks. EE is generally defined as the ratio between system throughput and total power consumption. We first model the total power consumption (Joule-per-second) by taking into account the transmit power consumption and additional hardware power dissipation at the transmitter which can be described as
\begin{eqnarray}\label{eqn:Ptotal}
P_\mathrm{tot}(\mathbf{w}_\mathrm{I},\mathbf{W}_\mathrm{E})=\frac{\norm{\mathbf{w}_\mathrm{I}}^2+\Tr(\mathbf{W}_\mathrm{E})}{{\xi}}+N_\mathrm{T}P_\mathrm{ant}+P_\mathrm{c}.
\end{eqnarray}
\noindent $\xi$ is the power amplifier efficiency, which is a constant with $0\leq\xi\leq1$. The first term in (\ref{eqn:Ptotal}) is the total power consumption in the power amplifier. ${N_{\mathrm{T}}P_{\mathrm{ant}}}$ accounts for the dynamic circuit power consumption proportional to the number of transmitting antenna. $P_\mathrm{ant}$ denotes the power dissipation at the transmitting antenna, including the transmit filter, mixer, frequency synthesizer, digital-to-analog converter (DAC), etc. $P_{\mathrm{c}}$ denotes the fixed circuit power consumption due to baseband signal processing.

Based on the general concept of efficiency, we define IR-EE and EH-EE as
\begin{eqnarray}
\Phi_\mathrm{IR}(\mathbf{w}_\mathrm{I},\mathbf{W}_\mathrm{E})&=&\frac{R}{P_\mathrm{tot}}\,\,\,\,=\frac{\log_2(1+\frac{1}{\sigma_\mathrm{I}^2}\mathbf{w}^H_\mathrm{I}\mathbf{H}\mathbf{w}_\mathrm{I})}{(\norm{\mathbf{w}_\mathrm{I}}^2+\Tr(\mathbf{W}_\mathrm{E}))/{\xi}+N_\mathrm{T}P_\mathrm{ant}+P_\mathrm{c}}\\ \notag\\
\mathrm{and}\quad\Phi_\mathrm{EH}(\mathbf{w}_\mathrm{I},\mathbf{W}_\mathrm{E})&=&\frac{P_\mathrm{harv}}{P_\mathrm{tot}}=\frac{\eta(\mathbf{w}^H_\mathrm{I}\mathbf{G}\mathbf{w}_\mathrm{I}+\Tr(\mathbf{G}\mathbf{W}_\mathrm{E}))}{(\norm{\mathbf{w}_\mathrm{I}}^2+\Tr(\mathbf{W}_\mathrm{E}))/{\xi}+N_\mathrm{T}P_\mathrm{ant}+P_\mathrm{c}},
\end{eqnarray}
respectively.

\section{Problem Formulation}
In SWIPT system, IR-EE maximization, EH-EE maximization, and total transmit power minimization are all desirable for system design. In this section, we first propose three problem formulations for single-objective system design for SWIPT. Each single-objective problem describes one important aspect of the system design. Then, we consider the three system design objectives jointly by MOO.

The first system design objective is the maximization of IR-EE. The optimization problem is formulated as
\begin{Prob}IR-EE Maximization:\label{prob:WIPT_IR-EE}
\begin{eqnarray}
\underset{\mathbf{W}_{\mathrm{E}}\in\mathbb{H}^{N_\mathrm{T}},\mathbf{w}_{\mathrm{I}}}{\maxo}\,\, &&\Phi_{\mathrm{IR}}(\mathbf{w}_\mathrm{I},\mathbf{W}_\mathrm{E})\notag\\
\mathrm{subject\,\,to}\,\, &&\mathrm{C1}:\,\,\norm{\mathbf{w}_\mathrm{I}}^2+\Tr(\mathbf{W}_{\mathrm{E}})\leq P_{\mathrm{max}},\notag\\
&&\mathrm{C2}:\,\,\mathbf{W}_{\mathrm{E}}\succeq \mathbf{0}.\notag
\end{eqnarray}
\end{Prob}

The second system design objective is the maximization of EH-EE. The problem formulation is given as
\begin{Prob}EH-EE Maximization:\label{prob:WIPT_EH-EE}
\begin{eqnarray}
\underset{\mathbf{W}_{\mathrm{E}}\in\mathbb{H}^{N_\mathrm{T}},\mathbf{w}_{\mathrm{I}}}{\maxo}\,\, &&\Phi_{\mathrm{EH}}(\mathbf{w}_\mathrm{I},\mathbf{W}_\mathrm{E})\notag\\
\mathrm{subject\,\,to}\,\, &&\mathrm{C1,\,C2}.\notag
\end{eqnarray}
\end{Prob}

The third system design objective is the minimization of the total transmit power at the transmitter. The problem formulation is proposed as
\begin{Prob}Total Transmit Power Minimization:\label{prob:WIPT_Ptrans}
\begin{eqnarray}
\underset{\mathbf{W}_{\mathrm{E}}\in\mathbb{H}^{N_\mathrm{T}},\mathbf{w}_{\mathrm{I}}}{\mino}\,\, &&\norm{\mathbf{w}_\mathrm{I}}^2+\Tr(\mathbf{W}_{\mathrm{E}})\notag\\
\mathrm{subject\,\,to}\,\, &&\mathrm{C1,\,C2}.\notag
\end{eqnarray}
\end{Prob}

As the above problem formulations stated, IR-EE, EH-EE, and the total transmit power are independently optimized respectively. In each single-objective problem, information beamforming vector, $\mathbf{w}_\mathrm{I}$, and the covariance matrix of the energy signal, $\mathbf{W}_{\mathrm{E}}$, are jointly designed by considering the maximum transmit power constraint C1. In addition, covariance matrix $\mathbf{W}_{\mathrm{E}}$ should be a positive semi-definite Hermitian matrix as indicated in constraint C2.

For the sake of notational simplicity, we denote the objective functions in the above problems as $F_j(\mathbf{w}_{\mathrm{I}},\mathbf{W}_{\mathrm{E}})$, $j=1,2,3$. We note that Problem \ref{prob:WIPT_Ptrans} is a trivial problem with optimal value zero since the transmitter does not need to provide any QoS to the receiver. Yet, Problem \ref{prob:WIPT_Ptrans} plays an important role in the following when we study the multi-objective power allocation algorithm design. Without loss of generality, Problem \ref{prob:WIPT_Ptrans} can be rewritten as an equivalent maximization problem in order to represent the three problems consistently. The corresponding objective function is written as $F_3(\mathbf{w}_{\mathrm{I}},\mathbf{W}_{\mathrm{E}})=-(\norm{\mathbf{w}_\mathrm{I}}^2+\Tr(\mathbf{W}_{\mathrm{E}}))$.

In practice, these three independent optimization objectives are all desirable from the system operator perspective. However, there are non-trivial trade-off between them. In order to optimize these conflicting system design objectives systematically and simultaneously, we apply MOO, cf. Appendix \ref{app:intro_MOOP}.

A common approach to formulate a MOOP is via "\emph{a prior method}". This method allows the system designer to indicate the relative importance of the system design objectives before running the optimization algorithm. In particular, a sequence of scalars, which is known as "preference parameters" or "weights", are a prior specified to scalarize system designer's preference on different objectives. There are many scalarization methods. Here we adopt
weighted min-max method \cite{JR:MOOP}. As introduced in Appendix \ref{app:intro_MOOP}, the optimal point(s) of a MOOP is defined by Pareto optimality. All Pareto optimal points, which form the Pareto optimal set, are important to the system designer. In fact, weighted min-max method can provide the complete Pareto optimal set by varying the preference parameters, despite the non-convexity of the MOOP. Based on weighted min-max method, we incorporate three system design objectives into a MOOP, which is formulated as
\begin{Prob}Multi-Objective Optimization Problem:\label{prob:multiobj_WIPTsepuser}
\begin{eqnarray}
\underset{\mathbf{W}_{\mathrm{E}}\in\mathbb{H}^{N_\mathrm{T}},\mathbf{w}_{\mathrm{I}}}{\mino}\,\,&&\max_{j=1,2,3}\,\, \Big\{\omega_j(F_j^*-F_j(\mathbf{w}_{\mathrm{I}},\mathbf{W}_{\mathrm{E}}))\Big\}\notag\\
\mathrm{subject\,\,to} &&\mathrm{C1,\,C2},\notag
\end{eqnarray}
\end{Prob}
\noindent where $F_j^*$ is the optimal objective value with respect to Problem $j$. $\omega_j$ is the weight imposed on objective function $j$ subject to $0\leq\omega_i\leq1$ and $\sum_j\omega_j=1$, which indicates the system designer's preference on $j$th objective function over the others. In extreme case, when $\omega_j=1$ and $\omega_i=0, \forall i\neq j $, Problem \ref{prob:multiobj_WIPTsepuser} is equivalent to the single-objective optimization problem $j$.

In this MOOP, we investigate the complete trade-off region between the three objectives regarding to system power allocation. To this end, we only take maximum transmit power constraint into consideration. In case other QoS constraints are imposed into the MOOP, a smaller Pareto optimal set can be obtained, which is actually a subset of the complete trade-off region.

\section{Multi-Objective Power Allocation Algorithm Design}
It can be observed that Problem \ref{prob:WIPT_IR-EE} and Problem \ref{prob:WIPT_EH-EE} are non-convex due to the fractional form of the objectives which leads to the non-convexity in Problem \ref{prob:multiobj_WIPTsepuser}. In order to obtain a tractable solution, we first transform the non-convex objective functions by Charnes-Cooper transformation method. Then, the transformed problems are solved by SDP relaxation approach.

We first reformulate aforementioned three single-objective optimization problems by defining a set of new optimization variables as follows:
\begin{eqnarray}\label{eqn:newvariabledefine}
\mathbf{W}_\mathrm{I}=\mathbf{w}_\mathrm{I}\mathbf{w}_\mathrm{I}^H,\,\,\theta=\frac{1}{P_\mathrm{tot}(\mathbf{w}_\mathrm{I},\mathbf{W}_\mathrm{E})},\,\,\overline{\mathbf{W}}_\mathrm{I}=\theta\mathbf{W}_\mathrm{I},\,\,\mathrm{and}\,\,\overline{\mathbf{W}}_\mathrm{E}=\theta\mathbf{W}_\mathrm{E}.
\end{eqnarray}
Then the original problems can be rewritten with respect to the new optimization variables $\{\overline{\mathbf{W}}_\mathrm{I},\overline{\mathbf{W}}_\mathrm{E}, \theta\}$, which are given by
\begin{Prob}Transformed IR-EE Maximization Problem:\label{prob:WIPT_IR-EE_reform}
\begin{eqnarray}
\underset{\overline{\mathbf{W}}_\mathrm{I},\overline{\mathbf{W}}_\mathrm{E}\in\mathbb{H}^{N_\mathrm{T}},\theta}{\maxo}\,\, &&\theta\log_2(1+\frac{\Tr(\mathbf{H}\overline{\mathbf{W}}_\mathrm{I})}{\theta\sigma_\mathrm{I}^2})\notag\\
\mathrm{subject\,\,to}\,\, &&\overline{\mathrm{C1}}:\,\,\Tr(\overline{\mathbf{W}}_\mathrm{I}+\overline{\mathbf{W}}_\mathrm{E})\leq\theta P_{\mathrm{max}},\notag\\
&&\overline{\mathrm{C2}}:\,\,\overline{\mathbf{W}}_\mathrm{I}\succeq \mathbf{0},\,\,\overline{\mathbf{W}}_\mathrm{E}\succeq \mathbf{0},\notag\\
&&\overline{\mathrm{C3}}:\,\,\Rank(\overline{\mathbf{W}}_\mathrm{I})\leq1,\notag\\
&&\overline{\mathrm{C4}}:\,\,\frac{\Tr(\overline{\mathbf{W}}_\mathrm{I}+\overline{\mathbf{W}}_\mathrm{E})}{\xi}+\theta(N_\mathrm{T}P_\mathrm{ant}+P_\mathrm{c})\leq1,\notag\\
&&\overline{\mathrm{C5}}:\,\,\theta\ge0. \notag
\end{eqnarray}
\end{Prob}
\begin{Prob}Transformed EH-EE Maximization Problem:\label{prob:WIPT_EH-EE_reform}
\begin{eqnarray}
\underset{\overline{\mathbf{W}}_\mathrm{I},\overline{\mathbf{W}}_\mathrm{E}\in\mathbb{H}^{N_\mathrm{T}},\theta}{\maxo}\,\, &&\eta\Tr(\mathbf{G}(\overline{\mathbf{W}}_\mathrm{I}+\overline{\mathbf{W}}_\mathrm{E}))\notag\\
\mathrm{subject\,\,to}\,\, &&\overline{\mathrm{C1}} - \overline{\mathrm{C5}}.\notag
\end{eqnarray}
\end{Prob}
\begin{Prob}Transformed Total Transmit Power Minimization Problem:\label{prob:WIPT_Ptrans_reform}
\begin{eqnarray}
\underset{\overline{\mathbf{W}}_\mathrm{I},\overline{\mathbf{W}}_\mathrm{E}\in\mathbb{H}^{N_\mathrm{T}},\theta}{\maxo}\,\, &&-\xi(\frac{1}{\theta}-N_\mathrm{T}P_\mathrm{ant}-P_\mathrm{c})\notag\\
\mathrm{subject\,\,to}\,\, &&\overline{\mathrm{C1}} - \overline{\mathrm{C5}}.\notag
\end{eqnarray}
\end{Prob}
Denote the transformed objective function as $\overline{F_j}(\overline{\mathbf{W}}_\mathrm{I},\overline{\mathbf{W}}_\mathrm{E},\theta)$, $j=1,2,3$. Constraints $\overline{\mathbf{W}}_\mathrm{I}\succeq \mathbf{0}$, $\overline{\mathbf{W}}_\mathrm{I}\in\mathbb{H}^{N_\mathrm{T}}$, and $\Rank(\overline{\mathbf{W}}_\mathrm{I})\leq1$ are imposed to guarantee that $\overline{\mathbf{W}}_\mathrm{I}=\theta\mathbf{w}_\mathrm{I}\mathbf{w}_\mathrm{I}^H$. Constraints $\overline{\mathrm{C4}}$ and $\overline{\mathrm{C5}}$ are introduced due to the proposed transformation.

Furthermore, in order to simplify the following algorithm design, we first normalize the transformed objective functions due to their different ranges and dimensions. A robust transformation, regardless of the original range or dimension of the objective function, is given as follows \cite{JR:MOOP},
\begin{eqnarray}\label{eqn:normalization}
\overline{F_j}^{\mathrm{nml}}(\overline{\mathbf{W}}_\mathrm{I},\overline{\mathbf{W}}_\mathrm{E},\theta)=\frac{\overline{F_j}(\overline{\mathbf{W}}_\mathrm{I},\overline{\mathbf{W}}_\mathrm{E},\theta)-\overline{F_j}^0}{\overline{F_j}^*-\overline{F_j}^0},
\end{eqnarray}
where $\overline{F_j}^*$ and $\overline{F_j}^0$ are the maximum and minimum value of the $j$th transformed objective function, i.e., $\overline{F_j}^0\leq \overline{F_j}(\overline{\mathbf{W}}_\mathrm{I},\overline{\mathbf{W}}_\mathrm{E},\theta)\leq \overline{F_j}^*$. $\overline{F_j}^*$ can result from the transformed single-objective problems. Then, the transformed objective functions are normalized to range $[0,1]$.

Regarding to the MOOP \ref{prob:multiobj_WIPTsepuser}, the objective function can be rewritten in a normalization form as $\max_{j=1,2,3}\,\, \{\omega_j(1-\overline{F_j}^{\mathrm{nml}}(\overline{\mathbf{W}}_\mathrm{I},\overline{\mathbf{W}}_\mathrm{E},\theta))\}$. A common approach for handling such a min-max optimization problem is to introduce an auxiliary optimization variable. Then, the MOOP can be transformed into its equivalent epigraph representation \cite{book:convex}, which is given by
\begin{Prob}Transformed MOOP:\label{prob:multiobj_WIPTsepuser2}
\begin{eqnarray}
\underset{\overline{\mathbf{W}}_\mathrm{I},\overline{\mathbf{W}}_\mathrm{E}\in\mathbb{H}^{N_\mathrm{T}},\theta,\tau}{\mino}&&\tau \notag\\
\mathrm{subject\,\,to}\,\,\,&&\overline{\mathrm{C1}} - \overline{\mathrm{C5}},\notag\\
&&\overline{\mathrm{C6}}:\,\omega_j(1-\overline{F_j}^{\mathrm{nml}}(\overline{\mathbf{W}}_\mathrm{I},\overline{\mathbf{W}}_\mathrm{E},\theta))\leq \tau,\,\,\forall j,\notag
\end{eqnarray}
\end{Prob}
\noindent where $\tau$ is the auxiliary optimization variable. We note that the optimal value of Problem \ref{prob:multiobj_WIPTsepuser2} lies between zero and one.

Now, we introduce the following proposition.
\begin{proposition}\label{prop:equivalency}
The transformed problems \ref{prob:WIPT_IR-EE_reform}-\ref{prob:multiobj_WIPTsepuser2} are equivalent transformations of the original problems \ref{prob:WIPT_IR-EE}-\ref{prob:multiobj_WIPTsepuser}, respectively.
\end{proposition}
\begin{proof}
Please refer to Appendix \ref{app:Prop_equivalency}.
\end{proof}
Based on Proposition \ref{prop:equivalency}, we can recover the solution of the original problems based on (\ref{eqn:newvariabledefine}). In particular, the optimal value $\overline{F_j}^*$ or the lower bound $\overline{F_j}^0$ of the $j$th transformed objective function equal to that of the $j$th original objective function, i.e., $\overline{F_j}^*=F_j^*$, $\overline{F_j}^0=F_j^0$, $j=1,2,3$. Thus, in (\ref{eqn:normalization}), $\overline{F_j}^0$ are trivial results from the original objective functions, that are $\overline{F_1}^0=\overline{F_2}^0=0$, $\overline{F_3}^0=P_\mathrm{max}$. We also denote $\overline{F_j}^*$ simply as $\overline{F_1}^*=\Phi_\mathrm{IR}^*$, $\overline{F_2}^*=\Phi_\mathrm{EH}^*$, and $\overline{F_3}^*=0$.

We note that if Problem \ref{prob:multiobj_WIPTsepuser2} can be solved optimally by an algorithm, then the algorithm can also be used to solve Problem \ref{prob:WIPT_IR-EE_reform}-\ref{prob:WIPT_Ptrans_reform}, since Problem \ref{prob:multiobj_WIPTsepuser2} is a generalization of Problem \ref{prob:WIPT_IR-EE_reform}-\ref{prob:WIPT_Ptrans_reform}. Thus, we focus on the method in solving Problem \ref{prob:multiobj_WIPTsepuser2}. It is evident that Problem \ref{prob:multiobj_WIPTsepuser2} is non-convex due to the rank-one beamforming matrix constraint $\overline{\mathrm{C3}}:\,\,\Rank(\overline{\mathbf{W}}_\mathrm{I})\leq1$. Now, we apply the SDP relaxation by removing constraint $\mathrm{C3}$ from Problem \ref{prob:multiobj_WIPTsepuser2}. As a result, the SDP relaxed problem is given by
\begin{Prob}SDP Relaxed Transformed MOOP:\label{prob:multiobj_WIPTsepuser_relaxed}
\begin{eqnarray}
\underset{\overline{\mathbf{W}}_\mathrm{I},\overline{\mathbf{W}}_\mathrm{E}\in\mathbb{H}^{N_\mathrm{T}},\theta,\tau}{\mino}&&\tau \notag\\
\mathrm{subject\,\,to}\,\,\, &&\overline{\mathrm{C1}},\,\overline{\mathrm{C2}},\,\overline{\mathrm{C4}},\,\overline{\mathrm{C5}},\,\overline{\mathrm{C6}},\notag
\end{eqnarray}
\end{Prob}
\noindent which is a convex SDP problem and can be solved by numerical convex program solvers such as CVX \cite{website:CVX}. In particular, if the obtained solution $\overline{\mathbf{W}}_\mathrm{I}^*$ of the SDP relaxed problem satisfies constraint $\overline{\mathrm{C3}}$, i.e., $\Rank(\overline{\mathbf{W}}_\mathrm{I}^*)\leq1$, then it turns out to be the optimal solution. Then, the optimal beamforming vector $\mathbf{w}_{\mathrm{I}}^*$ of the original problem can be achieved by solving the relaxed problem and recovering from the invertible mapping equations (\ref{eqn:newvariabledefine}). Now, we study the tightness of the SDP relaxation by the following theorem.
\begin{Thm}\label{thm:rankone}
The optimal solution of Problem \ref{prob:multiobj_WIPTsepuser_relaxed} satisfies $\Rank(\overline{\mathbf{W}}_\mathrm{I}^*)=1$ and $\Rank(\overline{\mathbf{W}}_\mathrm{E}^*)\leq1$. In particular, an optimal solution with $\Rank(\overline{\mathbf{W}}_\mathrm{I}^*)=1$ and $\overline{\mathbf{W}}_\mathrm{E}^*=\mathbf{0}$ can always be constructed.
\end{Thm}
\begin{proof}
Please refer to Appendix \ref{app:rankone}.
\end{proof}
Therefore, the adopted SDP relaxation is tight. Besides, Problem \ref{prob:WIPT_IR-EE_reform}-\ref{prob:WIPT_Ptrans_reform} can be solved by SDP relaxation as solving Problem \ref{prob:multiobj_WIPTsepuser_relaxed}.

Next, we construct an optimal solution with $\Rank(\overline{\mathbf{W}}_\mathrm{I}^*)=1$ and $\overline{\mathbf{W}}_\mathrm{E}^*=\mathbf{0}$ based on Theorem \ref{thm:rankone}. We redefine the optimization variable $\overline{\mathbf{W}}_\mathrm{I}$ as
\begin{eqnarray}\label{eqn:newvariabledefine2}
\overline{\mathbf{W}}_\mathrm{I}=\lambda\mathbf{u}\mathbf{u}^H,\,\,\mathbf{u}=[u_1,u_2,\dots,u_{N_\mathrm{T}}]^T,\,\,\mathrm{and}\,\,\overline{\mathbf{W}}_\mathrm{E}=\mathbf{0},
\end{eqnarray}
where $\mathbf{u}\in \mathbb{C}^{N_{\mathrm{T}}\times 1}$. $\mathbf{u}$ is an orthonormal vector, i.e., $\sum_{i=1}^{N_\mathrm{T}}\left|u_i\right|^2=1$. According to (\ref{eqn:newvariabledefine}) and (\ref{eqn:newvariabledefine2}), we have $\mathbf{w}_\mathrm{I}=\sqrt{\frac{\lambda}{\theta}}\mathbf{u}$. Then, MOOP \ref{prob:multiobj_WIPTsepuser2} can be reformed with respect to the optimization variables $\{\mathbf{u}, \lambda, \theta, \tau\}$ as follows:
\begin{Prob}\label{prob:multiobj_WIPTsepuser3}
\begin{eqnarray}
\underset{\mathbf{u},\lambda,\theta,\tau}{\mino}&&\tau \notag\\
\mathrm{subject\,\,to} &&\widehat{\mathrm{C1}}:\,\lambda\sum_{i=1}^{N_\mathrm{T}}\left|u_i\right|^2\leq\theta P_{\mathrm{max}},\notag\\
&&\widehat{\mathrm{C2}}:\,\,\lambda\ge0,\quad\widehat{\mathrm{C3}}:\,\,\theta\ge0, \notag\\
&&\widehat{\mathrm{C4}}:\,\,\frac{\lambda\sum_{i=1}^{N_\mathrm{T}}\left|u_i\right|^2}{\xi}+\theta(N_\mathrm{T}P_\mathrm{ant}+P_\mathrm{c})\leq1,\notag\\
&&\widehat{\mathrm{C5}}:\,\omega_1\Big(1-\frac{\theta}{\Phi_\mathrm{IR}^*}\log_2\Big(1+\frac{\lambda\left|\sum_{i=1}^{N_\mathrm{T}}u_i^*h_i\right|^2}{\theta\sigma_\mathrm{I}^2}\Big)\Big)\leq \tau, \notag\\
&&\widehat{\mathrm{C6}}:\,\omega_2\Big(1-\frac{\eta}{\Phi_\mathrm{EH}^*}\lambda\left|\sum_{i=1}^{N_\mathrm{T}}u_i^*g_i\right|^2\Big)\leq \tau, \notag\\
&&\widehat{\mathrm{C7}}:\,\omega_3\frac{\xi}{P_\mathrm{max}}(\frac{1}{\theta}-N_\mathrm{T}P_\mathrm{ant}-P_\mathrm{c})\leq \tau, \notag
\end{eqnarray}
\end{Prob}
\noindent where $h_i$ and $g_i$, $i\in\{1,\dots,N_\mathrm{T}\}$, are the elements of channel vectors $\mathbf{h}$ and $\mathbf{g}$, respectively. In order to investigate the structure of vector $\mathbf{u}$, we analyze the Karush-Kuhn-Tucker (KKT) conditions of Problem \ref{prob:multiobj_WIPTsepuser3} by introducing the Lagrangian function. The Lagrangian function is given by
\begin{eqnarray} \label{eqn:lagrangian_moop3}
&&{\cal L}\big(\mathbf{u},\lambda,\theta,\tau,\mu,\nu,\kappa_1,\kappa_2,\kappa_3,\zeta\big)\\
&=&\tau+\mu\big(\lambda\sum_{i=1}^{N_\mathrm{T}}\left|u_i\right|^2-\theta P_{\mathrm{max}}\big)+\nu\big(\frac{\lambda\sum_{i=1}^{N_\mathrm{T}}\left|u_i\right|^2}{\xi}+\theta(N_\mathrm{T}P_\mathrm{ant}+P_\mathrm{c})-1\big)-\zeta\theta\notag\\
&+&\kappa_1\Big[\omega_1\big(1-\frac{\theta}{\Phi_\mathrm{IR}^*}\log_2(1+\frac{\lambda\left|\sum_{i=1}^{N_\mathrm{T}}u_i^*h_i\right|^2}{\theta\sigma_\mathrm{I}^2})\big)-\tau\Big]\notag\\
&+&\kappa_2\Big[\omega_2\big(1-\frac{\eta}{\Phi_\mathrm{EH}^*}\lambda\left|\sum_{i=1}^{N_\mathrm{T}}u_i^*g_i\right|^2\big)-\tau\Big]+\kappa_3\Big[\omega_3\frac{\xi}{P_\mathrm{max}}(\frac{1}{\theta}-N_\mathrm{T}P_\mathrm{ant}-P_\mathrm{c})-\tau\Big],\notag
\end{eqnarray}
where $\mu,\nu,\kappa_1,\kappa_2,\kappa_3,\zeta$ are dual variables associated with the corresponding constraints, respectively. Constraint C2 is captured in the solution when deriving KKT conditions in the following. Since Problem \ref{prob:multiobj_WIPTsepuser3} satisfies Slater's constraint qualification and is convex with respect to the optimization variables, strong duality holds. Then, based on KKT optimality conditions, the gradient of Lagrangian function with respect to $u_i$, the element of $\mathbf{u}$, vanishes, from which we can result in\\
\begin{eqnarray}\label{eqn:w_i}
u_i&=&\omega_1ah_i+\omega_2bg_i,\\
\mathrm{where}\quad a&=&\frac{\kappa_1\theta\lambda\big(\sum_{i=1}^{N_\mathrm{T}}u_i^*h_i\big)^*}{\Phi_\mathrm{IR}^*(\mu+\frac{\nu}{\xi})\Big(\theta\sigma_\mathrm{I}^2+\lambda\left|\sum_{i=1}^{N_\mathrm{T}}u_i^*h_i\right|^2\Big)}\notag\\
\mathrm{and}\quad b&=&\frac{\kappa_2\eta\lambda}{\Phi_\mathrm{EH}^*(\mu+\frac{\nu}{\xi})}\Big(\sum_{i=1}^{N_\mathrm{T}}u_i^*g_i\Big)^*.\notag
\end{eqnarray}
Similarly, consider KKT condition with respect to $\lambda$, which is given by
\begin{eqnarray}\label{eqn:lambda}
\lambda=\theta\Bigg[\frac{\kappa_1\omega_1/\Phi_\mathrm{IR}^*}{ln(2)\Big(\mu+\frac{\nu}{\xi}-\frac{\kappa_2\eta\omega_2}{\Phi_\mathrm{EH}^*}\left|\sum_{i=1}^{N_\mathrm{T}}u_i^*g_i\right|^2\Big)}-\frac{\sigma_\mathrm{I}^2}{\left|\sum_{i=1}^{N_\mathrm{T}}u_i^*h_i\right|^2}\Bigg]^+.
\end{eqnarray}

As we can see, (\ref{eqn:w_i}) and (\ref{eqn:lambda}) imply the structure of beamforming vector $\mathbf{w}_\mathrm{I}$ by considering $\mathbf{w}_\mathrm{I}=\sqrt{\frac{\lambda}{\theta}}\mathbf{u}$. In particular, (\ref{eqn:w_i}) indicates the direction of the information signal. (\ref{eqn:lambda}) shows that the power allocation for the information signal follows the policy of water-filling solution. For specific case, when IR-EE is considered and EH-EE is discarded, i.e., $\omega_1\neq0$ and $\omega_2=0$, information beamforming vector $\mathbf{w}_\mathrm{I}$ is aligning to the direction of channel vector $\mathbf{h}$ according to (\ref{eqn:w_i}). Since $\sum_{i=1}^{N_\mathrm{T}}\left|u_i\right|^2=1$, we obtain
\begin{eqnarray}
\mathbf{w}_\mathrm{I}=\sqrt{p}\frac{\mathbf{h}}{\norm{\mathbf{h}}},\,\,\,\text{where}\,\,\,p=\Bigg[\frac{\kappa_1\omega_1/\Phi_\mathrm{IR}^*}{\ln(2)(\mu+\nu/\xi)}-\frac{\sigma_\mathrm{I}^2}{\norm{\mathbf{h}}^2}\Bigg]^+.
\end{eqnarray}

On the other hand, when IR-EE is not taken into account and EH-EE is maximized, i.e., $\omega_1=0$ and $\omega_2\neq0$, the beamforming vector $\mathbf{w}_\mathrm{I}$ directs to the energy harvester by following the direction of channel vector $\mathbf{g}$ as (\ref{eqn:w_i}) indicates. Especially, Problem \ref{prob:multiobj_WIPTsepuser3} becomes a linear programming with respect to $\lambda$. In extreme case, if transmit power minimization is not considered either, i.e., $\omega_3=0$, we solve a single-objective problem for EH-EE maximization. Then, the optimal solution is given as
\begin{eqnarray}
\mathbf{w}_\mathrm{I}=\sqrt{P_\mathrm{max}}\frac{\mathbf{g}}{\norm{\mathbf{g}}},\,\,\,\Phi_\mathrm{EH}^*=\frac{\eta P_\mathrm{max}\norm{\mathbf{g}}^2}{\frac{P_\mathrm{max}}{\xi}+N_\mathrm{T}P_\mathrm{ant}+P_\mathrm{c}}.
\end{eqnarray}

Furthermore, when both IR-EE maximization and EH-EE maximization are active objectives, i.e., $\omega_1\neq0$ and $\omega_2\neq0$, $\mathbf{w}_\mathrm{I}$ is designed as a dual use beamforming vector for simultaneous information delivery and power transfer. (\ref{eqn:w_i}) shows that it incorporates the directions of both channel vectors $\mathbf{h}$ and $\mathbf{g}$.
\section{Results}
In this section, we present simulation results to demonstrate the system performance of multi-objective system design. The simulation parameters are summarized in Table \ref{table:parameters}. In particular, we adopt the TGn path loss model \cite{report:tgn}. The multipath fading is modeled as Rician fading with Rician factor $3$ dB. Assume the carrier center frequency as $470$ MHz with bandwidth $200$ kHz. At the transmitter, we set the dynamic power consumption $P_\mathrm{ant}=75$ mWatt per antenna, static circuit power consumption $P_\mathrm{c}=1$ Watt, and the power amplifier efficiency $\xi=0.4$. The maximum transmit power is given as $P_\mathrm{max}=1$ Watt. Two receivers, namely, information receiver and energy harvester, are uniformly located between the reference distance 1 meters and the maximum service distance $10$ meters. Each receiver is equipped with a single antenna with antenna gain 10 dBi. Assume the noise covariances at the information receiver and the energy harvester are the same, i.e., $\sigma_\mathrm{I}^2=\sigma_\mathrm{E}^2=\sigma^2$. We set $\sigma^2=-47$ dBm which includes thermal noise at a temperature of $290$ Kelvin and signal processing noise. The signal processing noise is caused by a $12$-bit uniform quantizer employed in the analog-to-digital converter at the analog front-end of each receiver. At the energy harvester, the energy conversion efficiency for converting RF energy to electrical energy is $\eta=0.8$. In this setting, multiple channel realizations are simulated, where both pass loss and multipath fading effects are taken into account.
\begin{table}[htb]
\caption{Simulation Parameters} \label{table:parameters}
\centering
\begin{tabular}{ | l | l | } \hline
      Carrier center frequency                           & 470 MHz\\ \hline
      Bandwidth                                          & ${\cal B}=200$ kHz \\ \hline 
      Single antenna power consumption                   & $P_\mathrm{ant}=75$ mW \\ \hline
      Static circuit power consumption                           & $P_\mathrm{c}=1$ W \\ \hline
      Power amplifier efficiency                         & $\xi=0.4$ \\ \hline
      Antennas gain                                     & 10 dBi \\ \hline
      Noise power                                        & $\sigma^2= -47$ dBm \\ \hline
      Rician factor                                      & 3 dB \\ \hline
      Reference distance                                 & 1 meters \\ \hline
      Maximum service distance                           & 10 meters \\ \hline
      Energy conversion efficiency                       & $\eta=0.8$ \\ \hline
\end{tabular}
\end{table}

\begin{figure}[htb]
        \centering
        \includegraphics[scale=0.9]{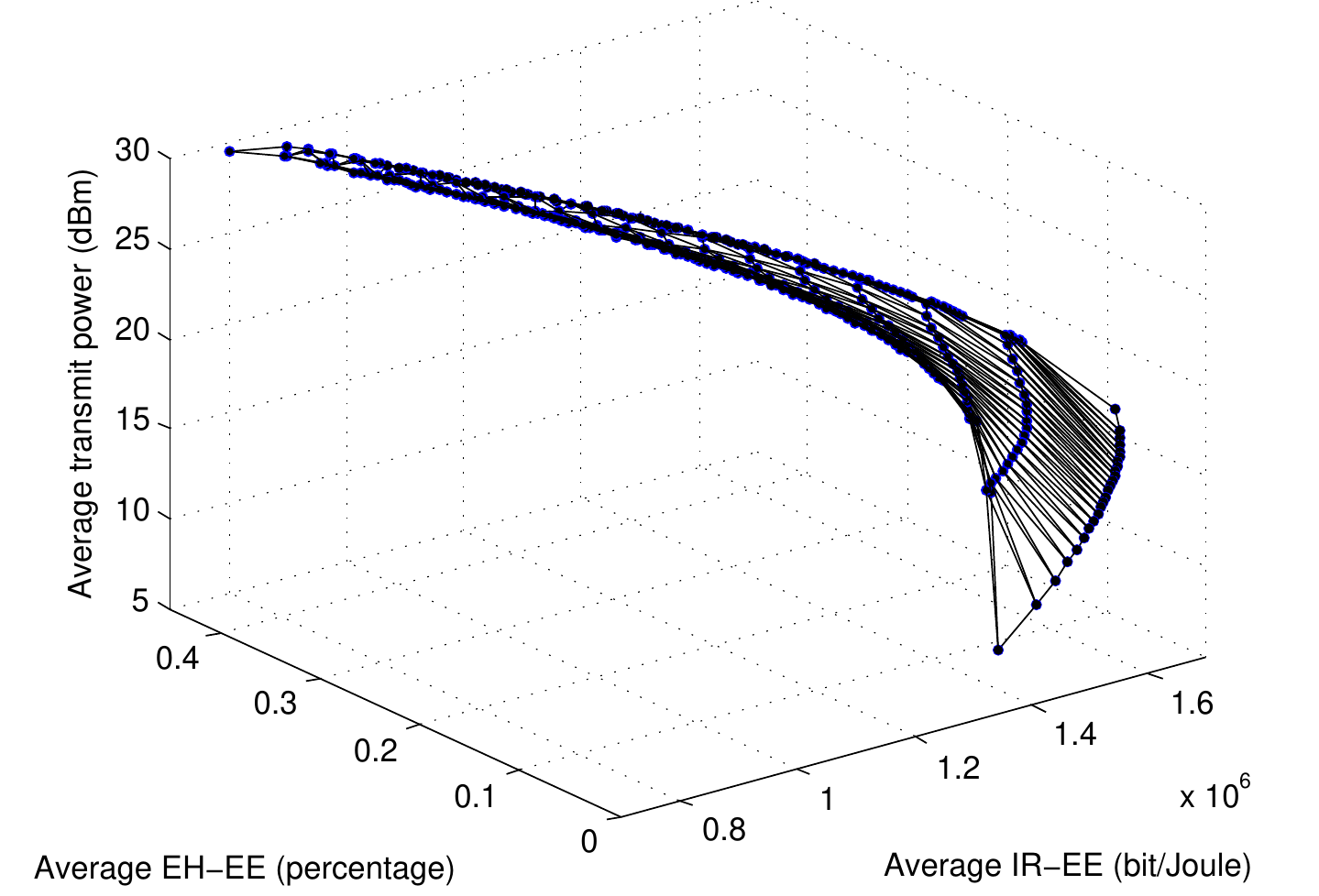}
        \caption{System performance trade-off region between IR-EE, EH-EE, and transmit power.}
        \label{fig:sepuser_3D_1}
        \vspace*{5mm}
        \includegraphics[scale=0.9]{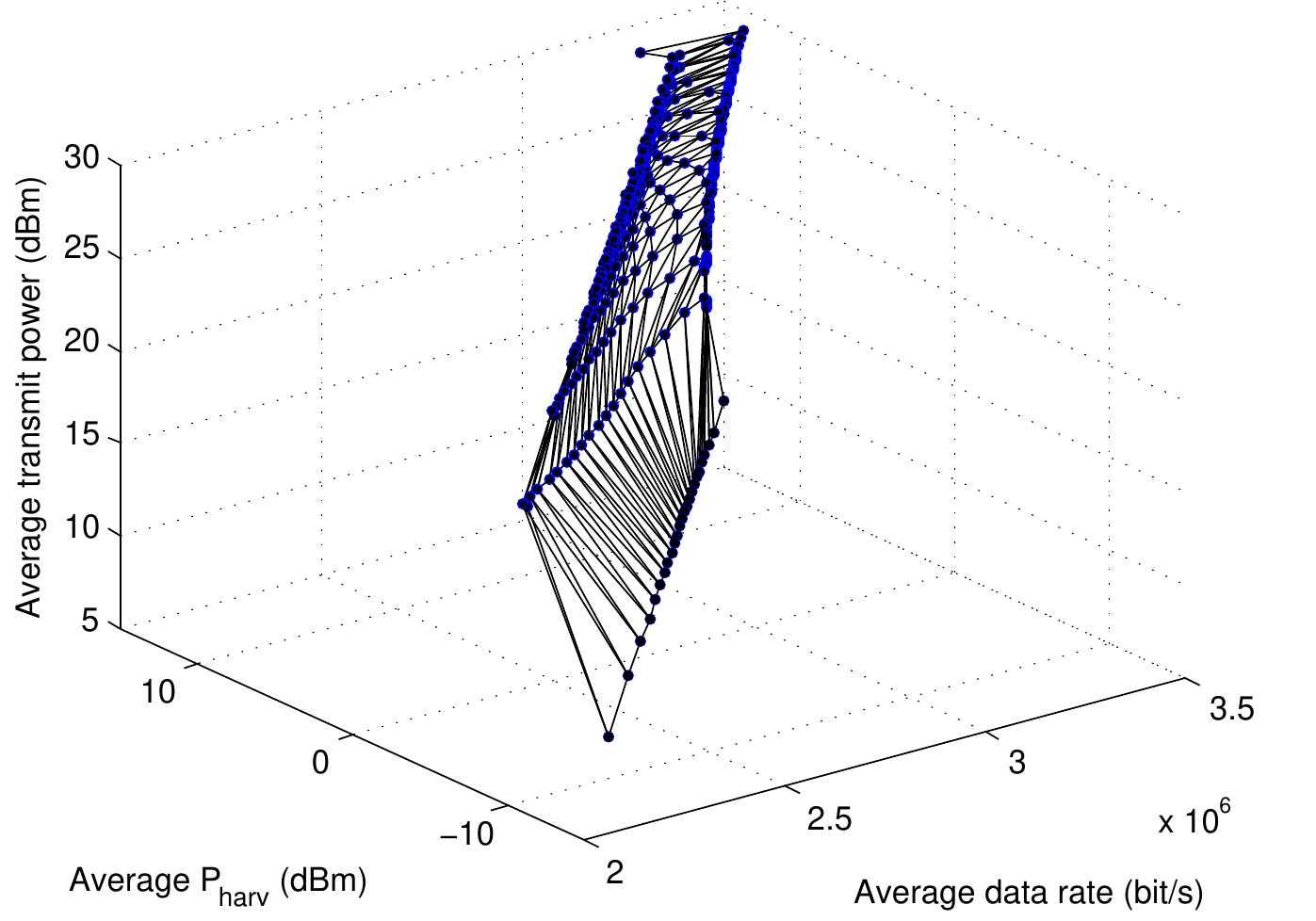}
        \caption{System performance trade-off region between achievable rate, harvested energy, and transmit power.}
        \label{fig:sepuser_3D_2}
\end{figure}

In the following, we show the trade-off region between multiple system objectives from two aspects. In one aspect, we examine the trade-off between the average IR-EE, EH-EE, and transmit power in terms of system EE, which is shown in Figures \ref{fig:sepuser_3D_1}, \ref{fig:2d_IRee_EHee}, \ref{fig:2d_IRee_Pt}, \ref{fig:2d_EHee_Pt}, \ref{fig:sepuser_IRee_EHee}, \ref{fig:sepuser_IRee_Ptrans}, and \ref{fig:sepuser_EHee_Ptrans}. On the other hand, from the aspect of system throughput, the trade-off between average achievable rate, harvested energy, and transmit power is illustrated in Figures \ref{fig:sepuser_3D_2}, \ref{fig:sepuser_Rate_Pharv}, \ref{fig:sepuser_Rate_Ptrans}, and \ref{fig:sepuser_Pharv_Ptrans}. For comparison, we also propose a baseline scheme, where MOOP of achievable rate maximization, harvested energy maximization, and transmit power minimization is solved. The system performance are compared between the proposed EE algorithm and the baseline scheme in Figure \ref{fig:sepuser_IRee_EHee}-- Figure \ref{fig:sepuser_Pharv_Ptrans}.

Figures \ref{fig:sepuser_3D_1} and \ref{fig:sepuser_3D_2} give the 3-dimension trade-off regions of the system energy efficiency and system throughput, respectively, for 8 transmitting antennas. The 3-dimension trade-off regions are obtained by solving the MOOP \ref{prob:multiobj_WIPTsepuser} with different sets of weights on the system design objectives. Specifically, the points consisting of the regions are calculated out by uniformly varying the weight $\omega_j$ with a step size of $0.04$ such that $\sum_j\omega_j=1$. It can be observed in Figure \ref{fig:sepuser_3D_1} that the trade-off region between IR-EE, EH-EE, and transmit power is formed by the points gradually spreading from the right bottom corner to the left top corner. In particular, both IR-EE and EH-EE grow rapidly for small transmit power. When the transmit power is high, EH-EE continues increasing, however, IR-EE declines. On the other hand, Figure \ref{fig:sepuser_3D_2} illustrates that high transmit power supports the increment of both the achievable rate and harvested energy.

\begin{figure}[htb]
        \centering
        \includegraphics[scale=0.9]{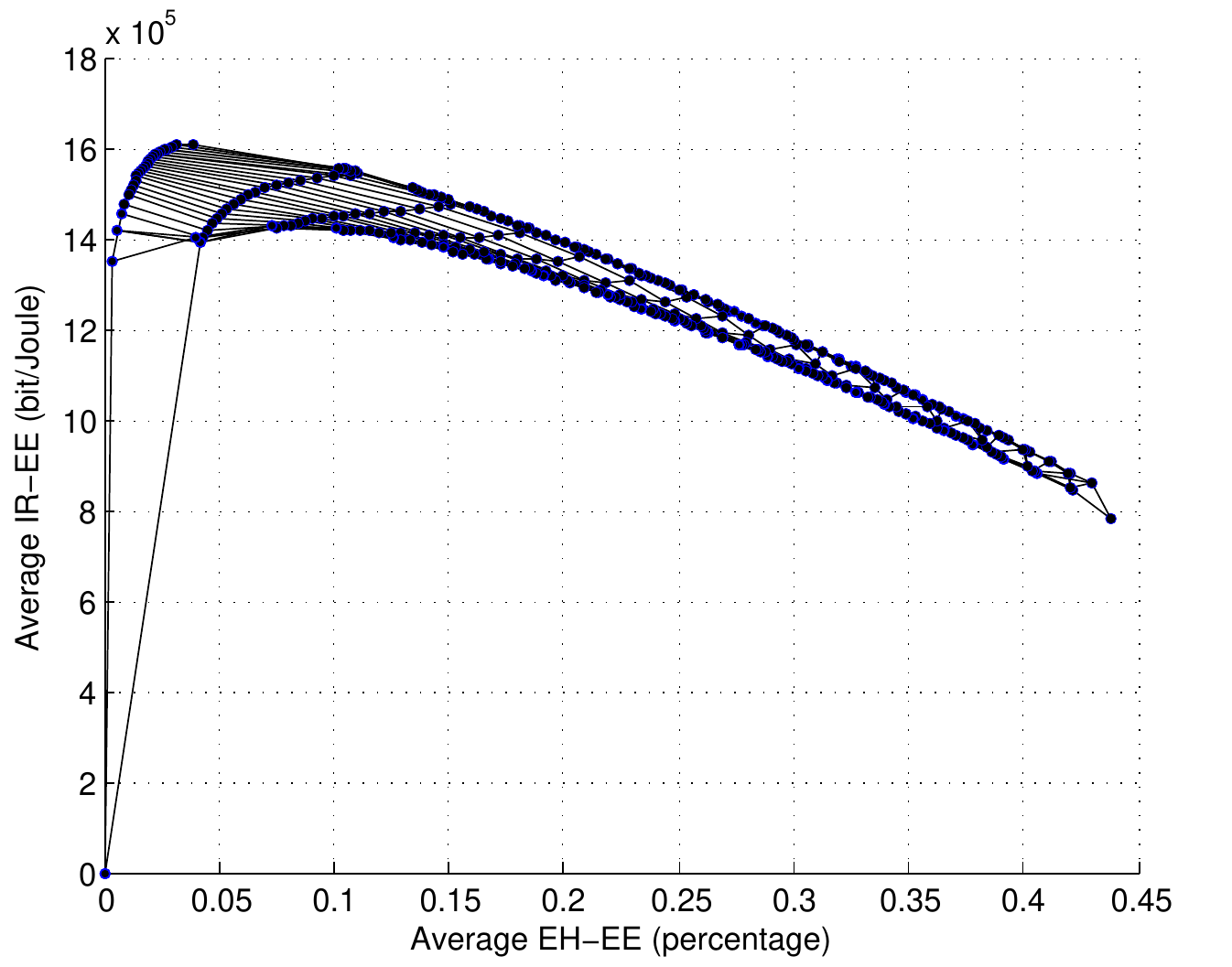}
        \caption{Trade-off region between IR-EE and EH-EE.}
        \label{fig:2d_IRee_EHee}
\end{figure}
\begin{figure}[htb]
        \centering
        \includegraphics[scale=0.9]{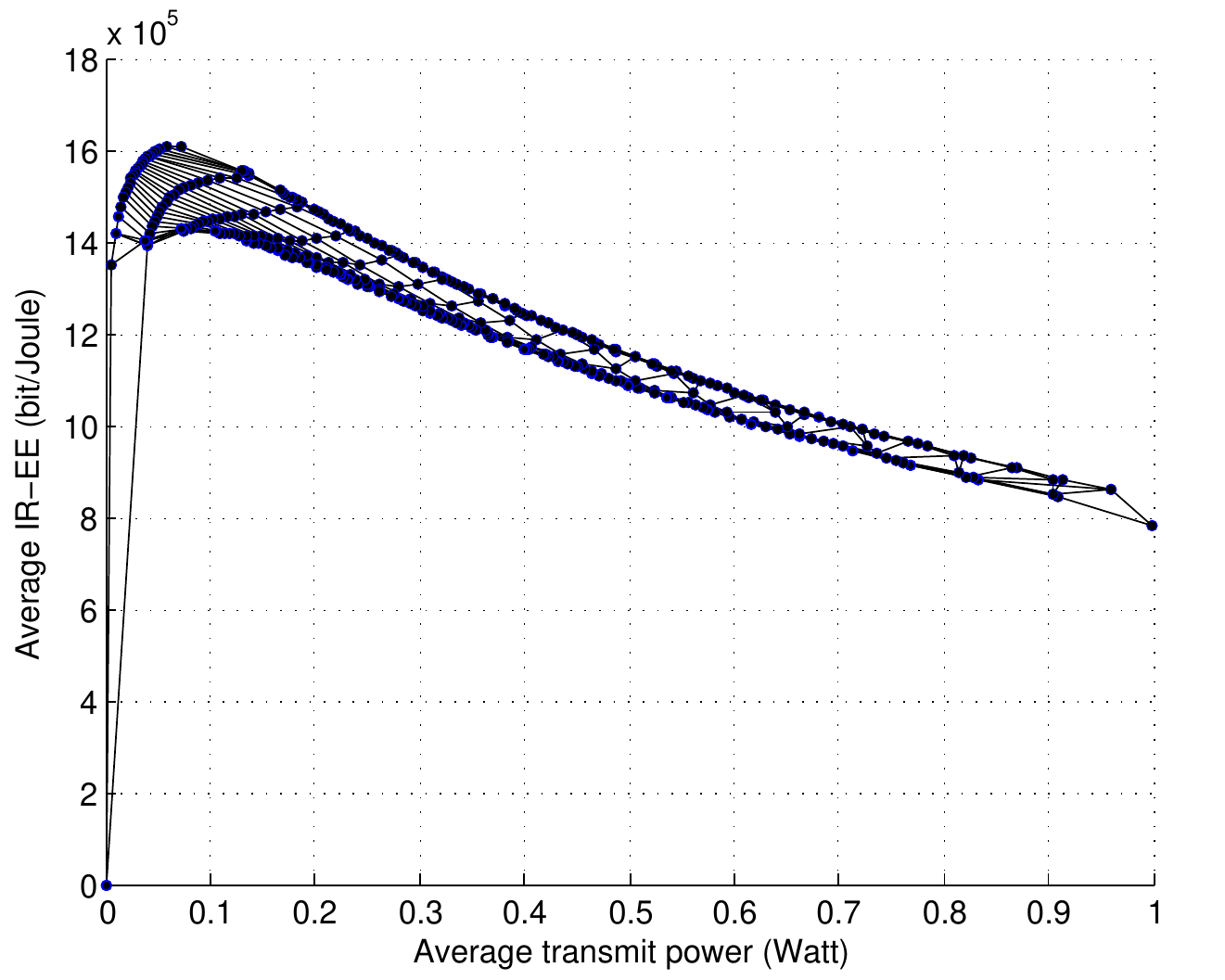}
        \caption{Trade-off region between IR-EE and transmit power.}
        \label{fig:2d_IRee_Pt}
        \vspace*{5mm}
        \includegraphics[scale=0.9]{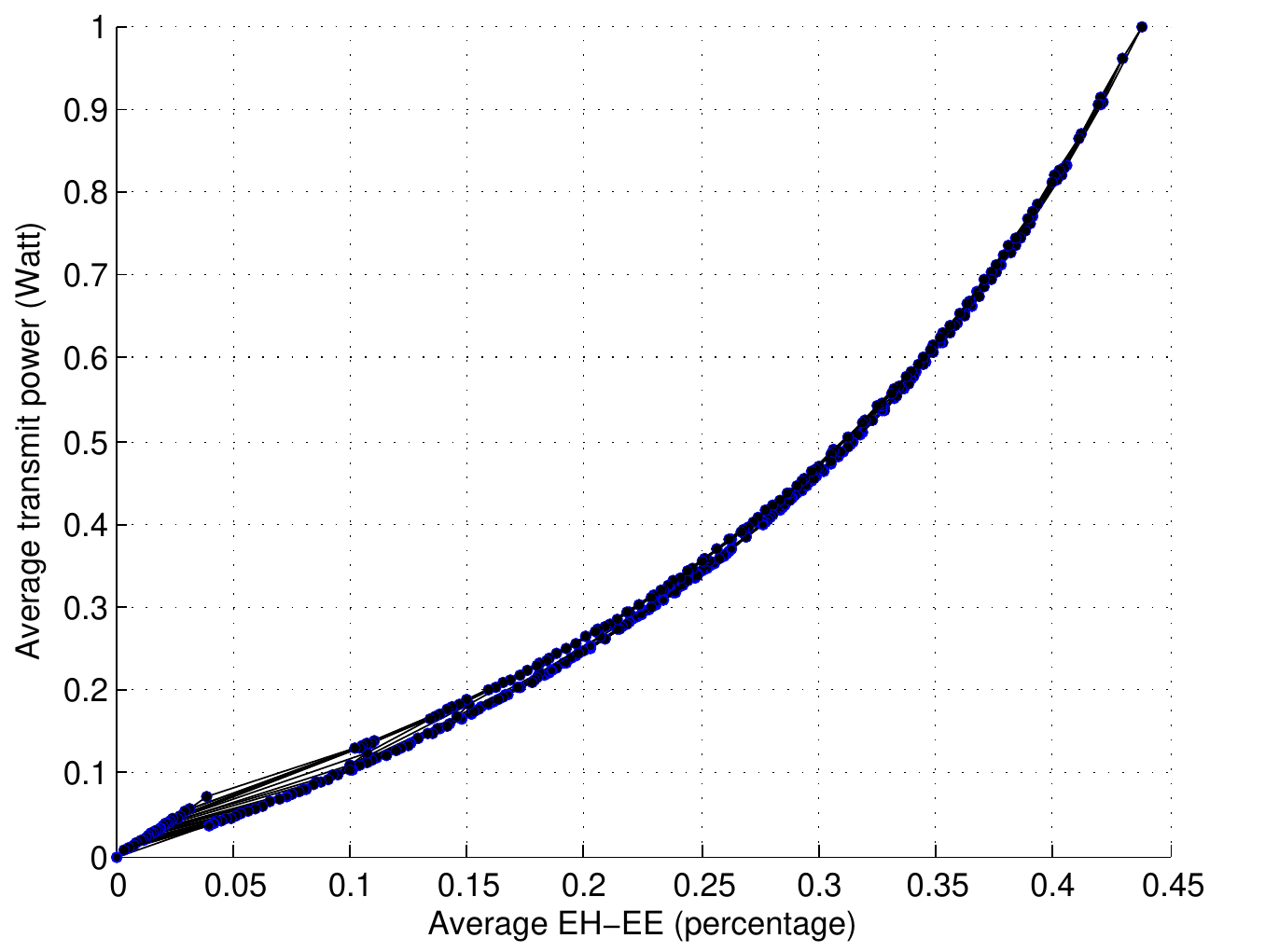}
        \caption{Trade-off region between EH-EE and transmit power.}
        \label{fig:2d_EHee_Pt}
\end{figure}

In addition, for a better illustration, we also provide different side-views of the 2-dimension trade-off region in Figures \ref{fig:2d_IRee_EHee}, \ref{fig:2d_IRee_Pt}, and \ref{fig:2d_EHee_Pt} for revealing the trade-offs between different pairs of objective functions. Figure \ref{fig:2d_IRee_EHee} shows the trade-off between IR-EE and EH-EE. Figure \ref{fig:2d_IRee_Pt} shows the trade-off between IR-EE and transmit power. Figure \ref{fig:2d_EHee_Pt} shows the trade-off between EH-EE and transmit power. It can be observed from these figures that IR-EE and EH-EE are partially aligned with each other for small transmit power. In particular, IR-EE and EH-EE both increase rapidly when transmit power grows from zero. However, IR-EE reduces dramatically in the high transmit power regime which results in bell-shaped curves as shown in Figures \ref{fig:2d_IRee_EHee} and \ref{fig:2d_IRee_Pt}. This diminishing return of IR-EE is due to the slow logarithmical growth of the achievable rate in the high transmit power regime while the transmit power is linearly increasing. In contrast, EH-EE is monotonically increasing with the increasing transmit power as shown in Figure \ref{fig:2d_EHee_Pt}. In other words, more energy is carried by the transmit signal, higher EH-EE can be achieved. This is thanks to the linear relationship between the harvested energy and the transmit power.

Notably, the trade-off region in Figure \ref{fig:2d_IRee_Pt} is non-convex. In other words, the proposed multi-objective system design algorithm is able to obtain the non-convex feasible region, despite of the non-convexity of the MOOP. Besides, the three extreme points in Figure \ref{fig:2d_IRee_EHee} correspond to the three single-objective functions, respectively. The zero point for IR-EE and EH-EE in Figure \ref{fig:2d_IRee_EHee} corresponds to the zero transmit power in Figures \ref{fig:2d_IRee_Pt} and \ref{fig:2d_EHee_Pt}, which represents the minimum transmit power. It is the optimal value of single-objective Problem \ref{prob:WIPT_Ptrans} which can also be obtained by solving the MOOP with $\omega_3=1$. The second extreme point in the middle of the curve in Figure \ref{fig:2d_IRee_EHee} is the maximum IR-EE, i.e., the optimal value of single-objective Problem \ref{prob:WIPT_IR-EE}, which can also be obtained by solving the MOOP with $\omega_1=1$. The third extreme point at the tail in Figure \ref{fig:2d_IRee_EHee} demonstrates the maximum EH-EE, which is the optimal value of single-objective Problem \ref{prob:WIPT_EH-EE}. It can also result from the MOOP with $\omega_2=1$.

In Figures \ref{fig:sepuser_IRee_EHee} and \ref{fig:sepuser_Rate_Pharv}, the average IR-EE versus the average EH-EE, and the average achievable rate versus the average harvested energy are showed, respectively. These curves are obtained by solving the MOOP for $\omega_3=0$ and $0\leq\omega_j\leq1, j=1,2$, where the value of $\omega_j$ is uniformly varied with a step size of $0.01$ such that $\sum_j\omega_j=1$. Figure \ref{fig:sepuser_IRee_EHee} shows the trade-off between IR-EE and EH-EE when the objective of transmit power minimization is not considered. We can see that IR-EE is monotonically decreasing as EH-EE increasing, since the objective preference shifts from IR-EE to EH-EE, i.e., $\omega_1$ decreases and $\omega_2$ increases. Interestingly, we have a distinct dropping point at the tail of the curve corresponding to $\omega_1=0$ and $\omega_2=1$. This point indicates the solution of the single-objective problem of EH-EE maximization. Based on Theorem \ref{thm:rankone} and Appendix \ref{app:rankone}, we have $\Rank(\mathbf{W}_\mathrm{E})=1$ at this point instead of $\mathbf{W}_\mathrm{E}=\mathbf{0}$ in other points. In other words, the energy signal occupies a part of the total available power. Thus, IR-EE drops due to a smaller power allocation on information signal. Moreover, compared to the baseline scheme, it is obvious that IR-EE of the proposed EE algorithm achieves a significant gain. Besides, when the number of transmitting antenna is increased from $N_\mathrm{T}=4$ to $N_\mathrm{T}=8$, the trade-off region is enlarged in both EE algorithm and the baseline scheme. Since extra degrees of freedom offered by more transmitting antennas can be exploited. Thus, the system performance on EE is improved. However, IR-EE for $N_\mathrm{T}=8$ is smaller than that for $N_\mathrm{T}=4$ in the low transmit power regime. This can be explained that with small transmit power, the achievable rates for both $N_\mathrm{T}=8$ and $N_\mathrm{T}=4$ are quite small. But the total power consumption is large for $N_\mathrm{T}=8$ due to a high antenna power dissipation, which result in a smaller IR-EE.

\begin{figure}[htb]
\begin{minipage}[t]{0.5\textwidth}
        \centering
        \includegraphics[width=\textwidth,height=\textwidth]{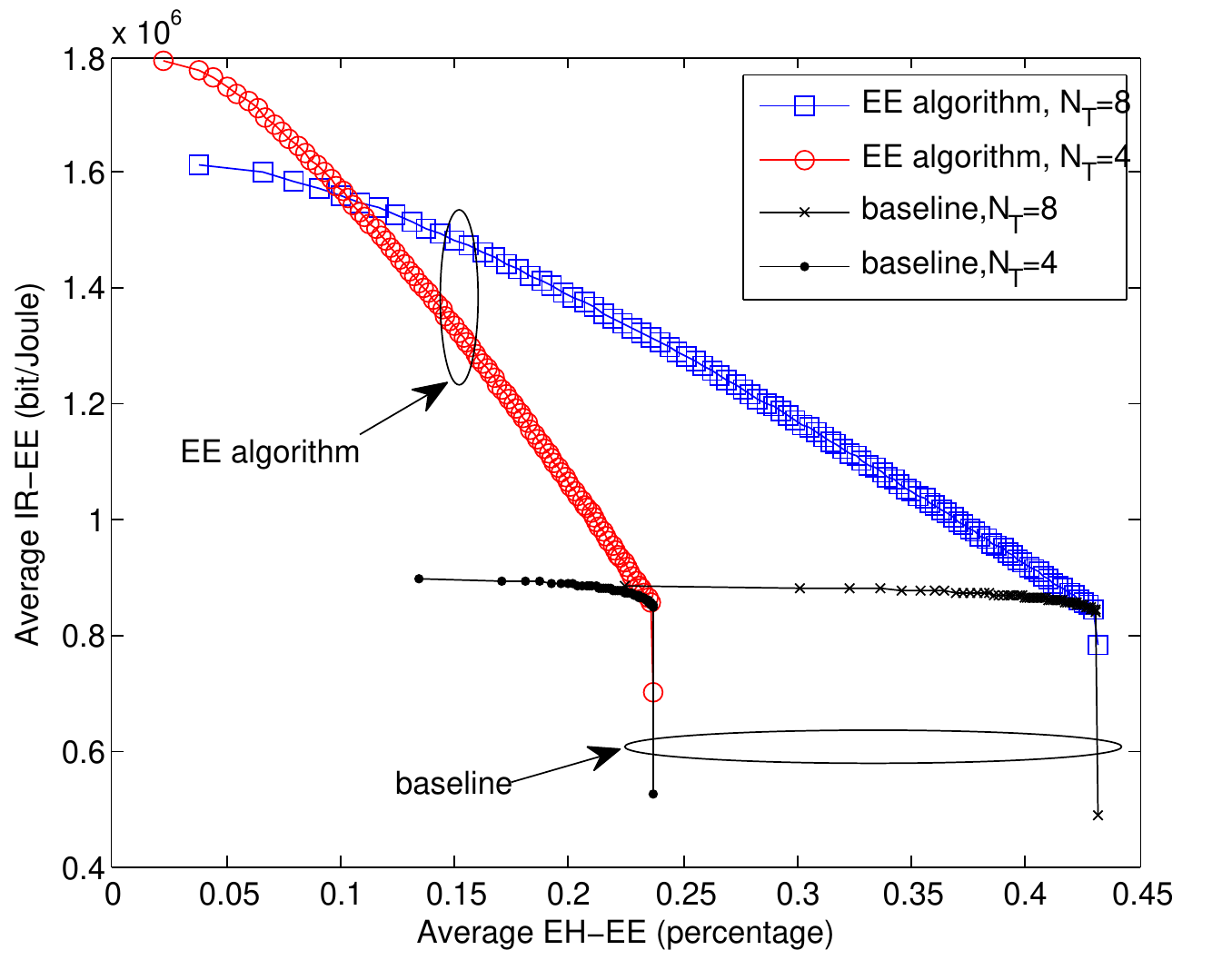}
        \caption{Average IR-EE versus \\average EH-EE.}
        \label{fig:sepuser_IRee_EHee}
\end{minipage}
\begin{minipage}[t]{0.5\textwidth}
\centering
        \includegraphics[width=\textwidth,height=\textwidth]{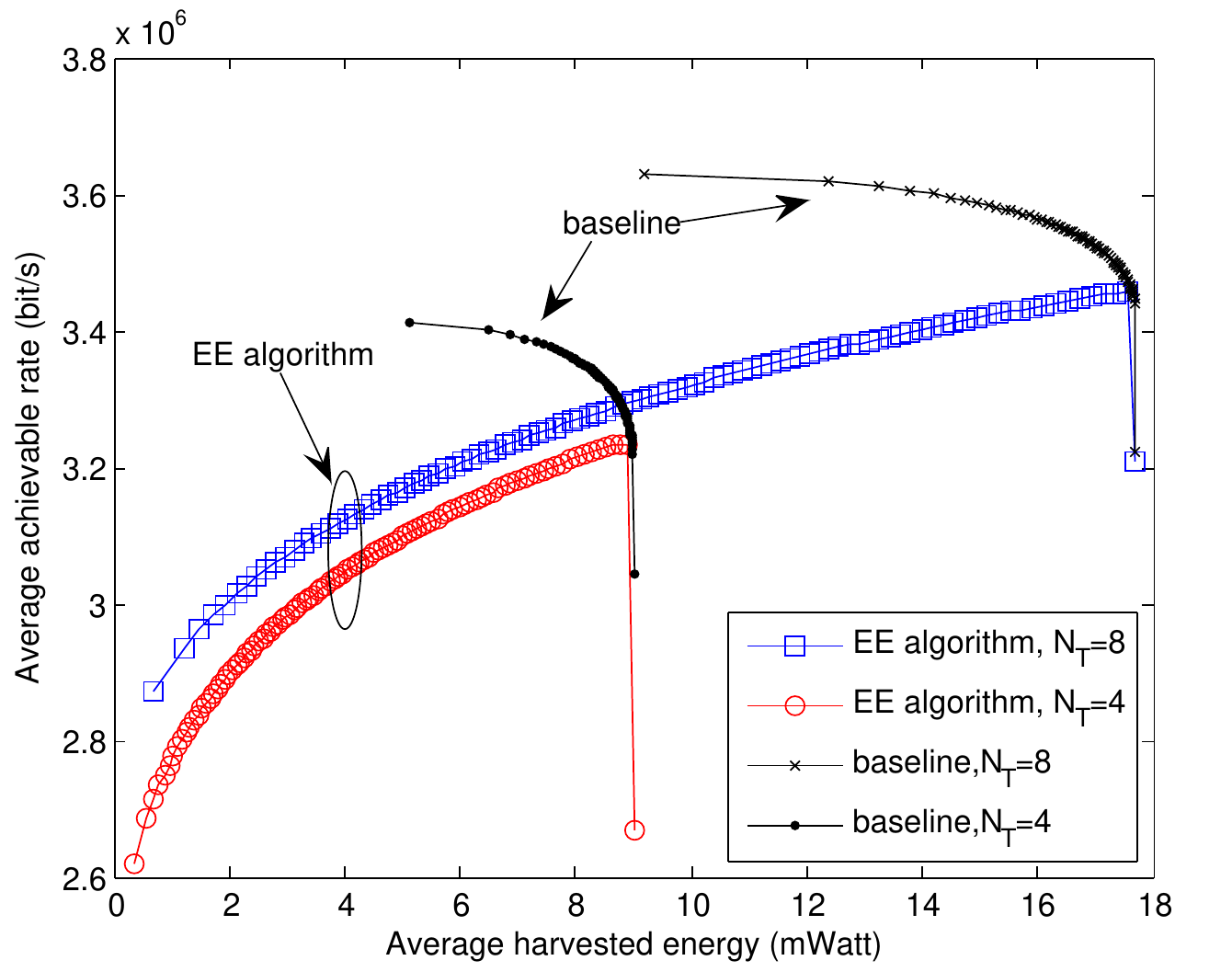}
        \caption{Average achievable rate versus \\average harvested energy.}
        \label{fig:sepuser_Rate_Pharv}
\end{minipage}
\end{figure}

In terms of of system throughput, we see from Figure \ref{fig:sepuser_Rate_Pharv} that there is a consistent trend between the achievable rate and the harvested energy for EE algorithm. This is because as EH-EE is gradually emphasized, the amount of harvested energy grows due an increment of transmit power. On the other hand, according to Theorem \ref{thm:rankone} and Appendix \ref{app:rankone}, $\Rank(\mathbf{W}_\mathrm{I})=1$ and $\mathbf{W}_\mathrm{E}=\mathbf{0}$ when IR-EE and EH-EE are both optimized. It means that the information signal occupies all the available transmit power. Thus, the information signal becomes stronger with the increasing transmit power which brings the improvement of achievable rate. As a result, the alignment between the achievable rate and the harvested energy occurs. In contrast to this consistency in the proposed algorithm, the achievable rate and the harvested energy in the baseline scheme conflict with each other. In particular, high data rate corresponds to low harvested energy, and vice versa. Moreover, it is noted that the curve stretches to the right end with a distinct drooping point as in Figure \ref{fig:sepuser_IRee_EHee}. This is caused by the same reason aforementioned. Besides, when more transmitting antennas are equipped at the transmitter, the system throughput with respect to the achievable rate and the harvested energy are both improved since extra degrees of freedom are utilized.

\begin{figure}[htb]
\begin{minipage}[t]{0.5\textwidth}
        \centering
        \includegraphics[width=\textwidth,height=\textwidth]{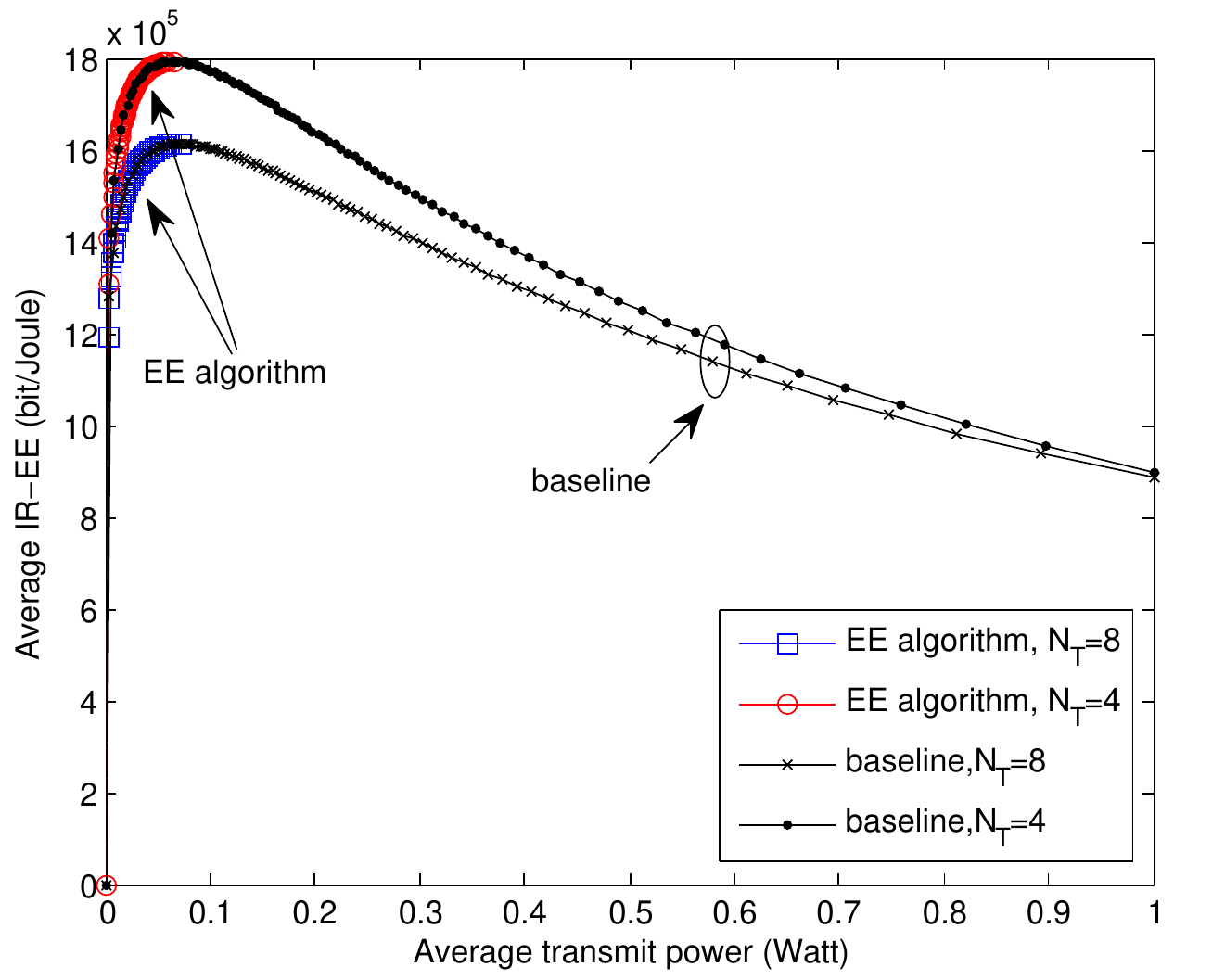}
        \caption{Average IR-EE versus \\average transmit power.}
        \label{fig:sepuser_IRee_Ptrans}
\end{minipage}
\begin{minipage}[t]{0.5\textwidth}
        \centering
        \includegraphics[width=\textwidth,height=\textwidth]{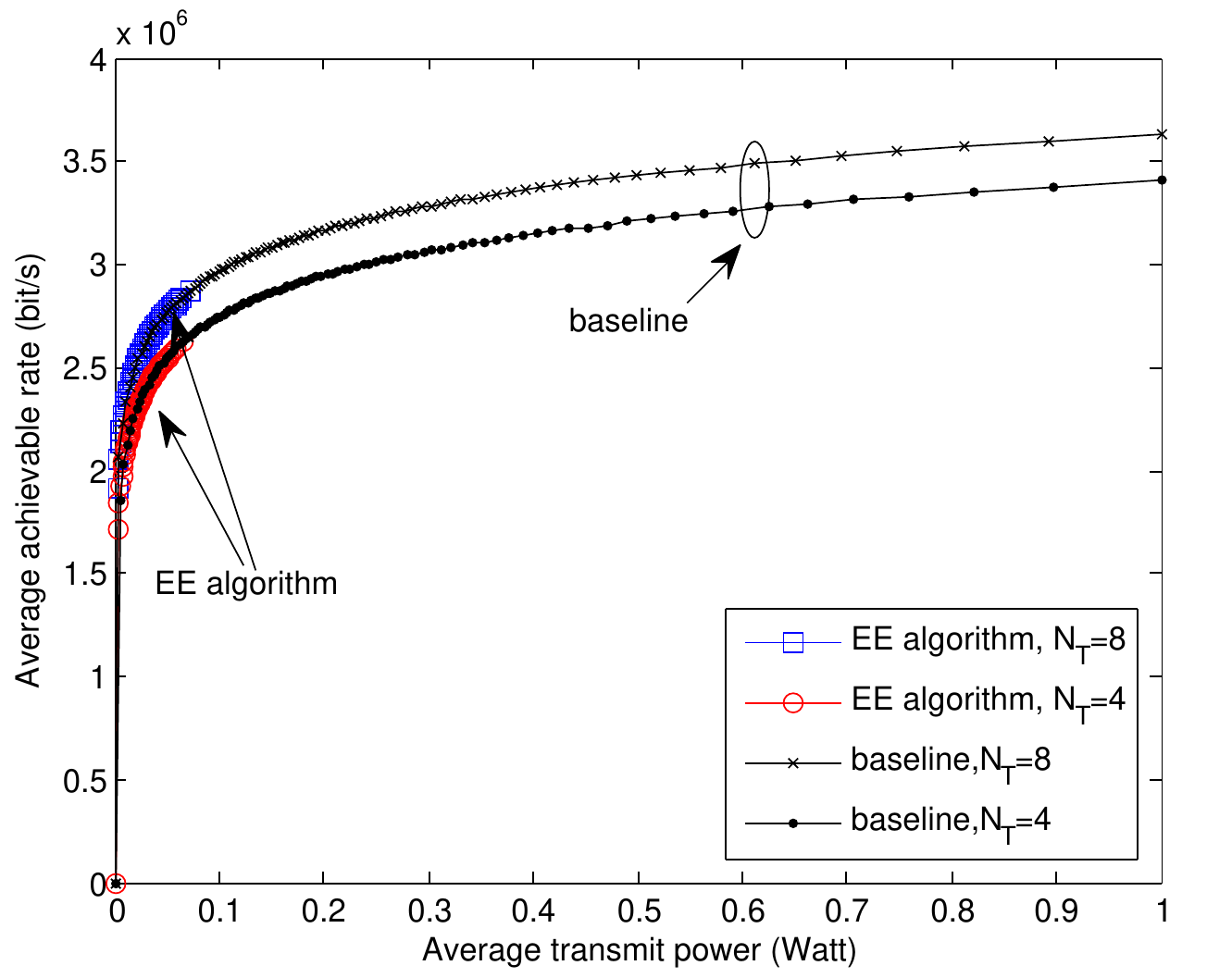}
        \caption{Average achievable rate versus \\average transmit power.}
        \label{fig:sepuser_Rate_Ptrans}
\end{minipage}
\end{figure}

Figures \ref{fig:sepuser_IRee_Ptrans} and \ref{fig:sepuser_Rate_Ptrans} illustrates the average IR-EE versus the average transmit power, and the average achievable rate versus the average transmit power, respectively. The curves are obtained by solving Problem \ref{prob:multiobj_WIPTsepuser} for $\omega_2=0$ and $0\leq\omega_j\leq1, j=1,3$, where the value of $\omega_j$ is uniformly varied with a step size of 0.01 such that $\sum_j\omega_j=1$. Without the objective of EH-EE maximization, the power allocation policy in the proposed EE algorithm is designed for IR-EE maximization and transmit power minimization. For the proposed EE algorithm, we can observe from Figures \ref{fig:sepuser_IRee_Ptrans} and \ref{fig:sepuser_Rate_Ptrans} that for small transmit power, both IR-EE and the achievable rate grow monotonically as the transmit power ascends from zero. Figure \ref{fig:sepuser_IRee_Ptrans} shows that IR-EE approaches to its maximum point at a very small transmit power. Figure \ref{fig:sepuser_Rate_Ptrans} shows the corresponding achievable rate which remains at a low level due to the small transmit power. In contrast, for the baseline scheme, we can see a bell-shaped trend of IR-EE in Figure \ref{fig:sepuser_IRee_Ptrans} and a monotonically ascending trend of the achievable rate in Figure \ref{fig:sepuser_Rate_Ptrans} with the increasing transmit power. In the small transmit power regime, IR-EE and data rate behave similarly as in the EE algorithm. However, in the high transmit power regime, the logarithmical growth of data rate is slower than the linear increment of the transmit power, which leads to energy-inefficient, i.e., IR-EE declines. Besides, when $N_\mathrm{T}=8$, the system performance shows a reduction on IR-EE and an growth on rate compared to the case of 4 transmitting antennas. This is because the achievable rate is improved by exploiting extra degrees of freedom offered by more transmitting antennas. However, this improvement cannot compensate the increment of antenna power consumption. Thus, a lower IR-EE is resulted.

\begin{figure}[htb]
\begin{minipage}[t]{0.5\textwidth}
        \centering
        \includegraphics[width=\textwidth,height=\textwidth]{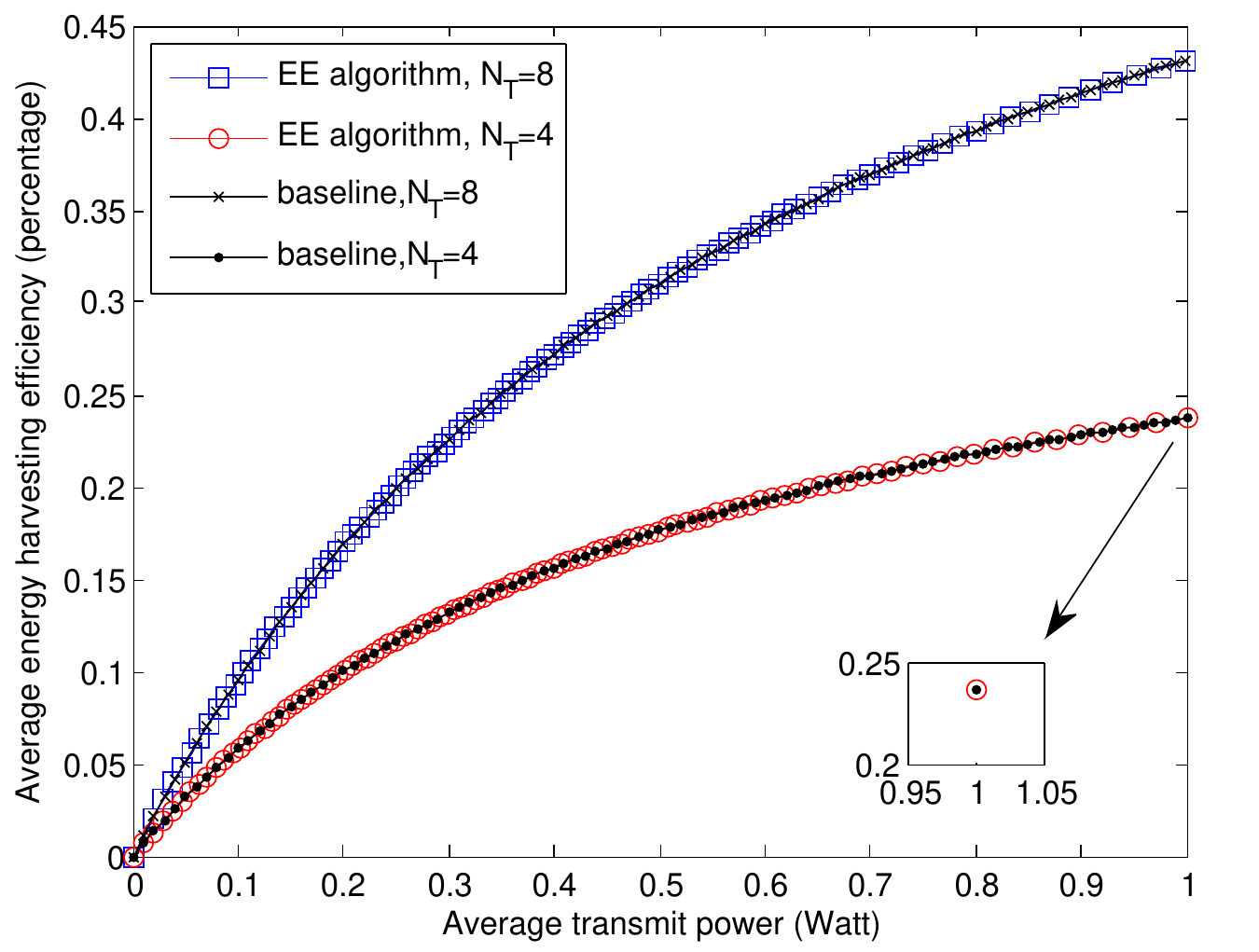}
        \caption{Average EH-EE versus \\average transmit power.}
        \label{fig:sepuser_EHee_Ptrans}
\end{minipage}
\begin{minipage}[t]{0.5\textwidth}
        \centering
        \includegraphics[width=\textwidth,height=\textwidth]{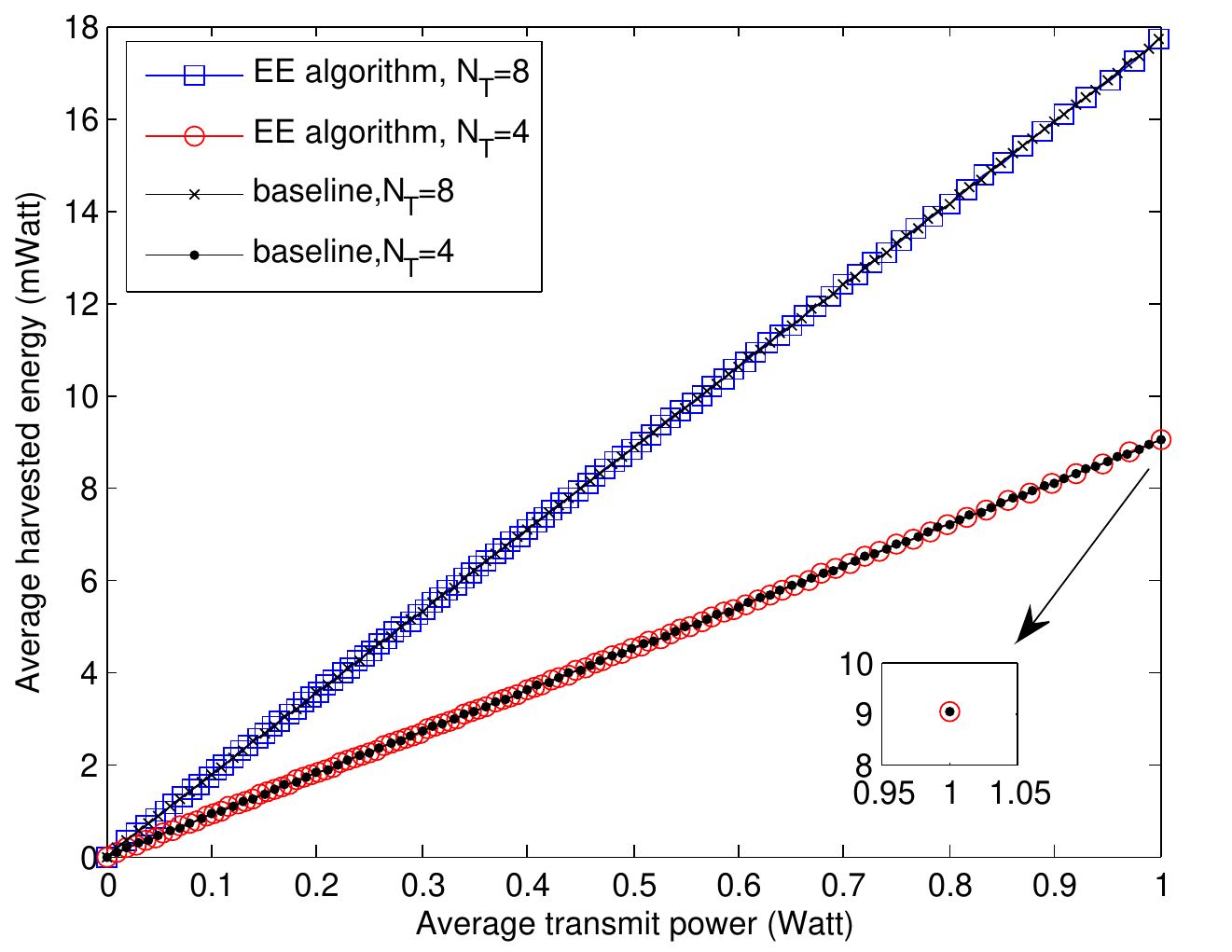}
        \caption{Average harvested energy \\versus average transmit \\power.}
        \label{fig:sepuser_Pharv_Ptrans}
\end{minipage}
\end{figure}

Figures \ref{fig:sepuser_EHee_Ptrans} and \ref{fig:sepuser_Pharv_Ptrans} depict the average EH-EE versus the average transmit power, and average harvested energy versus the average transmit power, respectively. Similarly, the curves are obtained by solving the MOOP for $\omega_1=0$ and $0\leq\omega_j\leq1, j=2,3$, where the value of $\omega_j$ is uniformly varied with a step size of $0.01$ such that $\sum_j\omega_j=1$. With no concerns on the objective of IR-EE maximization, the power allocation policy is designed for EH-EE maximization and transmit power minimization. It can be observed that both EH-EE and harvested energy are growing with the increasing transmit power. Especially, the curves of the proposed EE algorithm and the baseline scheme overlap, which means that the maximal EH-EE and the maximal harvest energy are simultaneously obtained by the same amount of transmit power. This is thanks to the linear relationship between the harvested power and the transmit power. In terms of the comparison for different number of antennas, a better performance on both EH-EE and harvested energy is showed for $8$ transmitting antennas. In particular, for $N_\mathrm{T}=8$, we see that EH-EE increases faster with the increasing transmit power than the case of $N_\mathrm{T}=4$. This implies a more efficient and effective power transfer is achieved by using more transmitting antennas.

%% file: 4_WIPT_secure.tex
In this chapter, we study resource allocation algorithm design for power-efficient SWIPT in secure communication systems with power splitting receivers. In a multiuser system, the transmitter supports SWIPT to the desired receivers and power transfer to roaming receivers. In particular, the roaming receivers are potential eavesdroppers, thus, artificial noise is applied to facilitate secure communication. Under the consideration of system QoS, an optimization problem is formulated for minimizing the total transmit power by jointly optimizing the beamforming vectors, power splitting ratios at the desired receivers, and the power of artificial noise. We proposed a power-efficient resource allocation algorithm which enables the dual use of artificial noise for supplying power transfer and guaranteeing communication security. The non-convex problem is transformed into SDP and solved by SDP relaxation. The global optimum can be achieved. Simulation results illustrate a significant power saving via the proposed algorithm.

\section{System Model}\label{sect:system model}

We consider a multiuser communication system with SWIPT in downlink scenario. The system consists of one transmitter and two types of receivers, namely, desired receivers and roaming receivers, cf. Figure \ref{fig:system_model}. The transmitter equipped with $N_\mathrm{T}$ antennas serves  $K$ desired receivers and $M$ roaming receivers. The desired receivers are single antenna devices with low computational capability. They exploit the received RF signal for both information decoding and energy harvesting. On the other hand, the roaming receivers are wireless terminals with $N_{\mathrm{R}}$ antennas ($N_\mathrm{T}>N_\mathrm{R}$). In particular, they could belong to other communication systems and search for additional RF energy supply. Suppose they temporally connect to the transmitter for energy harvesting from signals radiated from the transmitter\footnote{A possible scenario of the considered system model is a cognitive radio setup. Specifically, the roaming receivers may be primary receivers which harvest energy from a secondary transmitter for extending the lifetime of the primary network.}. However, it is possible that the roaming receivers eavesdrop the information signals deliberately. As a result, the roaming receivers are potential eavesdroppers which should be taken into account for secure communication.
 \begin{figure}[htb]
 \centering
\includegraphics[scale=0.75]{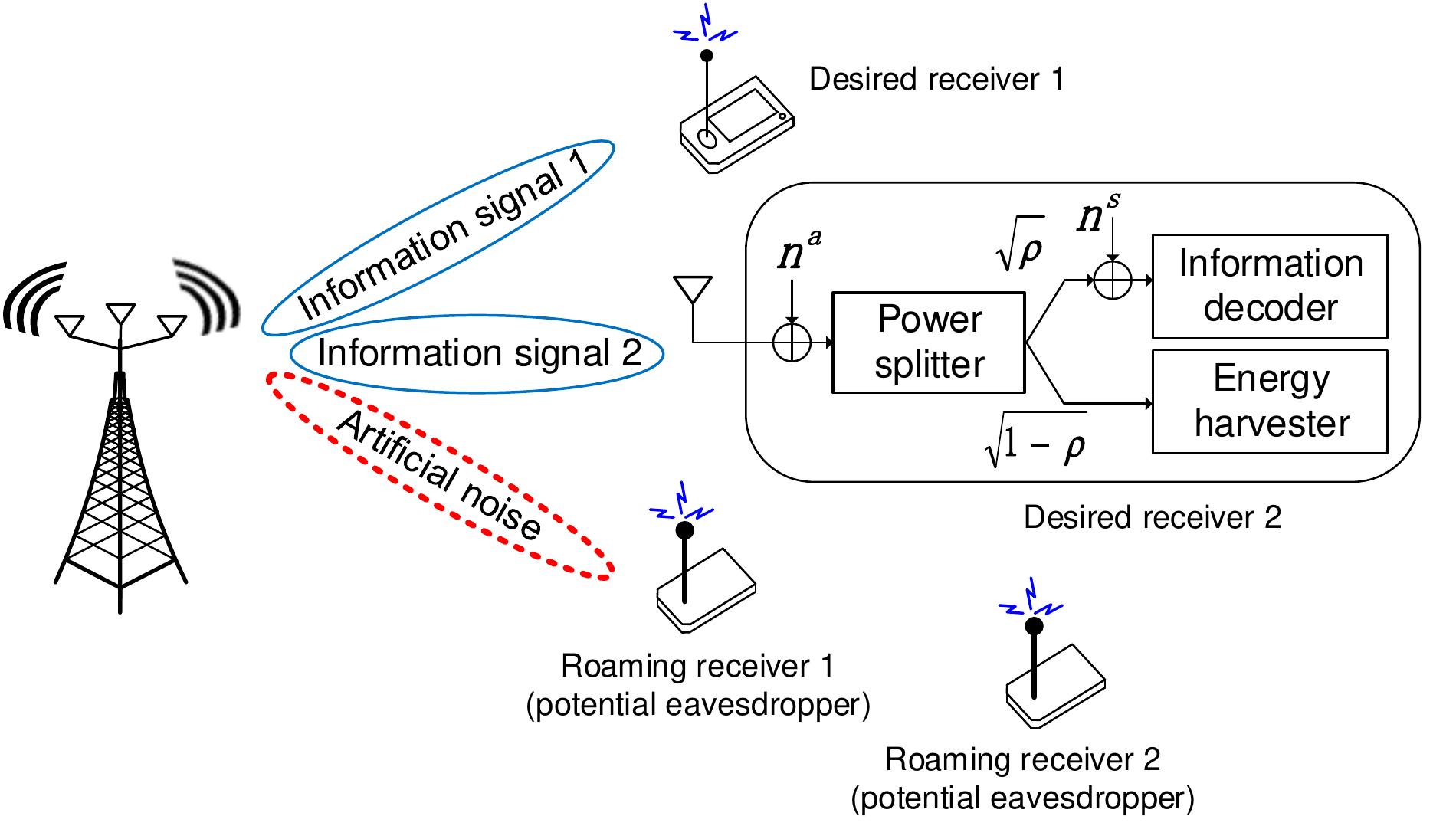}
 \caption{Multiuser SWIPT system with power splitting receivers.}
 \label{fig:system_model}
\end{figure}

We assume time-slot-based transmission. In each scheduling time slot, $K$ independent precoded signal streams are transmitted simultaneously to $K$ desired receivers. Specifically, a dedicated beamforming vector, $\mathbf{w}_k\in\mathbb{C}^{N_{\mathrm{T}}\times1}$, is assigned to each desired receiver to facilitate information transmission.  At the same time,  the  messages intended for the desired receivers may be overheard by roaming receivers since all receivers are in the range of service coverage. In order to provide secure communication, artificial noise is generated and transmitted concurrently to interfere the reception of the roaming receivers. As a result, the transmit signal $\mathbf{x}\in\mathbb{C}^{N_{\mathrm{T}}\times 1}$, is composed of the $K$ desired information signals and the artificial noise, which is given as
\begin{eqnarray}
\mathbf{x}&=&\sum_{k=1}^K\mathbf{w}_k s_k+\mathbf{v},
\end{eqnarray}
where $s_k\in\mathbb{C}$ is the signal to the $k$th desired receiver, $k=1,\dots,K$. Without loss of generality, we assume $\est{{\abs{s_k}}^2}=1,\forall k$. $\mathbf{v}\in\mathbb{C}^{N_{\mathrm{T}}\times1}$ is the artificial noise vector  generated by the transmitter to degrade the quality of the received signal of the potential eavesdroppers. In particular, we model the artificial noise vector as $\mathbf{v}\sim\cal{CN}(\mathbf{0},\mathbf{V})$ with zero mean and covariance matrix $\mathbf{V}=\est{\mathbf{v}\mathbf{v}^H}$.

We focus on a frequency flat fading channel and a TDD system. For the transmitter, perfect CSI of all receivers can be obtained by channel reciprocity and handshaking signals. The received signals at the desired receivers and the roaming receivers are given by
\begin{eqnarray}
y_{k}&=&\mathbf{h}_k^H\mathbf{x}+n_k^{\mathrm{a}},\,\,  \forall k\in\{1,\dots,K\},\\
\mathbf{y}_{\mathrm{I}_m}&=&\mathbf{G}_m^H\mathbf{x}+\mathbf{n}^\mathrm{a}_m,\,\,  \forall m\in\{1,\dots,M\},
\end{eqnarray}
where $\mathbf{h}_k\in\mathbb{C}^{N_{\mathrm{T}}\times1}$ denotes the channel vector between the transmitter and desired receiver $k$. The channel matrix between the transmitter and roaming receiver $m$ is denoted by $\mathbf{G}_m\in\mathbb{C}^{N_{\mathrm{T}}\times N_{\mathrm{R}}}$. $\mathbf{h}_k$ and $\mathbf{G}_m$ capture the joint effects of multipath fading and path loss. $n_k^{\mathrm{a}}\sim{\cal CN}(0,\sigma_{\mathrm{ant}}^2)$ and $\mathbf{n}^\mathrm{a}_m\sim{\cal CN}(\mathbf{0},\sigma_{\mathrm{ant}}^2\mathbf{I}_{N_{\mathrm{R}}})$ are additive white Gaussian noises (AWGN) caused by the thermal noises in the antennas of the desired receivers and the roaming receivers, respectively.

We assume a power splitting structure \cite{JR:WIP_receiver} is adopted in both desired receivers and roaming receivers. In particular, at the receiver RF front-end, the received signal at desired receiver $k$ is divided into two power streams where $100\times\rho_k \%$ are used for decoding information and the remaining $100\times(1-\rho_k)\%$ are used for harvesting energy, cf. Figure \ref{fig:system_model}. Here, $0\leq\rho_k\leq1$ is the power splitting ratio of desired receiver $k$. Similarly, power splitting is also performed at the roaming receivers for energy harvesting and information decoding. We assume that all receivers have enough energy for information decoding at the current time instant which is independent of the amount of harvested energy. The harvested energy is stored in battery and used to support normal operation of the receiver in the future.  Since a portion of received power is reserved for energy harvesting, the equivalent received signal model for information decoding at desired receiver $k$ can be written as
\begin{eqnarray}
y_{k}^{\mathrm{ID}}&=&\sqrt{\rho_k}(\mathbf{h}_k^H\mathbf{x}+n_k^{\mathrm{a}})+n^{\mathrm{s}}_k,\,\,
\end{eqnarray}
where  $n^{\mathrm{s}}_k$ is AWGN with zero mean and variance $\sigma_{\mathrm{s}}^2$ caused by signal processing, cf. Figure \ref{fig:system_model}. We assume that the signal processing noise variances are the same for all receivers in this chapter.

\section{Problem Formulation}
\label{sect:cross-Layer_formulation}
In this section, we first introduce the adopted QoS metrics for SWIPT system design. Then, the power-efficient resource allocation algorithm design is formulated as a non-convex optimization problem.

The achievable rate (bit/s/Hz) between the transmitter and desired receiver $k$ is given by
\begin{eqnarray}
R_{k}\hspace*{-1mm}&=&\hspace*{-1mm}\log_2(1+\Gamma_{k}),\quad\mathrm{where}\\
\Gamma_{k}\hspace*{-1mm}&=&\hspace*{-1mm}\frac{\rho_k\abs{\mathbf{h}_k^H\mathbf{w}_k}^2}{\rho_k\Big(\sum\limits_
{\substack{j\neq k}}^K\abs{\mathbf{h}_k^H\mathbf{w}_j}^2+\Tr(\mathbf{h}_k\mathbf{h}_k^H\mathbf{V})+\sigma_{\mathrm{ant}}^2\Big)+\sigma_{\mathrm{s}}^2}
\end{eqnarray}
is the SINR at desired receiver $k$.

On the other hand, to guarantee communication security, the roaming receivers are treated as potential eavesdroppers who attempt to decode the messages transmitted to desired receivers. Thereby, we consider the worst case scenario that roaming receivers are high computational capable eavesdroppers. We assume that roaming receiver $m$ is able to perform successive interference cancellation (SIC) to remove all multiuser interference before decoding the signal of receiver $k$. Therefore, the data rate between the transmitter and roaming receiver $m$ for decoding the signal of desired receiver $k$ can be represented as
\begin{eqnarray}\label{eqn:Capacity_eve}
\hspace*{-6mm}R^\mathrm{eav}_{m,k}\hspace*{-2mm}&=&\hspace*{-2mm}\log_2\det\Big(\mathbf{I}_{N_{\mathrm{R}}}
\hspace*{-0.5mm}+\hspace*{-0.5mm}
\gmat{\Delta}_{m}^{-1}
\rho^\mathrm{R}_m\mathbf{G}_m^H\mathbf{w}_k\mathbf{w}_k^H\mathbf{G}_m\Big),\quad\mathrm{where}\\
\gmat{\Delta}_{m}\hspace*{-2mm}&=&\hspace*{-2mm}\rho^\mathrm{E}_m\gmat{\Sigma}_{m}+
\sigma_\mathrm{s}^2\mathbf{I}_{N_{\mathrm{R}}},\quad \label{eqn:Sigma} \gmat{\Sigma}_{m}\hspace*{-2mm}\,\,\,=\,\mathbf{G}_m^H\mathbf{V}\mathbf{G}_m+\sigma_{\mathrm{ant}}^2\mathbf{I}_{N_{\mathrm{R}}}.
\end{eqnarray}
$0\le \rho^\mathrm{R}_m\le 1$ is the power splitting ratio of roaming receiver $m$. $\gmat{\Sigma}_{m}$ is the interference-plus-noise covariance matrix for roaming receiver $m$. In practice, the roaming receiver can be malicious and devote all the received energy for information decoding. Thus, the data rate in (\ref{eqn:Capacity_eve}) is bounded above by
\begin{eqnarray}\label{eqn:Capacity_eve_up}
\overline{R^\mathrm{eav}_{m,k}}=\log_2\det\Big(\mathbf{I}_{N_{\mathrm{R}}}\hspace*{-0.5mm}+\hspace*{-0.5mm}(\gmat{\Sigma}_{m}+\sigma_\mathrm{s}^2\mathbf{I}_{N_{\mathrm{R}}})^{-1}
\mathbf{G}_m^H\mathbf{w}_k\mathbf{w}_k^H\mathbf{G}_m\Big)
\end{eqnarray}
which is obtained by setting $\rho^{\mathrm{E}}_m=1$ in  (\ref{eqn:Capacity_eve}).

Consequently, considering this worst case scenario, the maximum achievable secrecy rate of desired receiver $k$ is given by
\begin{eqnarray}\label{eqn:secure_cap}
R^\mathrm{sec}_k&=&\Big[R_{k}-\underset{m=1,\dots,M}{\max}\,\Big\{\overline{R^\mathrm{eav}_{m,k}}\Big\}\Big]^+.
\end{eqnarray}
\begin{Remark}
We note that the results of this work are also applicable  to the case of roaming receivers (potential eavesdroppers) employing single user detectors by modifying the term $\gmat{\Sigma}_{m}$ in (\ref{eqn:Sigma}) accordingly.
\end{Remark}

In terms of power transfer, both the information signal and the artificial noise serve as energy source for the receivers due to wireless broadcasting property. The total harvested energy of desired receiver $k$ is given by
\begin{eqnarray}
E_{k}=\eta(1-\rho_k)\Big(\sum_{j=1}^K\abs{\mathbf{h}_k^H\mathbf{w}_j}^2+\Tr(\mathbf{h}_k\mathbf{h}_k^H\mathbf{V})+\sigma_{\mathrm{ant}}^2\Big),
\end{eqnarray}
where $0\leq\eta\leq1$ is the energy conversion efficiency, indicating the efficiency of converting the received RF energy to electrical energy for storage. We assume that it is a constant and is identical for all receivers.

Similarly, the total amount of energy harvested by roaming receiver $m$ is given by
\begin{eqnarray}
E^\mathrm{R}_m=\eta_m(1-\rho^\mathrm{R}_m)\Big(\sum_{k=1}^K\Tr(\mathbf{G}_m^H\mathbf{w}_k\mathbf{w}_k^H\mathbf{G}_m)+\Tr(\mathbf{G}_m\mathbf{G}_m^H\mathbf{V})+N_{\mathrm{R}}\sigma_{\mathrm{ant}}^2\Big).
\end{eqnarray}

The system design objective is to minimize the total transmit power while providing system QoS on secure communication and power transfer. The power efficient resource allocation algorithm design is formulated as an optimization problem which is given by
\begin{Prob}\label{prob:WIPT_secure}
\begin{eqnarray}
\underset{ \mathbf{V}\in \mathbb{H}^{N_{\mathrm{T}}},\mathbf{w}_k,\rho_k}{\mino}\,\, &&\sum_{k=1}^K\norm{\mathbf{w}_k}^2+\Tr(\mathbf{V})\notag\\
\mathrm{subject\,\,to}\,\,&&\mathrm{C1}:\,\,\Gamma_{k}\ge \Gamma^\mathrm{req}_k,\,\, \forall k, \notag\\
&&\mathrm{C2}:\,\,\overline{R^\mathrm{eav}_{m,k}}\le R^\mathrm{max}_{m,k},\,\, \forall k,\forall m,\notag\\
&&\mathrm{C3}:\,\,E_{k}\ge P^{\mathrm{req1}}_{k},\,\, \forall k,\notag\\
&&\mathrm{C4}:\,\,\eta\Big(\sum_{k=1}^K\Tr(\mathbf{G}_m^H\mathbf{w}_k\mathbf{w}_k^H\mathbf{G}_m)+\Tr(\mathbf{G}_m\mathbf{G}_m^H\mathbf{V})+N_{\mathrm{R}}\sigma_{\mathrm{ant}}^2\Big)\ge P^\mathrm{req2}_m,\,\, \forall m, \notag\\
&&\mathrm{C5}:\,\,0\leq\rho_k\leq1,\,\, \forall k, \notag\\
&&\mathrm{C6}:\,\,\mathbf{V}\succeq \mathbf{0}.\notag
\end{eqnarray}
\end{Prob}
\noindent Constraint C1 indicates that SINR at desired receiver $k$ is required to be larger than a given threshold, $\Gamma^\mathrm{req}_k>0$. Since any desired receiver could  be chosen as an eavesdropping target of roaming receiver $m$,
the upper limit $R^\mathrm{max}_{m,k}$ is imposed in C2 to restrict the eavesdropping rate of roaming receiver $m$ when it  attempts to decode the message of desired receiver $k$.  Notice that in practice we are interested in the case of $R_{k}> R^\mathrm{max}_{m,k}$ for secure communication, which means $R^\mathrm{sec}_k \ge R_{k}-\underset{m=1,\dots,M}{\max}\,\{R^\mathrm{max}_{m,k}\}= \log_2(1+\Gamma^\mathrm{req}_k)-\underset{m=1,\dots,M}{\max}\,\{R^\mathrm{max}_{m,k}\}>0$. In particular, the parameters
$\Gamma^\mathrm{req}_k$ and $R^\mathrm{max}_{m,k}$ can be selected to provide flexibility in designing power-efficient resource allocation algorithms for different applications. Constants $P_k^\mathrm{req1}$ and $P_m^\mathrm{req1}$ in constraints C3 and C4 specify the minimum required harvested energy at desired receiver $k$ and roaming receiver $m$, respectively. The physical meaning of constraint C4 is that the transmitter only guarantees the minimum required harvested power at roaming receiver $m$ when it does not intend to eavesdrop, i.e., $\rho^\mathrm{R}_m=0$. Constraint C5 implies the physical constraint for the power splitter. In addition, we assume that the power splitter is a passive device which does not consume or gain any power when splitting the received signal. Constraint C6 and $\mathbf{V}\in \mathbb{H}^{N_{\mathrm{T}}}$ ensure that the covariance matrix $\mathbf{V}$ is a positive semi-definite Hermitian matrix.

\section{Power-Efficient Resource Allocation Algorithm Design}
\label{sect:solution}
It can be observed that Problem \ref{prob:WIPT_secure} is non-convex due to constraints C1 and C2. To overcome the non-convexity of C2,  we recast Problem \ref{prob:WIPT_secure} as SDP. We first replace $\mathbf{w}_k\mathbf{w}_k^H$ by $\mathbf{W}_k=\mathbf{w}_k\mathbf{w}_k^H$ and rewrite C2 as
\begin{eqnarray}\label{eqn:original_c2}
\mathrm{C2:}&&\det\big(\mathbf{I}_{N_{\mathrm{R}}}+\mathbf{Q}_{m}^{-1}\mathbf{G}_m^H\mathbf{W}_k\mathbf{G}_m\big)\le\psi_{m,k},\, \forall m,k, \quad\text{where}\\
&&\mathbf{Q}_{m}=\mathbf{G}_m^H\mathbf{V}\mathbf{G}_m+
(\sigma_{\mathrm{ant}}^2+\sigma_{\mathrm{s}}^2)\mathbf{I}_{N_{\mathrm{R}}}\succ\mathbf{0}.\notag
\end{eqnarray}
$\psi_{m,k}$ is an auxiliary constant that $\psi_{m,k}=2^{R^\mathrm{max}_{m,k}}$,  $\psi_{m,k}> 1$
for $R^\mathrm{max}_{m,k}> 0$.  Then, we introduce the following proposition to simplify the considered problem.
\begin{proposition}\label{prop:relaxed_c2}
For $R^\mathrm{max}_{m,k}>0,\forall m,k$, the following implication on constraint C2 holds:
\begin{eqnarray}\label{eqn:det_to_matrix}\notag
\mathrm{C2}\Rightarrow\overline{\mathrm{C2}}\,\,\mathrm{: }\quad \mathbf{G}_m^H\mathbf{W}_k\mathbf{G}_m\preceq (\psi_{m,k}-1)\mathbf{Q}_{m},\,\, \forall m,k.
\end{eqnarray}
\end{proposition}
$\overline{\mathrm{C2}}$ is a linear matrix inequality (LMI) constraint. Specifically, constraints $\overline{\mathrm{C2}}$ and ${\mathrm{C2}}$ are equivalent if $\Rank(\mathbf{W}_k)=1,\forall k$.
\begin{proof}
Please refer to Appendix \ref{app:Prop_relaxed_c2}.
\end{proof}

Now, we apply Proposition \ref{prop:relaxed_c2} to Problem \ref{prob:WIPT_secure} by replacing constraint $\mathrm{C2}$ with $\overline{\mathrm{C2}}$.  Then, the reformulated problem can be written as
\begin{Prob}\label{prob:rank_one}
\begin{eqnarray}
\underset{\mathbf{W}_k, \mathbf{V}\in \mathbb{H}^{N_{\mathrm{T}}}, \rho_k}{\mino}\,\, &&\sum_{k=1}^K\Tr(\mathbf{W}_k)+\Tr(\mathbf{V})\notag\\
\mathrm{subject\,\,to} &&\mathrm{C1}:\,\,\frac{1}{\Gamma^\mathrm{req}_k}\Tr(\mathbf{h}_k\mathbf{h}_k^H\mathbf{W}_k)-\sum\limits_{\substack{j\neq k}}^K\Tr(\mathbf{h}_k\mathbf{h}_k^H\mathbf{W}_j)-\Tr(\mathbf{h}_k\mathbf{h}_k^H\mathbf{V})\ge\sigma_{\mathrm{ant}}^2+\frac{1}{\rho_k}\sigma_{\mathrm{s}}^2,\,\, \forall k, \notag\\
&&\overline{\mathrm{C2}}:\,\,\mathbf{G}_m^H\mathbf{W}_k\mathbf{G}_m\preceq (\psi_{m,k} -1)\mathbf{Q}_{m},\,\, \forall m,k,\notag\\
&&\mathrm{C3}:\,\,\Tr(\mathbf{h}_k\mathbf{h}_k^H(\mathbf{V}+\sum_{j=1}^K\mathbf{W}_j))\ge\frac{P^\mathrm{req1}_k}{\eta(1-\rho_k)}-\sigma_{\mathrm{ant}}^2,\forall k,\notag\\
&&\mathrm{C4}:\,\,\Tr(\mathbf{G}_m^H(\mathbf{V}+\sum_{k=1}^K\mathbf{W}_k)\mathbf{G}_m)\ge\frac{P^\mathrm{req2}_m}{\eta}-N_{\mathrm{R}}\sigma_{\mathrm{ant}}^2,\,\, \forall m,\notag\\
&&\mathrm{C5}:\,\,0\leq\rho_k\leq1,\,\, \forall k, \notag\\
&&\mathrm{C6}:\,\, \mathbf{V}\succeq \mathbf{0},\notag\\
&&\mathrm{C7}:\,\, \mathbf{W}_k\succeq \mathbf{0},\,\, \forall k, \notag\\
&&\mathrm{C8}:\,\, \Rank(\mathbf{W}_k)=1,\,\, \forall k. \notag
\end{eqnarray}
\end{Prob}
\noindent Constraints C7, C8, and $\mathbf{W}_k\in\mathbb{H}^{N_{\mathrm{T}}},\forall k$, are imposed to guarantee that $\mathbf{W}_k=\mathbf{w}_k\mathbf{w}_k^H$ holds for the optimal solution. In general, replacing constraint
C2  by $\overline{\mathrm{C2}}$  leads to a larger feasible solution set for the problem, cf. Proposition \ref{prop:relaxed_c2}.  However, Problem \ref{prob:WIPT_secure} and \ref{prob:rank_one} are equivalent for $\Rank(\mathbf{W}_k)= 1$, $\forall k$. Thus, we focus on Problem \ref{prob:rank_one} in the following.

Although the new constraint $\overline{\mathrm{C2}}$ is an affine function with respect to the optimization variables,  it can be verified that Problem \ref{prob:rank_one} is still non-convex due to the combinatorial rank constraint in C8. In order to achieve an efficient resource allocation algorithm design, we adopt SDP relaxation approach. In particular, we relax constraint $\mathrm{C8}$ by removing it from the problem formulation, such that Problem \ref{prob:rank_one} becomes a convex problem.  The SDP relaxed problem is given by
\begin{Prob}\label{prob:sdp_relaxation}
\begin{eqnarray}
\underset{\mathbf{W}_k, \mathbf{V}\in \mathbb{H}^{N_{\mathrm{T}}}, \rho_k}{\mino}\,\, &&\sum_{k=1}^K\Tr(\mathbf{W}_k)+\Tr(\mathbf{V})\notag\\
\mathrm{subject\,\,to} &&\mathrm{C1},\,\overline{\mathrm{C2}},\,\mathrm{C3},\,\mathrm{C4},\,\mathrm{C5},\,\mathrm{C6},\,\mathrm{C7}.
\end{eqnarray}
\end{Prob}
\noindent We note that Problem \ref{prob:sdp_relaxation} can be solved efficiently by numerical solvers such as CVX \cite{website:CVX}. Notably, if the optimal solution $\mathbf{W}_k^*$ of Problem \ref{prob:sdp_relaxation} admits a rank-one matrix, then Problem \ref{prob:WIPT_secure}, \ref{prob:rank_one}, and \ref{prob:sdp_relaxation} share the same optimal solution and the same optimal objective value, i.e., the global optimum is achieved.

Now, we introduce the following theorem to reveal the tightness of the SDP relaxation.
\begin{Thm}\label{thm:rankone_condition}
Suppose the optimal solution of Problem \ref{prob:sdp_relaxation} is denoted by $\{\mathbf{W}_k^*,\mathbf{V}^*,\rho_k^*\}$, ${\Gamma}^\mathrm{req}_k>0$, and $R^\mathrm{max}_{m,k}>0$. If $\exists k:\Rank(\mathbf{W}^*_k)>1$, then we can construct another solution for Problem \ref{prob:sdp_relaxation}, denoted as  $\{\mathbf{\widetilde  W}_k ,\mathbf{\widetilde  V},\widetilde \rho_k\}$, which not only achieves the same objective value as $\{\mathbf{W}_k^*,\mathbf{V}^*,\rho_k^*\}$, but admits a rank-one matrix, i.e.,  $\Rank(\mathbf{\widetilde W}_k)=1,\forall k$.
\end{Thm}
\begin{proof}
Please refer to Appendix \ref{app:thm_rankone_condition} for a proof of Theorem \ref{thm:rankone_condition} and a method for constructing $\{\mathbf{\widetilde  W}_k ,\mathbf{\widetilde  V},\widetilde \rho_k\}$ with $\Rank(\mathbf{\widetilde W}_k)=1,\forall k$.
\end{proof}

In other words, by applying Theorem \ref{thm:rankone_condition} and Proposition \ref{prop:relaxed_c2}, the global optimal solution of the original problem is obtained.

\section{Results}
In this section, we demonstrate the system performance of the proposed power efficient resource allocation design by simulation results. In particular, we solve Problem \ref{prob:WIPT_secure} for  different channel realizations and show the corresponding average system performance.

\begin{table}[htb]
\caption{Simulation Parameters} \label{table:parameters3}
\centering
\begin{tabular}{ | l | l | } \hline
      Carrier center frequency  & 470 MHz\\ \hline
      Number of desired receiver & $K=3$ \\ \hline
      Number of roaming receiver & $M=2$   \\ \hline
      Number of receiving antenna & $N_R=2$ \\ \hline
      Antennas gain & 10 dBi \\ \hline
      Antenna noise power & $\sigma_\mathrm{ant}^2= -124$ dBm \\ \hline
      Signal processing noise power & $\sigma_\mathrm{s}^2= -23$ dBm \\ \hline
      Rician factor & 3 dB \\ \hline
      Reference distance        & 2 meters \\ \hline
      Maximum service distance & 50 meters \\ \hline
      Minimum required SINR & $\Gamma^\mathrm{req}$ \\ \hline
      Maximum data rate of roaming receivers & $R^\mathrm{max}_{m,k}=1$ bit/s/Hz\\ \hline
      Minimum required harvested power & $P^\mathrm{req1}_k=P^\mathrm{req2}_m=0$ dBm \\ \hline
      Energy conversion efficiency & $\eta=0.5$ \\ \hline
    \end{tabular}
\end{table}

The simulation parameters are summarized in Table \ref{table:parameters3}. We adopt the TGn path loss model \cite{report:tgn}. In particular, we assume a carrier center frequency of $470$ MHz. Assume $3$ desired receivers and $2$ roaming receivers (potential eavesdroppers), which are uniformly distributed in the range between a reference distance of $2$ meters and a maximum distance of 50 meters. Each roaming receiver is equipped with $N_{\mathrm{R}}=2$ antennas. The multipath fading coefficients are modeled as independent and identically distributed Rician fading with Rician factor 3 dB. We set the minimum required SINR of all desired receivers identical, i.e., $\Gamma^\mathrm{req}_k=\Gamma^\mathrm{req},\forall k$, the maximum tolerable rate of each roaming receiver is $R^\mathrm{max}_{m,k}=1$ bit/s/Hz, $\forall m,k$, and the minimum required harvested power for all receivers is $P^\mathrm{req1}_k=P^\mathrm{req2}_m=0$ dBm.  The energy conversion efficiency $\eta$ is $0.5$.  Antenna gain is $10$ dBi, and the antenna noise power is $\sigma_{\mathrm{ant}}^2= -124$ dBm at a temperature of $290$ Kelvin. We assume a $8$-bit uniform quantizer is employed in the analog-to-digital converter at the analog front-end of each receiver, which result in a signal processing noise of $\sigma_{\mathrm{s}}^2= -23$ dBm.

\begin{figure}[t]
 \centering
\includegraphics[width=\textwidth]{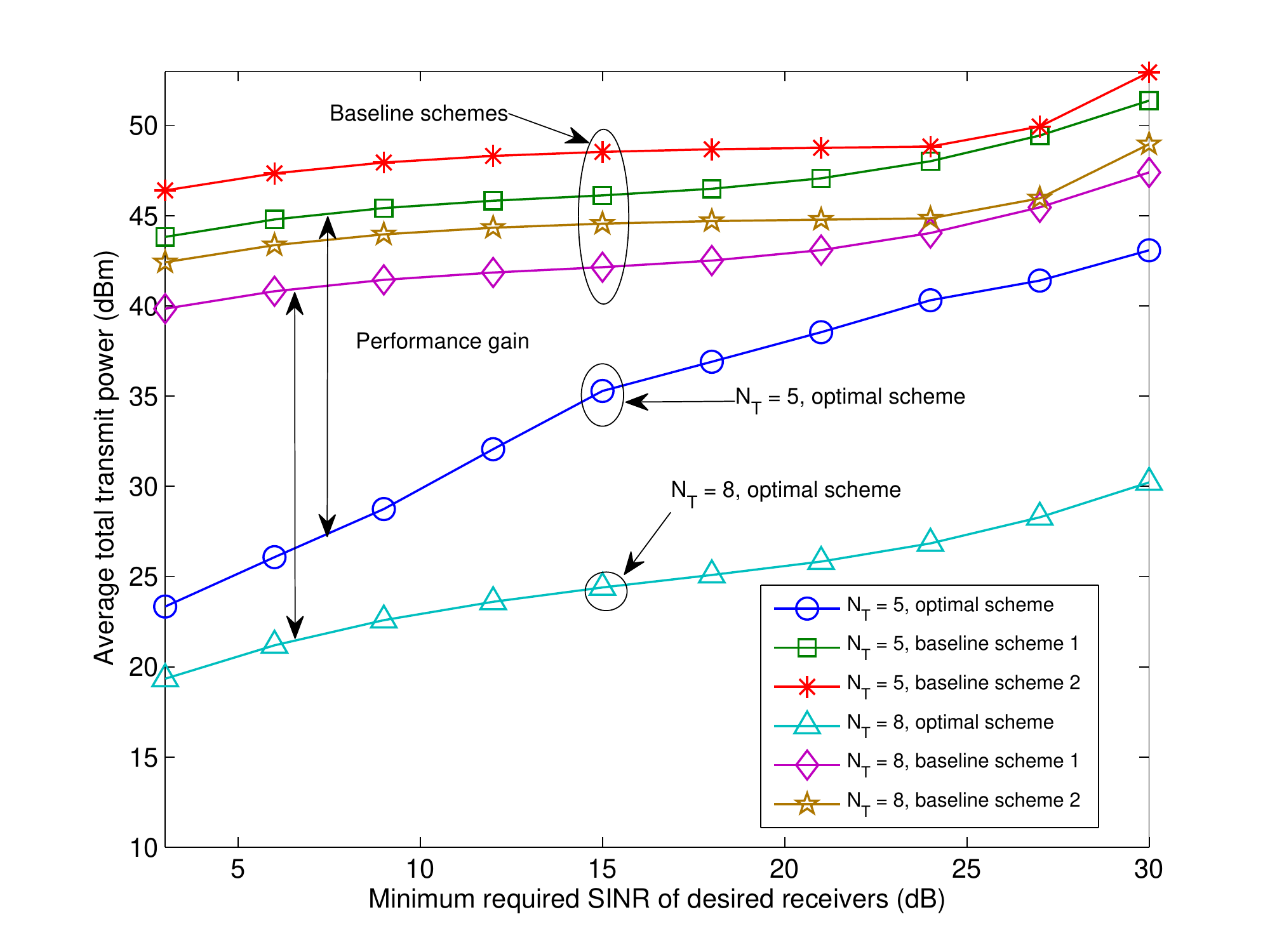}
 \caption{Average total transmit power versus the minimum required SINR of desired receivers for $N_{\mathrm{T}}=5$ and $N_{\mathrm{T}}=8$. The double-sided arrows indicate the power gains achieved by the optimal scheme compared to the baseline schemes.}
 \label{fig:pt_SINR}
\end{figure}
\begin{figure}[htb]
 \centering
\includegraphics[width=\textwidth]{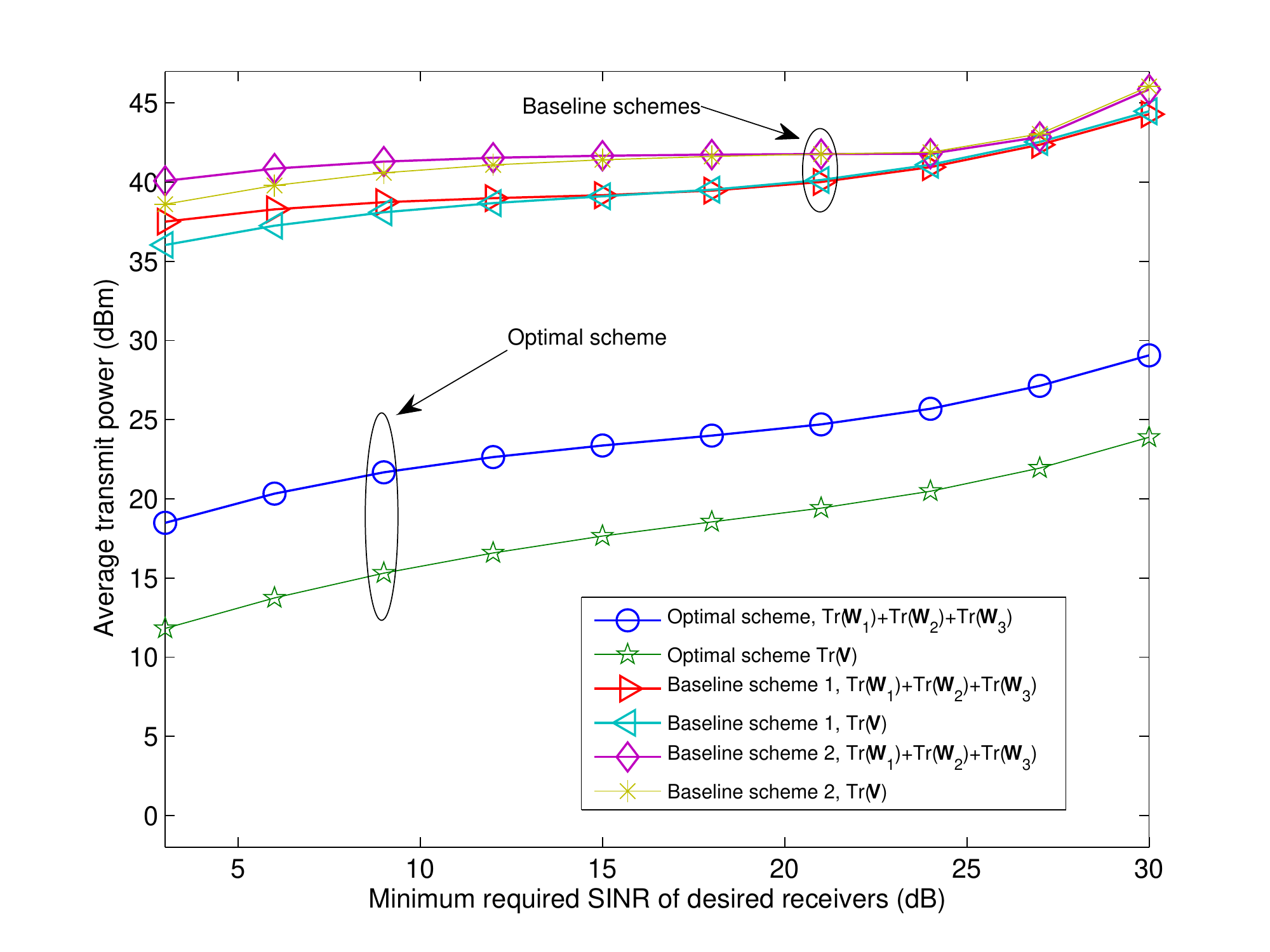}
 \caption{Average transmit power allocation for desired signal and artificial noise versus the minimum required SINR of desired receivers for $N_{\mathrm{T}}=8$.}
 \label{fig:sig_AN}
\end{figure}
\begin{figure}[htb]
 \centering
\includegraphics[width=\textwidth]{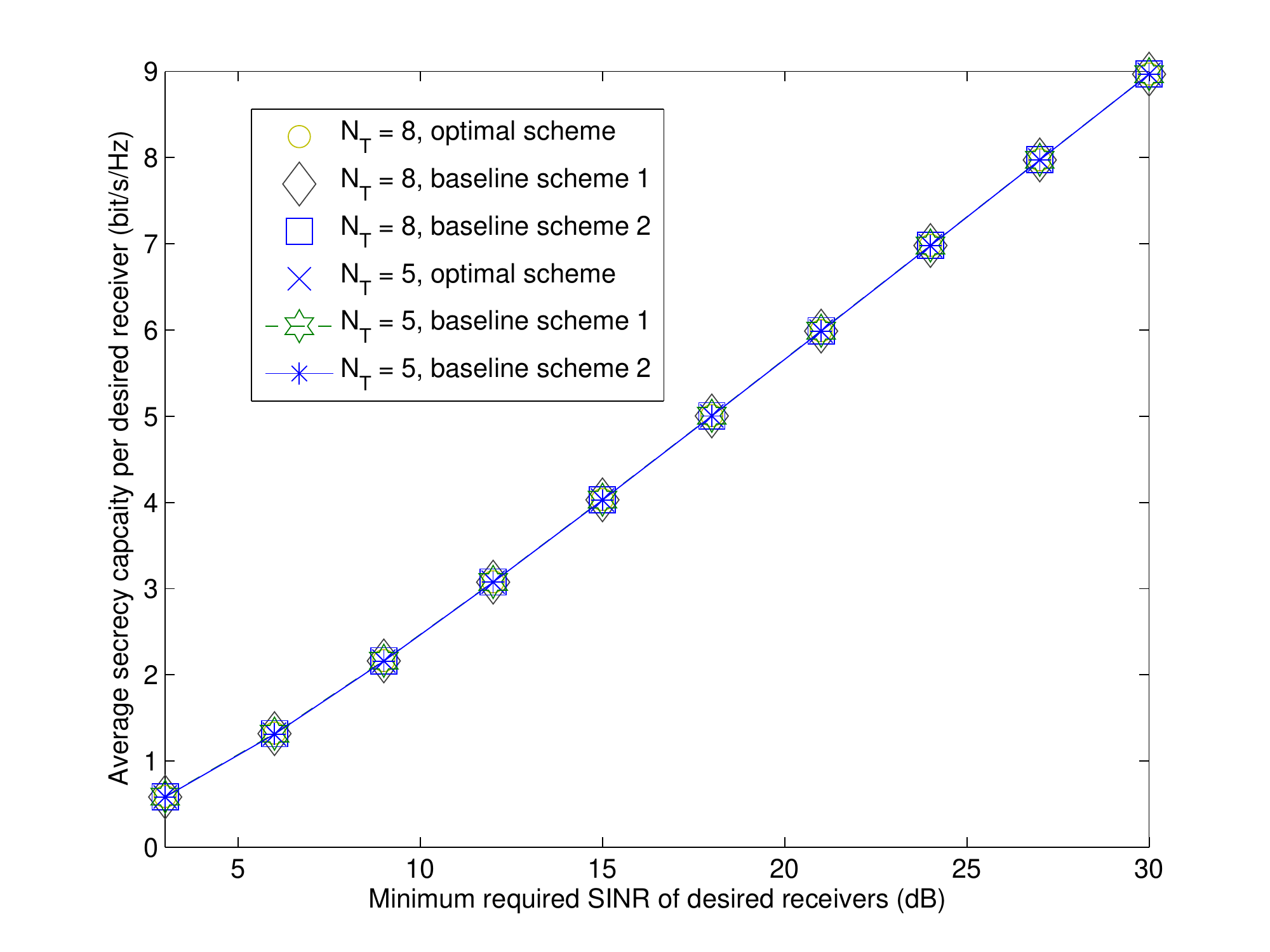}
 \caption{Average secrecy rate per desired receiver versus the minimum required SINR of desired receivers.}
 \label{fig:cap_SINR}
\end{figure}
\begin{figure}[htb]
 \centering
\includegraphics[width=\textwidth]{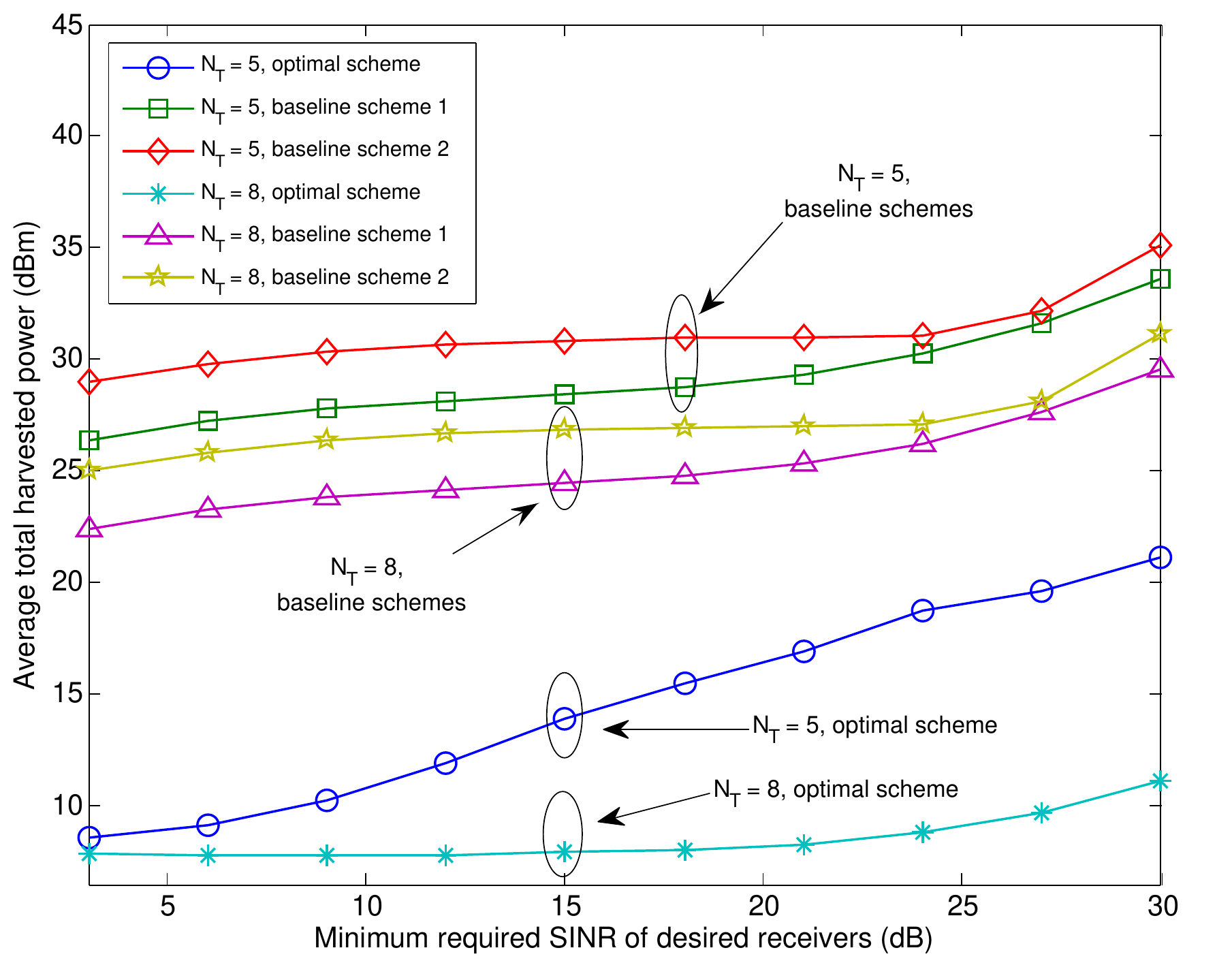}
 \caption{Average total harvested power versus the minimum required SINR of desired receivers.}
 \label{fig:hp_SINR}
\end{figure}

In Figure \ref{fig:pt_SINR}, we show the average total transmit power versus the minimum required SINR for different numbers of transmitting antenna and different resource allocation schemes. It can be observed that the total transmit power increases monotonically with an increasing SINR. The reason behind this is twofold. First, a higher transmit power of information signals is required to satisfy the more stringent requirement on SINR. Second, a higher amount of power also has to be allocated to the artificial noise to neutralize the increasing information leakage potential due to a stronger information signal, cf. Figure \ref{fig:sig_AN}. On the other hand, we see that a significant power saving can be achieved by the proposed optimal scheme when the number of antenna increase from $N_{\mathrm{T}}=5$ to $N_{\mathrm{T}}=8$. This is because more transmitting antennas provide extra degrees of freedom, which enables a more power-efficient resource allocation.

For comparison, we also present the performance of two simple suboptimal baseline schemes. For baseline scheme 1, zero-forcing beamforming is applied for information signal such that the desired receivers do not experience any multiuser interference. In particular, we calculate the eigenvalue decomposition of $\mathbf{H}_{-k}\mathbf{H}_{-k}^H=\mathbf{U}_k\gmat{\Sigma}_k\mathbf{U}^H_k$ for desired receiver $k$ where $\mathbf{H}_{-k}=[\mathbf{h}_1\,\ldots\,\mathbf{h}_{k-1}\,\mathbf{h}_{k+1}\,\ldots\,\mathbf{h}_K]$, $\mathbf{U}_k$ is an $N_\mathrm{T}\times N_\mathrm{T}$ unitary matrix, and $\gmat{\Sigma}_k$ is a diagonal matrix with ascending eigenvalues of  $\mathbf{H}_{-k}\mathbf{H}_{-k}^H$ as main diagonal elements. Then, we select $\mathbf{W}_k=q_{\mathrm{sub}_k}\mathbf{w}_{\mathrm{sub}_k}\mathbf{w}_{\mathrm{sub}_k}^H$, where $q_{\mathrm{sub}_k}\ge 0$ is a new scalar optimization variable and  $\mathbf{w}_{\mathrm{sub}_k}$ is the first column vector\footnote{In general, different column vectors with respect to the null space of $\mathbf{H}_{-k}\mathbf{H}_{-k}^H$ can be used as zero-forcing beamforming vector.  For algorithm computational simplicity, we select the first column vector corresponding to the minimum eigenvalue of matrix $\mathbf{H}_{-k}\mathbf{H}_{-k}^H$ as zero-forcing beamforming vector. } of $\mathbf{U}_k$ such that $\mathbf{H}_{-k}^H\mathbf{w}_{\mathrm{sub}_k}=\mathbf{0}$. In other words, the directions of the beamforming matrices are fixed for all desired receives. Then, we minimize the total transmit power by optimizing  $q_{\mathrm{sub}_k},\mathbf{V}$, and $\rho_k$ subject to the same constraints as in \ref{prob:sdp_relaxation}. We note that the zero-forcing beamforming matrix admits a rank-one structure. As for baseline scheme 2, it shares the same resource allocation policy  as baseline scheme 1 except that we set $\rho_k=0.5, \forall k$. It can be observed in Figure \ref{fig:pt_SINR} that the optimal scheme achieves  significant  power savings over the two baseline schemes. Notably, the performance gain of  the optimal scheme over the two baseline schemes  is further enlarged by more transmitting antennas. Since the optimal scheme can fully utilize the degrees of freedom for resource allocation. In contrast,  although multiuser interference is eliminated  in the two baseline schemes, the degrees of freedom for resource allocation in the baseline schemes are limited thus resulting in a higher transmit power. Furthermore, the performance gap between baseline scheme 1 and baseline scheme 2 reveals the performance gain on the optimization of the power splitting ratio.

Figure \ref{fig:sig_AN} depicts the average transmit power allocated to the information signal and the artificial noise for $N_{\mathrm{T}}=8$.  As shown, the power allocations for both information signal and artificial noise grow rapidly
with the increment of minimum SINR requirement. Besides, both the optimal scheme and the two baseline schemes indicate that a large portion of the total transmit power is allocated to the artificial noise. These results suggest that artificial noise generation is crucial for supporting communication security and WPT.

Figure \ref{fig:cap_SINR} plots the average secrecy rate per desired receiver with respect to the minimum required SINR of desired receivers for different numbers of transmitting antenna and different resource allocation schemes. The average achievable secrecy rate rises with the minimum required SINR since the eavesdropping rate of roaming receivers (potential eavesdroppers) is limited to $R^\mathrm{max}_{m,k}=1$ bit/s/Hz. Besides, all considered schemes are able to guarantee the QoS requirement on communication security (constraints C1 and C2) and achieve the same secrecy rate. However, the two baseline schemes achieve the same secrecy rate as the optimal scheme at the expense of a significantly higher transmit power, cf. Figure \ref{fig:pt_SINR}.

In Figure \ref{fig:hp_SINR}, we demonstrate the average total harvested power versus the minimum required SINR of desired receivers for different resource allocation schemes and different numbers of transmitting antenna. The average total
harvested power is computed by assuming the roaming receivers (potential eavesdroppers) do not eavesdrop. It is shown that the average total harvested power of all resource allocation schemes increase with $\Gamma^\mathrm{req}$. As more RF energy is available due to stronger information signals, cf. Figure \ref{fig:pt_SINR}. Besides, more energy is harvested in the two baseline schemes compared to the optimal scheme. The superior energy harvesting performance of the baseline schemes comes at the expense of a significant large transmit power. On the other hand, it can be observed that the average
total harvested power in the system decreases with an increasing number of transmitting antenna. This is because when extra degrees of freedom are offered by more transmitting antennas, system resource allocation becomes more efficient. In other words, the information leakage can be efficiently reduced since artificial noise jamming can be more accurately performed. As a results, a lower transmit power is required to fulfill the considered QoS requirements and less power is harvested.

%% file: 5_conclusion.tex
In this thesis, we focused on energy-efficient resource allocation algorithm design for SWIPT and considered the secure communication in SWIPT systems.

In Chapter \ref{chap:2_Sepuser}, we investigated the energy efficiencies for information transmission and energy harvesting in a basic SWIPT system with separated receivers, i.e., one information receiver and one energy harvester. The trade-off between three desirable but conflicting system design objectives, namely, IR-EE maximization, EH-EE maximization, and total transmit power minimization, was studied by formulating a MOOP. The non-convex problem was transformed by Charnes-Cooper transformation method and solved by SDP relaxation approach. The complete Pareto optimal set of the MOOP was achieved by the proposed energy-efficient resource allocation algorithm. In particular, the algorithm provided flexibility in power allocation when balancing between multiple objectives. Simulation results showed the trade-off region of the considered objectives and revealed the system performance gain on energy efficiency compared to the baseline scheme.

In Chapter \ref{chap:4_secureWIPT}, we studied the power-efficient resource allocation algorithm design for secure communication in a SWIPT system with power splitting receivers. Under the eavesdropping potential of multi-antenna energy harvesters, we adopted artificial noise generation to guarantee secure information delivery to legitimate receivers. The algorithm design was formulated as a non-convex optimization problem taking into account the QoS requirements on communication security and efficient power transfer. We applied SDP relaxation approach to obtain the optimal solution. Especially, the proposed power-efficient algorithm achieved the dual use of artificial noise on ensuring communication security and facilitating EH. Simulation results confirmed the remarkable performance of the proposed optimal scheme on energy saving and communication security.

In the future work, we are interested in energy-efficient resource allocation in orthogonal frequency multiple access (OFDMA) systems with secure SWIPT.

%% file: Appendix_A.tex
\section{Optimization Problem}\label{app:intro_MOOP}
A typical optimization problem consists of an objective function, constraints, and optimization variables. In general, the objective function is a function of optimization variables that evaluates one aspect of system performance. The set of constraints map the considered QoS requirements and physical limitations of the system. A single-objective optimization problem is given standardly by
\begin{eqnarray}
\underset{\mathbf{x}}\mino\,\,&&f_0(\mathbf{x})\\
\mathrm{subject\,\,to}\,\,&&g_l(\mathbf{x})\leq0,\quad l=1,2,\ldots,L,\notag\\
&&h_n(\mathbf{x})=0,\quad n=1,2,\ldots,N,\notag
\end{eqnarray}
where $f_0(\cdot)$ is a system objective function and $\mathbf{x}$ is a vector of optimization variables. $L$ and $N$ are the numbers of inequality constraints and equality constraints, respectively.  $g_l(\mathbf{x})$ is the $l$-th inequality constraint and $h_n(\mathbf{x})$ is the $n$-th equality constraint.

However, in practice multiple desirable system design objectives arise naturally in resource allocation problems. As the objective functions could be conflicting with each other, non-trivial trade-off occurs in this case, where the solution of single-objective resource allocation may not result in satisfactory system performance. Therefore, MOO is applied to address this type of resource allocation problem. A standard form of a MOOP can be posed as follows \cite{JR:MOOP}:
\begin{eqnarray}
\underset{\mathbf{x}}\mino\,\,&&\mathbf{F}(\mathbf{x})=[F_1(\mathbf{x}),F_2(\mathbf{x}),\dots,F_K(\mathbf{x})]^T\\
\mathrm{subject\,\,to}\,\,&&g_l(\mathbf{x})\leq0,\quad l=1,2,\ldots,L,\notag\\
&&h_n(\mathbf{x})=0,\quad n=1,2,\ldots,N,\notag
\end{eqnarray}
where $K$ is the number of objective functions and  $F_k(\mathbf{x}),\forall k\in\{1,\ldots,K\}$, is the $k$-th objective function.

In contrast to single-objective optimization, a solution to a MOOP is more of an abstract concept than a fixed point. In general, there is no single global solution which optimizes all the objective functions simultaneously. Typically, it is often necessary to determine a set of points that fit a predetermined definition for a optimum, which is Pareto optimality. Pareto optimality of MOOP is defined as
\begin{Def}Pareto Optimal: A point, $\mathbf{x}^*\in {\cal F}$, is Pareto Optimal if and only if there does not exist another point, $\mathbf{x}\in {\cal F}$, such that $\mathbf{F}(\mathbf{x})\preceq\mathbf{F}(\mathbf{x}^*)$ and $F_k(\mathbf{x})<F_k(\mathbf{x}^*)$, $k=1,\dots,K$, for at least one function.
\end{Def}

Evidently, any point that is not in the Pareto optimal set is strictly suboptimal. The Pareto optimal set is an analogy to global optimality that can achieve in MOO. We note that single-objective optimization problems are special case of MOOPs with $K=1$. In other words, if a resource allocation algorithm can solve the MOOP, then it can be used to solve the corresponding single-objective optimization problem.

%% file: Appendix_B.tex
\section{Proof of Proposition \ref{prop:equivalency}}\label{app:Prop_equivalency}
The proof is based on the Charnes-Cooper transformation \cite{JR:linear_fractional} and follows a similar approach in \cite{CN:mulobj_secure_WIPT}. By substitute the new optimization variables in (\ref{eqn:newvariabledefine}) into the original problem \ref{prob:WIPT_IR-EE}, we can obtain an equivalent problem representation as
\begin{eqnarray}\label{prob:app1}
\underset{\overline{\mathbf{W}}_\mathrm{I},\overline{\mathbf{W}}_\mathrm{E}\in\mathbb{H}^{N_\mathrm{T}},\theta}{\maxo} &&\overline{F_1}=\frac{\theta\log_2(1+\frac{\Tr(\mathbf{H}\overline{\mathbf{W}}_\mathrm{I})}{\theta\sigma_\mathrm{I}^2})}{\frac{\Tr(\overline{\mathbf{W}}_\mathrm{I}+\overline{\mathbf{W}}_\mathrm{E})}{\xi}+\theta(N_\mathrm{T}P_\mathrm{ant}+P_\mathrm{c})}\\
\mathrm{subject\,\,to}&&\overline{\mathrm{C1}},\,\,\overline{\mathrm{C2}},\,\,\overline{\mathrm{C3}},\,\,\overline{\mathrm{C5}}:\,\,\theta>0.\notag
\end{eqnarray}

Now we show that the above Problem (\ref{prob:app1}) is equivalent to Problem \ref{prob:WIPT_IR-EE_reform}. First, it can be observed that in Problem (\ref{prob:app1}) $\theta=0$ is impossible. Otherwise, $\overline{\mathbf{W}}_\mathrm{I}=\overline{\mathbf{W}}_\mathrm{E}=\mathbf{0}$ according to $\overline{\mathrm{C1}}$ and $\overline{\mathrm{C2}}$, the objective function is invalid. Thus, without loss of generality, the constraint $\theta>0$ can be replaced by $\theta\ge0$. Second, we prove by contradiction that $\overline{\mathrm{C4}}$ in Problem \ref{prob:WIPT_IR-EE_reform} is satisfied with equality for the optimal solution. Denote the optimal solution of Problem \ref{prob:WIPT_IR-EE_reform} as $(\overline{\mathbf{W}}_\mathrm{I}^*,\overline{\mathbf{W}}_\mathrm{E}^*,\theta^*)$. Then,
\begin{eqnarray}
\frac{\Tr(\overline{\mathbf{W}}_\mathrm{I}^*+\overline{\mathbf{W}}_\mathrm{E}^*)}{\xi}+\theta^*(N_\mathrm{T}P_\mathrm{ant}+P_\mathrm{c})=1.
\end{eqnarray}
Assume that $\overline{\mathrm{C4}}$ is fulfilled with strict inequality for the optimal solution, i.e., $\frac{\Tr(\overline{\mathbf{W}}_\mathrm{I}^*+\overline{\mathbf{W}}_\mathrm{E}^*)}{\xi}+\theta^*(N_\mathrm{T}P_\mathrm{ant}+P_\mathrm{c})<1$. Then, we construct a new feasible solution $(\overline{\mathbf{W}}_\mathrm{I}',\overline{\mathbf{W}}_\mathrm{E}',\theta')=(c\overline{\mathbf{W}}_\mathrm{I}^*,\overline{\mathbf{W}}_\mathrm{E}^*,c\theta^*)$, where $c>1$, such that $\frac{\Tr(\overline{\mathbf{W}}_\mathrm{I}'+\overline{\mathbf{W}}_\mathrm{E}')}{\xi}+\theta'(N_\mathrm{T}P_\mathrm{ant}+P_\mathrm{c})=1$. It can be verified that $(\overline{\mathbf{W}}_\mathrm{I}',\overline{\mathbf{W}}_\mathrm{E}',\theta')$ achieves a larger objective value in Problem \ref{prob:WIPT_IR-EE_reform} than $(\overline{\mathbf{W}}_\mathrm{I}^*,\overline{\mathbf{W}}_\mathrm{E}^*,\theta^*)$. Then, $(\overline{\mathbf{W}}_\mathrm{I}^*,\overline{\mathbf{W}}_\mathrm{E}^*,\theta^*)$ cannot be the optimal solution. Contradiction occurs. Thus,$\overline{\mathrm{C4}}$ must hold with equality. The equivalency between Problem (\ref{prob:app1}) and Problem \ref{prob:WIPT_IR-EE_reform} is proved, which means Problem \ref{prob:WIPT_IR-EE_reform} is equivalent to the original Problem \ref{prob:WIPT_IR-EE}. Similarly, The equivalency of Problem \ref{prob:WIPT_EH-EE_reform}, \ref{prob:WIPT_Ptrans_reform} and \ref{prob:multiobj_WIPTsepuser2} to their original problems can be proved by following the above approach.

\section{Proof of Theorem \ref{thm:rankone}}\label{app:rankone}
Theorem \ref{thm:rankone} can be proved by analyzing the KKT optimality conditions of the SDP relaxed Problem \ref{prob:multiobj_WIPTsepuser_relaxed}. First we need the Lagrangian function as the following
\begin{eqnarray} \label{eqn:appB7}
&&{\cal L}\big(\overline{\mathbf{W}}_\mathrm{I},\overline{\mathbf{W}}_\mathrm{E},\theta,\tau,\alpha,\beta,\mathbf{X},\mathbf{Y},\gamma_1,\gamma_2,\gamma_3,\delta\big)\\
&=&\tau+\alpha\big(\Tr(\overline{\mathbf{W}}_\mathrm{I}+\overline{\mathbf{W}}_\mathrm{E})-\theta P_{\mathrm{max}}\big)-\Tr(\mathbf{X}\overline{\mathbf{W}}_\mathrm{I})-\Tr(\mathbf{Y}\overline{\mathbf{W}}_\mathrm{E})-\delta\theta\notag\\
&+&\beta\big(\frac{\Tr(\overline{\mathbf{W}}_\mathrm{I}+\overline{\mathbf{W}}_\mathrm{E})}{\xi}+\theta(N_\mathrm{T}P_\mathrm{ant}+P_\mathrm{c})-1\big)+\gamma_1\Big[\omega_1\big(1-\frac{\theta}{\Phi_\mathrm{IR}^*}\log_2(1+\frac{\Tr(\mathbf{H}\overline{\mathbf{W}}_\mathrm{I})}{\theta\sigma_\mathrm{I}^2})\big)-\tau\Big]\notag\\
&+&\gamma_2\Big[\omega_2\big(1-\frac{\eta}{\Phi_\mathrm{EH}^*}\Tr(\mathbf{G}(\overline{\mathbf{W}}_\mathrm{I}+\overline{\mathbf{W}}_\mathrm{E}))\big)-\tau\Big]+\gamma_3\Big[\omega_3\frac{\xi}{P_\mathrm{max}}(\frac{1}{\theta}-N_\mathrm{T}P_\mathrm{ant}-P_\mathrm{c})-\tau\Big],\notag
\end{eqnarray}
where $\alpha,\beta,\mathbf{X},\mathbf{Y},\gamma_1,\gamma_2,\gamma_3,\delta$ are dual variables associated with the corresponding constraints, respectively. Since the SDP relaxed Problem \ref{prob:multiobj_WIPTsepuser_relaxed} satisfies Slater's constraint qualification and is convex with respect to the optimization variables, strong duality holds. Denote the optimal solution as $\{\overline{\mathbf{W}}_\mathrm{I}^*,\overline{\mathbf{W}}_\mathrm{E}^*,\theta^*,\tau^*\}$, and the optimal dual variables as $\{\alpha^*,\beta^*,\mathbf{X}^*,\mathbf{Y}^*,\gamma_1^*,\gamma_2^*,\gamma_3^*,\delta^*\}$. Then, based on KKT optimality conditions, the gradient of Lagrangian function with respect to $\overline{\mathbf{W}}_\mathrm{I}$ and $\overline{\mathbf{W}}_\mathrm{E}$ vanish, and the complementary slackness condition is satisfied as well. Thus, we have
\begin{eqnarray}
&&\frac{\partial{\cal L}}{\partial\overline{\mathbf{W}}_\mathrm{I}}=(\alpha+\frac{\beta}{\xi})\mathbf{I}-\mathbf{X}
-\frac{\gamma_1\omega_1\theta}{\Phi_\mathrm{IR}^*\big(\theta\sigma_\mathrm{I}^2+\Tr(\mathbf{H}\overline{\mathbf{W}}_\mathrm{I})\big)}\mathbf{H}
-\frac{\gamma_2\omega_2\eta}{\Phi_\mathrm{EH}^*}\mathbf{G}=\mathbf{0}\notag\\
&&\frac{\partial{\cal L}}{\partial\overline{\mathbf{W}}_\mathrm{E}}=(\alpha+\frac{\beta}{\xi})\mathbf{I}-\mathbf{Y}
-\frac{\gamma_2\omega_2\eta}{\Phi_\mathrm{EH}^*}\mathbf{G}=\mathbf{0}\notag\\
&&\Tr(\mathbf{X}\overline{\mathbf{W}}_\mathrm{I})=0,\notag\\
&&\Tr(\mathbf{Y}\overline{\mathbf{W}}_\mathrm{E})=0.\notag\\
\Longrightarrow &&\mathbf{X}=(\alpha+\frac{\beta}{\xi})\mathbf{I}-\frac{\gamma_2\omega_2\eta}{\Phi_\mathrm{EH}^*}\mathbf{G}
-\frac{\gamma_1\omega_1\theta}{\Phi_\mathrm{IR}^*\big(\theta\sigma_\mathrm{I}^2+\Tr(\mathbf{H}\overline{\mathbf{W}}_\mathrm{I})\big)}\mathbf{H}\label{eqn:appB1}\\
&&\mathbf{Y}=(\alpha+\frac{\beta}{\xi})\mathbf{I}-\frac{\gamma_2\omega_2\eta}{\Phi_\mathrm{EH}^*}\mathbf{G}\label{eqn:appB2}\\
&&\mathbf{X}\overline{\mathbf{W}}_\mathrm{I}=\mathbf{0},\label{eqn:appB3}\\
&&\mathbf{Y}\overline{\mathbf{W}}_\mathrm{E}=\mathbf{0}.\label{eqn:appB4}
\end{eqnarray}

Now, we investigate $\overline{\mathbf{W}}_\mathrm{I}$ and $\overline{\mathbf{W}}_\mathrm{E}$ in three cases.

\textbf{Case 1}: $\omega_1=\omega_2=0$, $\omega_3=1$. (\ref{eqn:appB1}) and (\ref{eqn:appB2}) become $\mathbf{X}=\mathbf{Y}=(\alpha+\frac{\beta}{\xi})\mathbf{I}$, which means that $\mathbf{X}$ and $\mathbf{Y}$ are matrices of full rank. Due to (\ref{eqn:appB3}) and (\ref{eqn:appB4}), $\overline{\mathbf{W}}_\mathrm{I}$ lies in the null space of $\mathbf{X}$, and $\overline{\mathbf{W}}_\mathrm{E}$ lies in the null space of $\mathbf{Y}$. Thus, $\overline{\mathbf{W}}_\mathrm{I}=\overline{\mathbf{W}}_\mathrm{E}=\mathbf{0}$.

\textbf{Case 2}: $\omega_1=0$, $\omega_2\neq0$. This means that IR-EE maximization is not considered and EH-EE maximization is required. Thus, $\gamma_2\neq0$ must hold to keep constraint $\mathrm{C6}$ for $j=2$. Then, (\ref{eqn:appB1}) and (\ref{eqn:appB2}) become $\mathbf{X}=\mathbf{Y}=(\alpha+\frac{\beta}{\xi})\mathbf{I}-\frac{\gamma_2\omega_2\eta}{\Phi_\mathrm{EH}^*}\mathbf{G}$. Since $\frac{\gamma_2\omega_2\eta}{\Phi_\mathrm{EH}^*}\mathbf{G}$ is of rank one, based on the basic property of the rank of matrices, $\Rank(\mathbf{X})\ge N_\mathrm{T}-1$ and $\Rank(\mathbf{Y})\ge N_\mathrm{T}-1$. If $\Rank(\mathbf{X})=\Rank(\mathbf{Y})=N_\mathrm{T}$, then $\overline{\mathbf{W}}_\mathrm{I}=\overline{\mathbf{W}}_\mathrm{E}=\mathbf{0}$. which can not occur since EH-EE maximization is required. As a result, we can conclude $\Rank(\mathbf{X})=\Rank(\mathbf{Y})=N_\mathrm{T}-1$, $\Rank(\overline{\mathbf{W}}_\mathrm{I})=1$ and $\Rank(\overline{\mathbf{W}}_\mathrm{E})=1$.

In particular, we can see from the problem formulation that $\overline{\mathbf{W}}_\mathrm{I}$ and $\overline{\mathbf{W}}_\mathrm{E}$ are equivalent optimization variables in this case. Thus, the optimization variables can be redefined as $\overline{\mathbf{W}}_\mathrm{I}'=\overline{\mathbf{W}}_\mathrm{I}+\overline{\mathbf{W}}_\mathrm{E}$ and $\overline{\mathbf{W}}_\mathrm{E}'=\mathbf{0}$. Then, the Lagrangian multiplier $\mathbf{X}'$ with respect to $\overline{\mathbf{W}}_\mathrm{I}'$ results in $\mathbf{X}'=(\alpha+\frac{\beta}{\xi})\mathbf{I}-\frac{\gamma_2\omega_2\eta}{\Phi_\mathrm{EH}^*}\mathbf{G}$. As aforementioned, $\Rank(\mathbf{X}')=N_\mathrm{T}-1$ and $\Rank(\overline{\mathbf{W}}_\mathrm{I}')=1$ must hold.

\textbf{Case 3}: $\omega_1\neq0$, $\omega_2\neq0$. IR-EE maximization and EH-EE maximization are both considered, thus constraint $\mathrm{C6}$ for $j=1,2$ is active, i.e., $\gamma_1\neq0$ and $\gamma_2\neq0$. Since the Lagrangian multiplier $\mathbf{Y}\succeq\mathbf{0}$, we have  $\alpha+\frac{\beta}{\xi}\ge\frac{\gamma_2\omega_2\eta}{\Phi_\mathrm{EH}^*}g_1$ according to equation (\ref{eqn:appB2}), where $g_1$ is the largest eigenvalue of $\mathbf{G}$ \cite{JR:multiuser_MISO}. If the equality holds, i.e., $\alpha+\frac{\beta}{\xi}=\frac{\gamma_2\omega_2\xi}{\Phi_\mathrm{EH}^*}g_1$, we consider $\mathbf{X}$ in (\ref{eqn:appB1}). Since $\theta>0$ is proved in Appendix \ref{app:Prop_equivalency}, to satisfy $\mathbf{X}\succeq\mathbf{0}$, $\gamma_1=0$ must hold. Contradiction occurs. This implies $\alpha+\frac{\beta}{\xi}>\frac{\gamma_2\omega_2\eta}{\Phi_\mathrm{EH}^*}g_1$. Thus, $\mathbf{Y}\succ\mathbf{0}$, which means $\mathbf{Y}$ is a full-rank matrix. Therefore, $\overline{\mathbf{W}}_\mathrm{E}=\mathbf{0}$. On the other hand, according to equation (\ref{eqn:appB1}), $\mathbf{X}=\mathbf{Y}-\frac{\gamma_1\omega_1\theta}{\Phi_\mathrm{IR}^*\big(\theta\sigma_\mathrm{I}^2+\Tr(\mathbf{H}\overline{\mathbf{W}}_\mathrm{I})\big)}\mathbf{H}$. Since $\Rank\Bigg(\frac{\gamma_1\omega_1\theta}{\Phi_\mathrm{IR}^*\big(\theta\sigma_\mathrm{I}^2+\Tr(\mathbf{H}\overline{\mathbf{W}}_\mathrm{I})\big)}\mathbf{H}\Bigg)=1$, we have $\Rank(\mathbf{X})\ge N_\mathrm{T}-1$. As aforementioned, $\Rank(\mathbf{X})=N_\mathrm{T}-1$ and  $\Rank(\overline{\mathbf{W}}_\mathrm{I})=1$ must hold.

Consequently, in all cases, the optimal solution of the relaxed Problem \ref{prob:multiobj_WIPTsepuser_relaxed} satisfies $\Rank(\overline{\mathbf{W}}_\mathrm{I}^*)=1$ and $\Rank(\overline{\mathbf{W}}_\mathrm{E}^*)\leq1$. In particular, an optimal solution with $\Rank(\overline{\mathbf{W}}_\mathrm{I}^*)=1$ and $\overline{\mathbf{W}}_\mathrm{E}^*=\mathbf{0}$ can always be constructed.

\section{Proof of Proposition \ref{prop:relaxed_c2}} \label{app:Prop_relaxed_c2}
We start the proof by  expressing  constraint C2 as
\begin{eqnarray}\label{eqn:det_ineq}
\det(\mathbf{I}_{N_{\mathrm{R}}}+\mathbf{Q}_{m}^{-1}\mathbf{G}_m^H\mathbf{W}_k\mathbf{G}_m)&\le&\psi_{m,k}\\
\stackrel{(a)}{\Longleftrightarrow }\det(\mathbf{I}_{N_{\mathrm{R}}}+\mathbf{Q}_{m}^{-1/2}\mathbf{G}_m^H\mathbf{W}_k\mathbf{G}_m\mathbf{Q}_{m}^{-1/2})&\le&\psi_{m,k},\label{eqn:det_ineq2}
\end{eqnarray}
where $(a)$ is due to the fact that $\det(\mathbf{I}+\mathbf{AB})=\det(\mathbf{I}+\mathbf{BA})$ holds for any matrices $\mathbf{A}$ and $\mathbf{B}$. Then, we introduce the following lemma which provides a lower bound on the left hand side of (\ref{eqn:det_ineq2}).
\begin{Lem}\label{lemma:det_trace} For any square matrix $\mathbf{A}\succeq \mathbf{0}$, we have $\det(\mathbf{I}+\mathbf{A})\ge 1+\Tr(\mathbf{A})$ \cite{JR:AN_MISO_secrecy}, where the equality holds if and only if $\Rank(\mathbf{A})\le 1$.
\end{Lem}
Exploiting Lemma \ref{lemma:det_trace}, the left hand side of (\ref{eqn:det_ineq2}) is bounded below by
\begin{eqnarray}
\det(\mathbf{I}_{N_{\mathrm{R}}}+\mathbf{Q}_{m}^{-1/2}\mathbf{G}_m^H\mathbf{W}_k\mathbf{G}_m\mathbf{Q}_{m}^{-1/2})\,\ge\, 1+\Tr(\mathbf{Q}_{m}^{-1/2}\mathbf{G}_m^H\mathbf{W}_k\mathbf{G}_m\mathbf{Q}_{m}^{-1/2}).\label{eqn:trace_ineq3}
\end{eqnarray}
Subsequently, by combining equations (\ref{eqn:det_ineq}), (\ref{eqn:det_ineq2}),  and (\ref{eqn:trace_ineq3}), we have the following implications:
\begin{subequations}
\begin{eqnarray}
&&\mathrm{(\ref{eqn:det_ineq})} \Longleftrightarrow \mathrm{(\ref{eqn:det_ineq2})}\\
&\Longrightarrow& \Tr(\mathbf{Q}_{m}^{-1/2}\mathbf{G}_m^H\mathbf{W}_k\mathbf{G}_m\mathbf{Q}_{m}^{-1/2})\le \psi_{m,k} -1\\
&\stackrel{(b)}{\Longrightarrow}&\lambda_{\max}(\mathbf{Q}_{m}^{-1/2}\mathbf{G}_m^H\mathbf{W}_k\mathbf{G}_m\mathbf{Q}_{m}^{-1/2})\le \psi_{m,k} -1\\
&\Longleftrightarrow&\mathbf{Q}_{m}^{-1/2}\mathbf{G}_m^H\mathbf{W}_k\mathbf{G}_m\mathbf{Q}_{m}^{-1/2} \preceq (\psi_{m,k} -1)\mathbf{I}_{N_{\mathrm{R}}}\\
&\Longleftrightarrow & \mathbf{G}_m^H\mathbf{W}_k\mathbf{G}_m\preceq (\psi_{m,k} -1)\mathbf{Q}_{m}. \label{eqn:trace_final}
\end{eqnarray}
\end{subequations}
$\lambda_{\mathrm{max}}(\mathbf{A})$ denotes the maximum eigenvalue of matrix $\mathbf{A}$. $(b)$ is due to  $\Tr(\mathbf{A})\ge \lambda_{\mathrm{max}}(\mathbf{A})$ for a positive semidefinite matrix $\mathbf{A}\succeq \mathbf{0}$. We note that  equations (\ref{eqn:det_ineq}) and (\ref{eqn:trace_final}) are equivalent when $\Rank(\mathbf{W}_k)= 1,\forall k$.

\section{Proof of Theorem \ref{thm:rankone_condition}} \label{app:thm_rankone_condition}
We follow a similar approach as in \cite{JR:secure_WIPT_MISO_ruizh,JR:Kwan_secure_imperfect} to prove Theorem 1. The proof is divided into two parts. In the first part, we study the solution structure of Problem \ref{prob:sdp_relaxation}. Then in the second part, we propose a simple method for constructing an optimal solution with rank-one $\mathbf{W}_k$. In order to verify the tightness of the adopted SDP relaxation,  we analyze the Karush-Kuhn-Tucker (KKT) conditions of the SDP relaxed Problem \ref{prob:sdp_relaxation} by introducing the corresponding Lagrangian and the dual problem. The
Lagrangian of Problem \ref{prob:sdp_relaxation} is given by
\begin{eqnarray}
&&{\cal L}\big(\mathbf{W}_k,\mathbf{V},\rho_k,\mathbf{Z}_k,\mathbf{Y},\mathbf{X}_{m,k},\alpha_k,\beta_k,\nu_m\big)\\
&=&\sum_{k=1}^K\Tr(\mathbf{W}_k)+\Tr(\mathbf{V})-\Tr(\mathbf{Y}\mathbf{V})-\sum_{k=1}^K\Tr(\mathbf{Z}_k\mathbf{W}_k)\notag\\
&+&\sum_{k=1}^K\alpha_k\big[-\frac{1}{\Gamma^\mathrm{req}_k}\Tr(\mathbf{h}_k\mathbf{h}_k^H\mathbf{W}_k)+\sum\limits_
{\substack{j\neq k}}^K\Tr(\mathbf{h}_k\mathbf{h}_k^H\mathbf{W}_j)+\Tr(\mathbf{h}_k\mathbf{h}_k^H\mathbf{V})+\sigma_{\mathrm{ant}}^2+\frac{1}{\rho_k}\sigma_{\mathrm{s}}^2\big]\notag\\
&+&\sum_{k=1}^K\beta_k\Big[\frac{P^{\mathrm{req1}}_{k}}{\eta(1-\rho_k)}-\sigma_{\mathrm{ant}}^2-\Tr\big(\mathbf{h}_k\mathbf{h}_k^H(\mathbf{V}+\sum_{j=1}^K\mathbf{W}_j)\big)\Big]\notag\\
&+&\sum_{m=1}^M\nu_m\Big[\frac{P^\mathrm{req2}_m}{\eta}-N_{\mathrm{R}}\sigma_{\mathrm{ant}}^2-\Tr\big(\mathbf{G}_m\mathbf{G}_m^H(\mathbf{V}+\sum_{k=1}^K\mathbf{W}_k)\big)\Big]\notag\\
&+&\sum_{m=1}^M\sum_{k=1}^K\Tr\big[\mathbf{X}_{m,k}\big(\mathbf{G}_m^H\mathbf{W}_k\mathbf{G}_m-(\psi_{m,k}-1)\mathbf{Q}_{m}\big)\big]\notag,
\end{eqnarray}
where $\mathbf{X}_{m,k}$, $\mathbf{Y}$, and $\mathbf{Z}_k$ are the dual variable matrices for constraints $\overline{\mathrm{C2}}$ , C6, and C7, respectively. $\alpha_k$, $\beta_k$, and $\nu_m$ are the scalar dual variables of  constraints C1, C3, and C4, respectively.
On the other hand, constraint C5 for $\rho_k$
is  satisfied automatically and the optimal $\rho_k$ will be  illustrated in the later part of this proof. Then, the dual problem of the SDP relaxed Problem \ref{prob:sdp_relaxation} is given by
\begin{equation}\hspace*{-0mm}\label{eqn:dual}
\underset{\underset{\mathbf{Z}_k,\mathbf{Y},\mathbf{X}_{m,k}\succeq \mathbf{0}}{\alpha_k,\beta_k,\nu_m\ge0}}{\maxo}\quad \underset{\rho_k,\mathbf{W}_k,\mathbf{V}\in\mathbb{H}^{N_{\mathrm{T}}}}{\mino} \quad {\cal L}\big(\mathbf{W}_k,\mathbf{V},\rho_k,\mathbf{Z}_k,\mathbf{Y},\mathbf{X}_{m,k},\alpha_k,\beta_k,\nu_m\big).
\end{equation}
Since Problem \ref{prob:sdp_relaxation} satisfies Slater's constraint qualification and is jointly convex with respect to optimization variables, strong duality holds. Thus solving \ref{eqn:dual} is equivalent to solving Problem \ref{prob:sdp_relaxation}. We define $\{\mathbf{W}_k^*,\mathbf{V}^*,\rho_k^*\}$ and $\{\mathbf{Z}_k^*,\mathbf{Y}^*,\mathbf{X}_{m,k}^*,\nu_m^*,\beta_k^*,\alpha_k^*\}$ as the optimal primal solution and the optimal dual solution of Problem \ref{prob:sdp_relaxation} with $\mathbf{Z}_k^*,\mathbf{X}_{m,k}^*\succeq\mathbf{0},\,\,\alpha_k^*,\,\beta_k^*,\nu_m^*\ge 0$. Now, we focus on those KKT conditions which are useful in the proof:
\begin{eqnarray}
&&\mathbf{Z}_k^*\mathbf{W}_k^*=\mathbf{0},\label{eqn:complementary_cond}\\
&&\mathbf{Z}_k^*=\mathbf{U}_{k}-(\beta_k^*+\frac{\alpha_k^*}{\Gamma^\mathrm{req}_k})\mathbf{h}_k\mathbf{h}_k^H,\label{eqn:lagrangian_gradient}\\
\mathrm{where}&&\mathbf{U}_k=\mathbf{I}_{N_{\mathrm{T}}}+\sum_{m=1}^M\mathbf{G}_m(\mathbf{X}_{m,k}^*-\nu_m^*\mathbf{I}_{N_{\mathrm{T}}})\mathbf{G}_m^H+\sum_{j\neq k}^K(\alpha_j^*-\beta_j^*)\mathbf{h}_j\mathbf{h}_j^H,\\
&&\rho^*_k=\frac{\sqrt{\alpha_k^*\sigma_{\mathrm{s}}^2\eta}}{\sqrt{\alpha_k^*\sigma_{\mathrm{s}}^2\eta}+\sqrt{\beta_k^*P^{\mathrm{req1}}_{k}}}.\label{eqn:optimal_rho}
\end{eqnarray}
It can be observed from (\ref{eqn:optimal_rho}) that constraint $\mathrm{C5}$ is  automatically satisfied. Besides,  $\alpha^*_k,\beta^*_k>0$  must holds for $\Gamma^\mathrm{req}_k>0$ and $P_{k}^\mathrm{req1}>0$. On the other hand, because of the complementary slackness condition on $\mathbf{W}^*_k$ in (\ref{eqn:complementary_cond}), $\mathbf{W}_k^*$ lies in the null space of $\mathbf{Z}_k^*$ for $\mathbf{W}_k^*\ne \mathbf{0}$.  In other words, the structure of  $\mathbf{W}^*_k$ depends on the space spanned by $\mathbf{Z}_k^*$. Thus, we focus on the following two cases to reveal the space spanned by $\mathbf{Z}_k^*$. Without loss of generality, we denote $r_k=\Rank(\mathbf{U}_k)$. In the first case, we investigate  the structure of  $\mathbf{W}^*_k$ when $\mathbf{U}_k$ is a full-rank matrix, i.e., $r_k=N_{\mathrm{T}}$. By exploiting (\ref{eqn:lagrangian_gradient}) and a basic  inequality for the rank of matrices, we have
\begin{eqnarray}\label{eqn:rank_inequality}
&&\Rank(\mathbf{Z}_k^*)+\Rank((\frac{\alpha_k^*}{\Gamma^\mathrm{req}_k}+\beta_k^*)\mathbf{h}_k\mathbf{h}_k^H)\ge\Rank(\mathbf{U}_k)\notag\\
\Longleftrightarrow &&\Rank(\mathbf{Z}_k^*)\ge N_{\mathrm{T}}-1\quad  \mathrm{for} \quad\alpha^*_k,\beta_k^*>0.
\end{eqnarray}
For $\Gamma^\mathrm{req}_k>0$ and   $\Rank(\mathbf{U}_k)=N_{\mathrm{T}}$, $\Rank(\mathbf{W}_k^*)=1$ and  $\Rank(\mathbf{Z}_k^*)=N_{\mathrm{T}}-1$ must hold simultaneously. Next, we consider the case when $\Rank(\mathbf{U}_k)$ is rank-deficient, i.e., $ r_k<N_{\mathrm{T}}$. Without loss of generality, we define $\nullspace(\mathbf{U}_k)=\mathbf{N}_k$, $\mathbf{N}_k\in\mathbb{C}^{N_{\mathrm{T}}\times (N_{\mathrm{T}}-r_k)}$ such that $\mathbf{U}_k\mathbf{N}_k=\mathbf{0}$ and $\Rank(\mathbf{N}_k)=N_{\mathrm{T}}-r_k$. Let $\boldsymbol{\varrho}_{t_k}\in\mathbb{C}^{N_{\mathrm{T}}\times 1}$, $1\le t_k\le N_{\mathrm{T}}-r_k$, denote the $t_k$-th column vector of $\mathbf{N}_k$.  Then, by exploiting (\ref{eqn:lagrangian_gradient}),  we have the following equality:
\begin{eqnarray}\label{eqn:pre-post}
\boldsymbol{\varrho}_{t_k}^H\mathbf{Z}_k^*\boldsymbol{\varrho}_{t_k}=- (\frac{\alpha^*_k}{\Gamma^\mathrm{req}_k}+\beta^*_k)\boldsymbol{\varrho}_{t_k}^H\mathbf{h}_k\mathbf{h}_k^H\boldsymbol{\varrho}_{t_k}.
\end{eqnarray}
Combining  $\mathbf{Z}^*_k\succeq \mathbf{0}$ and $\frac{\alpha^*_k}{\Gamma^\mathrm{req}_k}+\beta^*_k>0$,   $(\frac{\alpha^*_k}{\Gamma^\mathrm{req}_k}+\beta^*_k)\boldsymbol{\varrho}_{t_k}^H\mathbf{h}_k\mathbf{h}_k^H\boldsymbol{\varrho}_{t_k}={0},\forall t_k\in\{1,\ldots,N_{\mathrm{T}}-r_k\}$, holds in (\ref{eqn:pre-post}). In other words,
\begin{eqnarray}\label{eqn:rank-nullspace}
\mathbf{Z}^*_k\mathbf{N}_k=\mathbf{0} \quad \mathrm{ and }\quad \mathbf{h}_k\mathbf{h}_k^H\mathbf{N}_k=\mathbf{0}
\end{eqnarray}
hold and $\mathbf{N}_k$  lies in the null spaces of $\mathbf{h}_k\mathbf{h}_k^H$ and $ \mathbf{Z}^*_k$ simultaneously. Furthermore, $\Rank\big(\nullspace(\mathbf{Z}^*) \big)\ge N_{\mathrm{T}}-r_k$ holds for satisfying  $\mathbf{Z}_k^*\mathbf{N}_k=\mathbf{0}$. On the other hand, from (\ref{eqn:rank_inequality}) and $\Rank(\mathbf{U}_k)=r_k$, we obtain
\begin{eqnarray}\label{eqn:temp}
 \Rank(\mathbf{Z}_k^*)\ge r_k-1.
\end{eqnarray} Then, by utilizing  (\ref{eqn:rank-nullspace})  and (\ref{eqn:temp}),  $\Rank\big(\nullspace(\mathbf{Z}^*_k)\big)$ is bounded between
\begin{eqnarray}
&&N_{\mathrm{T}}- r_k+1   \ge  \Rank\big(\nullspace(\mathbf{Z}^*_k) \big)\ge N_{\mathrm{T}}-r_k.
\end{eqnarray}
As a result, either  $\Rank\big(\nullspace(\mathbf{Z}^*_k) \big)=N_{\mathrm{T}}-r_k$ or $\Rank\big(\nullspace(\mathbf{Z}^*_k)\big)=N_{\mathrm{T}}-r_k+1 $ holds for the optimal solution. Suppose $\Rank\big(\nullspace(\mathbf{Z}^*_k) \big)=N_{\mathrm{T}}-r_k$ and thus  $\nullspace(\mathbf{Z}^*_k)= \mathbf{N}_k$.  Then, we can express $\mathbf{W}^*_k$ as $\mathbf{W}^*_k=\sum_{t_k=1}^{N_{\mathrm{T}}-r_k} \gamma_{t_k}\boldsymbol{\varrho}_{t_k} \boldsymbol{\varrho}_{t_k}^H$  for some positive constants $\gamma_{t_k}\ge 0,\forall t_k\in\{1,\ldots,N_{\mathrm{T}}-r_k\} $. Yet,  due to (\ref{eqn:rank-nullspace}),
\begin{eqnarray}
\Tr\big(\mathbf{h}_k\mathbf{h}^H_k\mathbf{W}^*_k\big)=\sum_{t_k=1}^{N_{\mathrm{T}}-r_k}\gamma_{t_k}
\Tr\big(\boldsymbol{\varrho}_{t_k}^H\mathbf{h}_k\mathbf{h}^H_k\boldsymbol{\varrho}_{t_k} \big)=0
\end{eqnarray}
holds which cannot satisfy constraint C1 for $\Gamma^\mathrm{req}_k>0$. Thus, $\Rank\big(\nullspace(\mathbf{Z}^*_k)\big)=N_{\mathrm{T}}-r_k+1$ has to hold for the optimal  $\mathbf{W}^*_k$. Besides,  there exists one subspace spanned by a unit norm vector  $\mathbf{u}_k\in \mathbb{C}^{N_{\mathrm{T}}\times 1}$  such that $\mathbf{Z}^*_k\mathbf{u}_k=\mathbf{0}$ and $\mathbf{N}_k^H\mathbf{u}_k=\mathbf{0}$. Therefore,  the orthonormal  null space of $\mathbf{Z}^*_k$ can be presented as
\begin{eqnarray}
\nullspace(\mathbf{Z}^*_k)=\Big\{\mathbf{N}_k\cup\mathbf{u}_k\Big\}.
\end{eqnarray}
In summary, without loss of generality, we can express the optimal solution of  $\mathbf{W}^*_k$  as
\begin{eqnarray}\label{eqn:general_structure}
\mathbf{W}^*_k=\sum_{t_k=1}^{N_{\mathrm{T}}-r_k} \gamma_{t_k}\boldsymbol{\varrho}_{t_k} \boldsymbol{\varrho}_{t_k}^H  + f_k\mathbf{u}_k\mathbf{u}_k^H,
\end{eqnarray}
where  $f_k>0 $ is some positive scaling constant.In the second part of the proof, for $\Rank(\mathbf{W}^*_k)>1$,  we reconstruct another solution of the Problem \ref{prob:sdp_relaxation},  $\{\mathbf{\widetilde  W}_k ,\mathbf{\widetilde  V},\widetilde \rho_k\}$, based on  (\ref{eqn:general_structure}).Let the constructed solution set be given by
\begin{eqnarray}\label{eqn:rank-one-structure}\mathbf{\widetilde W}_k&=&f_k\mathbf{u}_k\mathbf{u}_k^H=\mathbf{W}^*_k-\sum_{t_k=1}^{N_{\mathrm{T}}-r_k} \gamma_{t_k} \boldsymbol{\varrho}_{t_k} \boldsymbol{\varrho}_{t_k} ^H, \\
 \mathbf{\widetilde V}&=&\mathbf{ V^*}+\sum_{t_k=1}^{N_{\mathrm{T}}-r_k} \gamma_{t_k} \boldsymbol{\varrho}_{t_k} \boldsymbol{\varrho}_{t_k}^H, \quad\widetilde \rho_k=\rho^*_k \label{eqn:rank-one-structure2}.
\end{eqnarray}
It can be easily verified  that   $\{\mathbf{\widetilde W}_k,\mathbf{\widetilde V},  \widetilde \rho_k\}$ not only satisfies the constraints in Problem \ref{prob:sdp_relaxation}, but also achieves the same optimal objective value as $\{\mathbf{ W}_k,\mathbf{ V},  \rho_k\}$ with $\Rank(\mathbf{\widetilde W}_k)=1,\forall k$. The actual values of  $\{\mathbf{\widetilde W}_k,\mathbf{\widetilde V},  \widetilde \rho_k\}$ can be obtained by substituting  (\ref{eqn:rank-one-structure}) and (\ref{eqn:rank-one-structure2}) into Problem \ref{prob:sdp_relaxation} and solving the resulting  convex optimization problem for  $f_k$ and $\gamma_{t_k}$.

%% file: Main_Document.bbl
\begin{thebibliography}{10}
\providecommand{\url}[1]{#1}
\csname url@samestyle\endcsname
\providecommand{\newblock}{\relax}
\providecommand{\bibinfo}[2]{#2}
\providecommand{\BIBentrySTDinterwordspacing}{\spaceskip=0pt\relax}
\providecommand{\BIBentryALTinterwordstretchfactor}{4}
\providecommand{\BIBentryALTinterwordspacing}{\spaceskip=\fontdimen2\font plus
\BIBentryALTinterwordstretchfactor\fontdimen3\font minus
  \fontdimen4\font\relax}
\providecommand{\BIBforeignlanguage}[2]{{%
\expandafter\ifx\csname l@#1\endcsname\relax
\typeout{** WARNING: IEEEtran.bst: No hyphenation pattern has been}%
\typeout{** loaded for the language `#1'. Using the pattern for}%
\typeout{** the default language instead.}%
\else
\language=\csname l@#1\endcsname
\fi
#2}}
\providecommand{\BIBdecl}{\relax}
\BIBdecl

\bibitem{Powercast}
\BIBentryALTinterwordspacing
{Powercast Coporation}, ``{RF} {E}nergy {H}arvesting and {W}ireless {P}ower for
  {L}ow-{P}ower {A}pplications,'' 2011. [Online]. Available:
  \url{http://www.mouser.com/pdfdocs/Powercast-Overview-2011-01-25.pdf}
\BIBentrySTDinterwordspacing

\bibitem{Krikidis2014}
I.~Krikidis, S.~Timotheou, S.~Nikolaou, G.~Zheng, D.~W.~K. Ng, and R.~Schober,
  ``Simultaneous {W}ireless {I}nformation and {P}ower {T}ransfer in {M}odern
  {C}ommunication {S}ystems,'' \emph{IEEE Commun. Mag.}, vol.~52, no.~11, pp.
  104--110, Nov. 2014.

\bibitem{Ding2014}
Z.~Ding, C.~Zhong, D.~W.~K. Ng, M.~Peng, H.~A. Suraweera, R.~Schober, and H.~V.
  Poor, ``Application of smart {A}ntenna {T}echnologies in {S}imultaneous
  {W}ireless {I}nformation and {P}ower {T}ransfer,'' 2015, to appear in the
  \emph{ IEEE Commun. Mag.}

\bibitem{CN:WIPT_fundamental}
L.~Varshney, ``{Transporting Information and Energy Simultaneously},'' in
  \emph{Proc. IEEE Intern. Sympos. on Inf. Theory}, Jul. 2008, pp. 1612 --1616.

\bibitem{CN:Shannon_meets_tesla}
P.~Grover and A.~Sahai, ``{Shannon Meets Tesla: Wireless Information and Power
  Transfer},'' in \emph{Proc. IEEE Intern. Sympos. on Inf. Theory}, Jun. 2010,
  pp. 2363 --2367.

\bibitem{JR:WIP_receiver}
X.~Zhou, R.~Zhang, and C.~K. Ho, ``{Wireless Information and Power Transfer:
  Architecture Design and Rate-Energy Tradeoff},'' \emph{IEEE Trans. Commun.},
  vol.~61, pp. 4754--4767, Nov. 2013.

\bibitem{CN:MIMO_WIPT}
R.~Zhang and C.~K. Ho, ``{MIMO Broadcasting for Simultaneous Wireless
  Information and Power Transfer},'' in \emph{Proc. IEEE Global Telecommun.
  Conf.}, Dec. 2011, pp. 1 --5.

\bibitem{JR:multiuser_MISO}
J.~Xu, L.~Liu, and R.~Zhang, ``{Mutiuser MISO Beamforming for Simultaneous
  Wireless Information and Power Transfer},'' \emph{IEEE Trans. Signal
  Process.}, vol.~62, pp. 4798 -- 4810, Jul. 2014.

\bibitem{JR:beamforming_MISO}
Q.~Shi, L.~Liu, W.~Xu, and R.~Zhang, ``{Joint Transmit Beamforming and Receive
  Power Splitting for MISO SWIPT Systems},'' \emph{submitted for possible
  journal publication}, Sep. 2013.

\bibitem{Xu2013}
J.~Xu, L.~Liu, and R.~Zhang, ``Multiuser {B}eamforming for {S}imultaneous
  {W}ireless {I}nformation and {P}ower {T}ransfer,'' \emph{Proc. IEEE Intern.
  Conf. on Acoustics, Speech and Signal Process.}, pp. 4754--4758, May 2013.

\bibitem{CN:multiuser_OFDM_WIPT}
D.~W.~K. Ng, E.~S. Lo, and R.~Schober, ``{Energy-Efficient Resource Allocation
  in Multiuser OFDM Systems with Wireless Information and Power Transfer},'' in
  \emph{Proc. IEEE Wireless Commun. and Networking Conf.}, 2013.

\bibitem{CN:Eurosip_SWIPT}
D.~W.~K. Ng and R.~Schober, ``{Spectral Efficient Optimization in OFDM Systems
  With Wireless Information and Power Transfer},'' in \emph{Proc. Europ. Signal
  Process. Conf.}, Sep. 2013, pp. 1--5.

\bibitem{JR:WIPT_fullpaper_OFDMA}
D.~W.~K. Ng, E.~S. Lo, and R.~Schober, ``{Wireless Information and Power
  Transfer: Energy Efficiency Optimization in OFDMA Systems},'' \emph{IEEE
  Trans. Wireless Commun.}, vol.~12, pp. 6352--6370, Dec. 2013.

\bibitem{CN:strategies_twouserMIMO}
J.~Park and B.~Clerckx, ``{Transmission strategies for joint wireless
  information and energy transfer in a two-user MIMO interference channel},''
  in \emph{Proc. IEEE Intern. Commun. Conf.}, 2013, pp. 591 --595.

\bibitem{Morsi2014}
R.~Morsi, D.~Michalopoulos, and R.~Schober, ``Multi-{U}ser {S}cheduling
  {S}chemes for {S}imultaneous {W}ireless {I}nformation and {P}ower
  {T}ransfer,'' in \emph{Proc. IEEE Intern. Commun. Conf.}, Jun., pp.
  4994--4999.

\bibitem{CN:Kwan_globecom2014}
D.~W.~K. Ng and R.~Schober, ``{Resource Allocation for Coordinated Multipoint
  Networks With Wireless Information and Power Transfer},'' in \emph{Proc. IEEE
  Global Telecommun. Conf.}, Dec. 2014, pp. 4281--4287.

\bibitem{CN:Maryna_2015}
\BIBentryALTinterwordspacing
M.~Chynonova, R.~Morsi, D.~W.~K. Ng, and R.~Schober, ``{Optimal Multiuser
  Scheduling Schemes for Simultaneous Wireless Information and Power
  Transfer},'' Feb. 2015. [Online]. Available:
  \url{http://arxiv.org/abs/1502.02179}
\BIBentrySTDinterwordspacing

\bibitem{CN:tao_2015}
Q.~Wu, M.~Tao, D.~W.~K. Ng, W.~Chen, and R.~Schober, ``{Energy-Efficient
  Transmission for Wireless Powered Multiuser Communication Networks},'' in
  \emph{Proc. IEEE Intern. Commun. Conf.}, Jun. 2015.

\bibitem{JR:WIPT_relaying_timeswitching}
\BIBentryALTinterwordspacing
A.~A. Nasir, X.~Zhou, S.~Durrani, and R.~A. Kennedy, ``{Wireless Energy
  Harvesting and Information Relaying: Adaptive Time-Switching Protocols and
  Throughput Analysis},'' \emph{submitted for possible journal publication},
  Oct. 2013. [Online]. Available: \url{http://arxiv.org/abs/1310.7648}
\BIBentrySTDinterwordspacing

\bibitem{JR:WIPT_CR}
G.~Zheng, Z.~Ho, E.~A. Jorswieck, and B.~Ottersten, ``{Information and Energy
  Cooperation in Cognitive Radio Networks},'' \emph{IEEE Trans. Signal
  Process.}, vol.~62, pp. 2290 -- 2303, Mar. 2014.

\bibitem{Report:Wire_tap}
A.~D. Wyner, ``{The Wire-Tap Channel},'' Tech. Rep., Oct. 1975.

\bibitem{JR:Artifical_Noise1}
S.~Goel and R.~Negi, ``{Guaranteeing Secrecy using Artificial Noise},''
  \emph{IEEE Trans. Wireless Commun.}, vol.~7, pp. 2180 -- 2189, Jun. 2008.

\bibitem{JR:AN_MISO_secrecy}
Q.~Li and W.~K. Ma, ``{Spatically Selective Artificial-Noise Aided Transmit
  Optimization for MISO Multi-Eves Secrecy Rate Maximization},'' \emph{IEEE
  Trans. Signal Process.}, vol.~61, pp. 2704--2717, May 2013.

\bibitem{JR:eff_secure_ofdma_kwan}
D.~W.~K. Ng, E.~S. Lo, and R.~Schober, ``{Energy-Efficient Resource Allocation
  for Secure OFDMA Systems},'' \emph{IEEE Trans. Veh. Technol.}, vol.~61, pp.
  2572 -- 2585, May 2012.

\bibitem{JR:secure_ofdma_DFrelay_kwan}
------, ``{Secure Resource Allocation and Scheduling for OFDMA
  Decode-and-Forward Relay Networks},'' \emph{IEEE Trans. Wireless Commun.},
  vol.~10, pp. 3528 -- 3540, Aug. 2011.

\bibitem{JR:robust_secure_CR_kwan}
\BIBentryALTinterwordspacing
D.~W.~K. Ng, M.~Shaqfeh, R.~Schober, and H.~Alnuweiri, ``{Robust Layered
  Transmission in Secure MISO Multiuser Unicast Cognitive Radio Systems},''
  \emph{submitted for possible journal publication}, Jun. 2014. [Online].
  Available: \url{http://arxiv.org/abs/1406.6542}
\BIBentrySTDinterwordspacing

\bibitem{JR:secure_WIPT_MISO_ruizh}
L.~Liu, R.~Zhang, and K.~C. Chua, ``{Secrecy Wireless Information and Power
  Transfer with MISO Beamforming},'' \emph{IEEE Trans. Signal Process.},
  vol.~62, pp. 1850 -- 1863, Jan. 2014.

\bibitem{JR:Kwan_secure_imperfect}
\BIBentryALTinterwordspacing
D.~W.~K. Ng, E.~S. Lo, and R.~Schober, ``{Robust Beamforming for Secure
  Communication in Systems with Wireless Information and Power Transfer},''
  \emph{accepted for publication, IEEE Trans. Wireless Commun.}, Mar. 2014.
  [Online]. Available: \url{http://arxiv.org/abs/1311.2507}
\BIBentrySTDinterwordspacing

\bibitem{CN:mulobj_secure_WIPT}
D.~W.~K. Ng., L.~Xiang, and R.~Schober, ``{Multi-Objective Beamforming for
  Secure Communication in Systems with Wireless Information and Power
  Transfer},'' in \emph{Proc. IEEE Personal, Indoor and Mobile Radio Commun.
  Sympos.}, Sep. 2013, pp. 8--12.

\bibitem{CN:Kwan_globecom2013}
D.~W.~K. Ng and R.~Schober, ``{Resource Allocation for Secure Communication in
  Systems with Wireless Information and Power Transfer},'' in \emph{Proc. IEEE
  Global Telecommun. Conf.}, Dec. 2013.

\bibitem{CN:kwan_vicky}
S.~Leng, D.~W.~K. Ng, and R.~Schober, ``{Power Efficient and Secure Multiuser
  Communication Systems with Wireless Information and Power Transfer},'' in
  \emph{Proc. IEEE Intern. Commun. Conf.}, Jun. 2014, pp. 800--806.

\bibitem{CN:Multicast_SWIPT}
D.~W.~K. Ng, R.~Schober, and H.~Alnuweiri, ``{Secure Layered Transmission in
  Multicast Systems With Wireless Information and Power Transfer},'' in
  \emph{Proc. IEEE Intern. Commun. Conf.}, Jun. 2014, pp. 5389--5395.

\bibitem{JR:Kwan_SEC_DAS}
\BIBentryALTinterwordspacing
D.~W.~K. Ng and R.~Schober, ``{Secure and Green {SWIPT} in Distributed Antenna
  Networks with Limited Backhaul Capacity},'' \emph{submitted for possible
  journal publication}, Oct. 2014. [Online]. Available:
  \url{http://arxiv.org/abs/1410.3065}
\BIBentrySTDinterwordspacing

\bibitem{CN:PHY_SEC_max_min}
------, ``{Max-Min Fair Wireless Energy Transfer for Secure Multiuser
  Communication Systems},'' in \emph{Proc. IEEE Inf. Theory Workshop}, Nov.
  2014, pp. 326--330.

\bibitem{JR:kwan_MOOP_SWIPT}
\BIBentryALTinterwordspacing
D.~W.~K. Ng, E.~S. Lo, and R.~Schober, ``{Multi-Objective Resource Allocation
  for Secure Communication in Cognitive Radio Networks with Wireless
  Information and Power Transfer},'' \emph{submitted for possible journal
  publication}, Mar. 2014. [Online]. Available:
  \url{http://arxiv.org/abs/1403.0054}
\BIBentrySTDinterwordspacing

\bibitem{JR:MOOP}
R.~T. Marler and J.~S. Arora, ``{Survey of Multi-objective Optimization Methods
  for Engineering},'' \emph{{Structural and Multidisciplinary Optimization}},
  vol.~26, pp. 369--395, Apr. 2004.

\bibitem{book:convex}
S.~Boyd and L.~Vandenberghe, \emph{{Convex Optimization}}.\hskip 1em plus 0.5em
  minus 0.4em\relax {Cambridge University Press}, 2004.

\bibitem{website:CVX}
M.~Grant and S.~Boyd, ``{CVX: Matlab software for disciplined convex
  programming, version 2.0 beta},'' [Online] \url{https://cvxr.com/cvx}, Sep.
  2012.

\bibitem{report:tgn}
{IEEE P802.11 Wireless LANs, ``TGn Channel Models", IEEE 802.11-03/940r4},
  Tech. Rep., May 2004.

\bibitem{JR:linear_fractional}
A.~Charnes and W.~W. Cooper, ``{Programming with Linear Fractional
  Functions},'' \emph{{Naval Res. Logist. Quart.}}, vol.~9, pp. 181--186, Apr.
  1962.

\end{thebibliography}
